\documentclass[reqno]{amsart}
\usepackage[foot]{amsaddr}
\usepackage{graphicx}
\graphicspath{ {./figures/} }
\usepackage[margin=3cm]{geometry}
\usepackage{amsmath, amssymb,amsthm}
\usepackage[scr=rsfs]{mathalpha}
\usepackage[shortlabels]{enumitem}
\usepackage{mlmodern}
\DeclareSymbolFont{largesymbols}{OMX}{cmex}{m}{n} % prevent extra large \sqrt from mlmodern
\usepackage{mathtools} % multlined
\usepackage{stmaryrd} % to get [[ ]] brackets
\usepackage{tikz}
\usetikzlibrary{decorations.pathreplacing,patterns,math,external,decorations.pathmorphing,arrows.meta}
\usepackage{upref} % \ref is upright even in italicized environments
\usepackage[colorlinks]{hyperref}
\hypersetup{citecolor=blue,filecolor=blue,linkcolor=blue,urlcolor=navyblue}
\definecolor{navyblue}{rgb}{0.0, 0.0, 0.5}
\usepackage{CJKutf8}
\DeclareRobustCommand{\rc}[1]{\begin{CJK*}{UTF8}{gbsn}#1\end{CJK*}}

\newtheorem{thm}{Theorem}[section]
\newtheorem{prop}[thm]{Proposition}
\newtheorem{lem}[thm]{Lemma}
\newtheorem{cor}[thm]{Corollary}
\theoremstyle{definition}

\newtheorem{rmk}{Remark}[section]

\newtheorem{ex}{Example}[section]

\newtheorem*{claim*}{Claim}

\newcommand{\N}{\mathbb{N}}
\newcommand{\Z}{\mathbb{Z}}
\newcommand{\R}{\mathbb{R}}
\newcommand{\C}{\mathbb{C}}

\renewcommand{\Im}{\operatorname{\mathrm{Im}}}
\renewcommand{\Re}{\operatorname{\mathrm{Re}}}
\renewcommand{\mod}{\,\operatorname{mod}\,}

\newcommand{\intbrr}[1]{\llbracket#1\rrbracket}

\newcommand{\E}{\mathbf{E}}
\newcommand{\oneb}{\mathbf{1}}
\renewcommand{\P}{\mathbf{P}}

\newcommand{\Tr}{\operatorname{Tr}}
\newcommand{\Var}{\operatorname{Var}}
\newcommand{\dist}{\operatorname{dist}}

\newcommand{\tlog}{\gamma}
\newcommand{\dsim}{\sim}
\newcommand{\vsim}{\sim}

\newcommand{\udm}[2]{U^{\langle \frac{#1}{M}\backslash#2\rangle}}
\newcommand{\udmnom}[2]{U^{\langle#1\backslash#2\rangle}} 
\newcommand{\budm}[2]{\bar U^{\langle \frac{#1}{M}\backslash#2\rangle}}
\newcommand{\splm}[2]{Z_{\langle \frac{#1}{M}\rangle,#2}}
\newcommand{\splmnom}[2]{Z_{\langle #1\rangle,#2}}
\newcommand{\hs}{\mathrm{hs}}
\newcommand{\id}{\operatorname{id}}
\newcommand{\tv}{\mathrm{TV}}
\newcommand{\srw}{\mathrm{srw}}

\newcommand{\ui}[1]{u^{[#1]}}
\newcommand{\bui}[1]{{\bar u^{[#1]}}}
\newcommand{\Bi}[1]{B^{[#1]}}

\newcommand{\bj}[1]{b^{[#1]}}
\newcommand{\Ti}[1]{T^{[#1]}}
\newcommand{\spi}[1]{P^{[#1]}}
\newcommand{\pii}[1]{\pi^{[#1]}}
\DeclareMathAlphabet{\mathcalboondox}{U}{BOONDOX-calo}{m}{n}
\newcommand{\utwo}{\mathcalboondox{u}}

\newcommand{\tmixs}{t_\mathrm{mix}^*}
\newcommand{\tmix}{t_\mathrm{mix}}
\newcommand{\spl}[1]{Z_{#1}}
\newcommand{\pt}[1]{p^{\langle#1\rangle}}
\newcommand{\tmeet}{t_\mathrm{meet}}
\newcommand{\tmeets}{t_\mathrm{meet}^*}
\newcommand{\Pxy}{\P_{xy}}
\newcommand{\fakepair}{\mathcal P_{12}}

\newcommand{\U}{\mathrm{U}}
\newcommand{\Sp}{\mathrm{Sp}}
\newcommand{\On}{\mathrm{O}}
\newcommand{\parentheses}[1]{\left(#1\right)}

\newcommand{\op}[1]{\operatorname{#1}}

\newcommand{\rank}{\op{rank}}
\newcommand{\sym}{\mathrm{sym}}
\newcommand{\ind}{\mathrm{ind}}
\newcommand{\Yt}{\tilde Y}
\newcommand{\Zt}{\tilde Z}

\newcommand\numberthis{\stepcounter{equation}\tag{\theequation}}
\numberwithin{equation}{section}

\begin{document}

\title[]{Depth dependence of entanglement and circuit complexity in random linear optical networks}

\author{\hbox{%
Laura Shou$^{1}$,
Joseph T. Iosue$^{1,2}$,
Yu-Xin Wang (\rc{王语馨})$^2$,
Victor Galitski$^1$,
Alexey V. Gorshkov$^{1,2}$
}\vspace{-\baselineskip}}

\address{\normalfont$^1$Joint Quantum Institute, Department of Physics, NIST/University of Maryland, College Park, MD 20742, USA}
\address{\normalfont$^2$Joint Center for Quantum Information and Computer Science,
NIST/University of Maryland, College Park, MD, 20742, USA}

\begin{abstract}
We study the growth of entanglement and circuit complexity in random passive linear optical networks as a function of the circuit depth. For entanglement dynamics, we start with an initial Gaussian state with all $n$ modes equally squeezed. For random brickwork circuits, we show that entanglement, as measured by the R\'enyi-2 entropy, grows diffusively as a function of the depth, in contrast to the linear growth for random qubit brickwork circuits. In the other direction, for arbitrary circuit geometries we prove bounds on depths which ensure the average-case subsystem entanglement reaches within a constant factor of the maximum value in all subsystems, and bounds which ensure closeness of the random linear optical unitary to a Haar random unitary in $L^2$ Wasserstein distance. We also consider robust circuit complexity for random one-dimensional brickwall circuits, as measured by the minimum number of gates required in any circuit that approximately implements the linear optical unitary. Viewing this as a function of the number of modes and the circuit depth, we show the robust circuit complexity for random one-dimensional brickwall circuits scales at most diffusively in the depth with high probability, again in contrast to the linear scaling for certain random qubit circuits. The corresponding Gaussian unitary $\tilde{\mathcal U}$ for the approximate implementation retains high output fidelity $|\langle\psi|\mathcal U^\dagger \tilde{\mathcal U}|\psi\rangle|^2$ for pure states $|\psi\rangle$ with constrained expected photon-number.
\end{abstract}

\maketitle

\tableofcontents

\section{Introduction}

A prominent class of quantum devices are those that utilize continuous variable resources, thereby allowing one to programmably perform computation or explore quantum phenomena with infinite dimensional (\textit{a.k.a.}~bosonic) Hilbert spaces \cite{braunstein2005quantum-informa,adesso2007entanglement-in,adesso2014continuous-vari}.
Typical physical realization involves optical or photonic devices, which naturally host linear optical (also known as Gaussian) components
\cite{slussarenko2019photonic-quantu,kok2007linear-optical-}.
Despite the promise of continuous variable devices \cite{albert2022bosonic-coding:}, 
much is known about discrete variable quantum devices that is still not understood in continuous variable systems.
One important such example is the dynamics of entanglement and circuit complexity in random circuits.
Random quantum circuits are a crucial component of many tasks that are thought to be efficiently solvable by quantum devices but not by classical devices \cite{hangleiter2023computational-a}.
In finite dimensional (i.e., discrete variable) quantum systems, random circuits have found applications to many areas, including black hole physics \cite{brown2018second-law-of-q,susskind2018black-holes-and,haferkamp2022linear-growth-o} and statistical mechanics \cite{fisher2023random-quantum-}.
Indeed, entanglement and circuit complexity as a function of circuit depth are well-understood in discrete variable systems \cite{nahum2017quantum,haferkamp2022linear-growth-o,fisher2023random-quantum-,chen2024incompressibility}.
In this work, we prove bounds on entanglement dynamics and circuit complexity in random linear optical circuits, thereby partially closing the gap between what is known about discrete and continuous variable systems.

A well-known continuous variable sampling task which is expected to be classically hard is boson sampling \cite{aa}, along with its variant Gaussian boson sampling (GBS) \cite{hamilton2017gaussian,kruse2019detailed}, which has been implemented in several large-scale experiments \cite{zhong2020quantum,zhong2021phase,madsen2022quantum,deng2023gaussian,liu2025robust}. In both protocols, an initial state on $n$ photonic modes is inserted into an $n$-mode passive linear optical network consisting of beamsplitters and phaseshifters, whose overall structure is described by an $n\times n$ unitary matrix $U$. The photons interfere in the linear optical network, and their resulting output state is measured in the photon-number basis. 
Sampling from this output distribution for Haar-random $U$ is argued (assuming some complexity theoretic assumptions) to be classically hard \cite{aa}, relying on classical hardness of exactly calculating \emph{permanents} or \emph{hafnians} of matrices \cite{valiant1979complexity}.

Such complexity theoretic arguments effectively assume a random, infinite-depth linear optical circuit, in which the linear optical unitary $U$ is given by the aforementioned Haar random $n\times n$ unitary matrix.
However, there is much interest in the behavior of linear optical networks as a function of the depth, in particular in understanding the depth dependence of entanglement, complexity, thermalization, and their relationships.
In this paper, we study the first of these questions, and one type of the second, by proving rigorous upper and lower bounds on the growth of entanglement and circuit complexity in random linear optical circuits.
We measure entanglement using the R\'enyi-2 entropy, and starting from a Gaussian input state with all modes squeezed, which is directly relevant for experimental realizations of GBS \cite{madsen2022quantum}.
\begin{itemize}[leftmargin=*]
\item For random circuit elements in an arbitrary geometry, we prove an exact relation between expectation values of $|\langle x|U|y\rangle|^2$ and associated classical random walk probabilities (Theorem~\ref{thm:boson-rw}). 
\item Using this relation, we then consider one-dimensional brickwall circuits (Figure~\ref{fig:bricks}) and prove that, in contrast to random qubit brickwall circuits where entanglement grows ballistically \cite{nahum2017quantum}, the entanglement in random linear optical brickwall circuits grows at most diffusively in the depth (up to a logarithmic factor) with probability $1-1/\text{poly(depth)}$ (Theorem~\ref{thm:up}).
For average-case (i.e. ensemble-averaged) entanglement, we prove that entanglement grows diffusively, $\Theta(\sqrt{d})$ in the depth $d$, up to a saturation time. A similar average-case diffusive growth holds for higher-dimensional brickwork circuits, with the bound involving the boundary area of the subsystem.
\item Additionally, for arbitrary circuit geometries, which allow for all-to-all connectivity, we prove bounds on the depth required to attain maximal order average entanglement, i.e. within a constant factor of the maximum possible value, in all subsystems, as well as on the depth required for $L^2$ Wasserstein closeness of $U$ to Haar random (Theorems~\ref{thm:lower} and \ref{thm:gen}). 
\end{itemize}

Using some of the above results, we also obtain bounds on the circuit complexity for $U$.
Informally, we take the exact circuit complexity of a unitary matrix $U$ to be the minimum number of gates required to construct $U$ in any circuit, allowing for all-to-all gates.
Approximate or robust circuit complexity is then the minimum number of gates required to construct a unitary $\tilde U$ that approximates $U$ in some norm. 
For random qubit circuits, it was shown in \cite{haferkamp2022linear-growth-o} that exact circuit complexity for random brickwall circuits, among a class including other circuit architectures, grows linearly in the number of gates until the saturation time. Linear growth of robust, or approximate, circuit complexity was shown in \cite{chen2024incompressibility}. The linear growth implies such circuits are incompressible, since they cannot be implemented (even approximately) with significantly fewer gates. 

In the continuous variable setting, the notion of circuit complexity is a bit more complicated, since the unitary evolution occurs on infinite-dimensional Hilbert space. We can still consider the circuit complexity of the $n\times n$ linear optical unitary $U$, in the sense of the minimum number of beamsplitter-phaseshifter gates required in any circuit $\tilde U$ which approximates $U$ in Hilbert--Schmidt norm.
Then using \cite{becker2021energy}, this approximation lifts to the infinite-dimensional Hilbert space and implies high fidelity between $\mathcal U|\psi\rangle$ and $\tilde{\mathcal{U}}|\psi\rangle$, where $\mathcal U$ and $\tilde{\mathcal U}$ are the corresponding Gaussian unitaries acting on the (infinite-dimensional) Hilbert space,  for pure states $|\psi\rangle$ with constrained expected photon-number. 
We consider random one-dimensional brickwall linear optical circuits, and consider the robust circuit complexity as a function of the number of modes $n$ and the depth $d$.
As with entanglement, we find different behavior in circuit complexity growth as a function of the depth $d$ compared to random qubit circuits:

\begin{itemize}[leftmargin=*]
\item Random one-dimensional brickwall linear optical circuits \emph{are} compressible, in the sense that the robust circuit complexity for the linear optical unitary $U$ scales at most diffusively in the depth $d$ with high probability (Theorem~\ref{thm:complexity}), up to a logarithmic factor. 
In particular, a depth-$d$ random one-dimensional brickwall linear optical unitary $U$, which uses $\Theta(nd)$ beamsplitter-phaseshifter gates, can with high probability be implemented approximately by a linear optical network with only $O(n\sqrt{d\log d})$ gates.
\end{itemize}

\subsection{Set-up}
We briefly recall some background on linear optical networks. For further details, we refer to \cite{Serafini2017Quantum-Continu,Fukuda2019Typical-entangl}.
The Hilbert space for $n$ bosonic modes can be taken to be $L^2(\R)^{\otimes n}$, with the $i$th mode's Hilbert space $L^2(\R)$ having a basis of photon number states $|0\rangle_i,|1\rangle_i,|2\rangle_i,\ldots$, and creation and annihilation operators $a_i^\dagger$ and $a_i$.
We consider passive linear optical networks on $n$ modes, consisting of $2$-mode beamsplitters and phaseshifters arranged in some circuit geometry. Such a network can be described by an $n\times n$ linear optical unitary $U$ which is defined by its action
\begin{align}
a_i^\dagger\mapsto\sum_{j=1}^nU_{ji}a_j^\dagger.
\end{align}
Beamsplitters are then represented as $2\times2$ unitary blocks $\begin{bmatrix}\cos\theta&-e^{i\phi}\sin\theta\\e^{-i\phi}\sin\theta&\cos\theta\end{bmatrix}\in\mathrm{SU}(2)$ acting on the desired modes, and phaseshifters as $1\times 1$ phase shifts $[e^{i\alpha}]$. We will typically use beamsplitter-phaseshifter combinations, consisting of a beamsplitter with phaseshifters applied at the input and output modes; these are represented by arbitrary $2\times2$ unitary matrices in $\U(2)$.

Our main linear optical network example will be the brickwall circuit illustrated in Figure~\ref{fig:bricks}, although in some cases we will also consider arbitrary circuit geometries, which allow for all-to-all connectivity. In the circuit drawings, each $2$-mode gate depicts a $2\times2$ unitary matrix representing a beamsplitter-phaseshifter combination, and each $1$-mode gate depicts a $1\times1$ unitary matrix representing a phaseshifter. In the one-dimensional brickwall circuit drawn in Figure~\ref{fig:bricks}, a single ``step'' $U^{(i)}$ in the circuit consists of two adjacent layers of beamsplitters and phaseshifters, and the number of steps is the depth $d$ of the circuit. The matrix structure of a single brickwall step $U^{(i)}$ is defined below and shown in Figure~\ref{fig:blocks}.

\begin{figure}[htb]
\begin{tikzpicture}
\foreach \mode in {1,...,8}{
\draw (.5,-\mode/2) --++ (2.6,0);
\draw (3.75,-\mode/2) --++(1.4,0);
\ifnum\mode < 4
\node[left] at (.55,-\mode/2) {$\mode$};
\fi
}
\node[left] at (.55,-4) {$n$};
\node[left] at (.55,-2) {$\vdots$};
\begin{scope}[yshift=-.3cm,xshift=-.2cm]
\foreach \layer in {1,2,4}{
\begin{scope}[xshift=\layer*.05cm]
\foreach \lbrick in {0,...,3}{ % make left layer U(2) blocks
	\draw[fill=gray!30] (\layer,-\lbrick) rectangle (\layer+.4,-\lbrick-.9);
}
\foreach \rbrick in {0,...,2}{ % make right layer U(2) blocks
	\draw[fill=gray!30,yshift=-.5cm,xshift=.5cm] (\layer,-\rbrick) rectangle (\layer+.4,-\rbrick-.9);
}
% R layer top and bottom U(1) blocks
\draw[fill=gray!30,xshift=.5cm] (\layer,0) rectangle (\layer+.4,-.425);
\draw[fill=gray!30,xshift=.5cm] (\layer,-3.5) rectangle (\layer+.4,-3.5-.425);
% label layers
\draw[decoration={brace,raise=3pt,aspect=.5,amplitude=4pt},decorate] (\layer,0)--(\layer+.9,0);
\ifnum\layer < 4
    \node [above] at (\layer+.5,.25) {$U^{(\layer)}$};
\else
    \node [above] at (\layer+.5,.25) {$U^{(d)}$};
\fi
\end{scope}
}
\node[below] at (1.3,-4) {$U^L$};
\node[below] at (1.8,-4) {$U^R$};
\end{scope}
\node at (3.5,-2.25) {$\cdots$};
\node[rotate=90] at (0,-2) {$n$ input modes};
\end{tikzpicture}
\caption{A one-dimensional depth $d$ brickwall circuit on modes $\{1,\ldots,n\}$, $n\in2\N$, is constructed as the product of $d$ steps: $U=\prod_{i=d}^1U^{(i)}=U^{(d)}\cdots U^{(2)}U^{(1)}$. Each $2$-mode gate corresponds to a $2\times2$ unitary matrix representing a beamsplitter with phaseshifters, and each $1$-mode gate corresponds to a $1\times1$ unitary matrix representing a phase-shifter. 
Note that, in the circuit, $U^{(1)}$ is applied first, then $U^{(2)}$, and so on, so $U^{(1)}$ appears in the right-most position in the product for $U$.
}\label{fig:bricks}
\end{figure}

\begin{figure}[htb]
\begin{tikzpicture}
\def\psize{1.7cm}
\def\msize{1.6cm}
%structure inside
\def\w{.8cm} % block size
\begin{scope}[xshift=3.6cm]
\foreach \b in {0,1,3}{
	\draw[fill=gray!30,xshift=-\msize+\b*\w,yshift=\msize-\b*\w] (0,0)--(\w,0)--(\w,-\w)--(0,-\w)--cycle;
	\node at (-\msize+\b*\w+.4cm,\msize-\b*\w-.4cm) {$U(2)$};
}
\node at (.1,0) [below right] {$\ddots$};
% left and right parenthesis
\node at (-\psize,0) {$\left(\vphantom{\rule{\psize}{\psize}}\right.$};
\node at (\psize,0) {$\left.\vphantom{\rule{\psize}{\psize}}\right)$};
\draw[decoration={brace,raise=3pt,aspect=.5,amplitude=4pt,mirror},decorate] (-\psize,-\psize)--(\psize,-\psize);
\node[below] at (0,-\psize-2mm) {$U^L$};
\end{scope}
\begin{scope}[xshift=0cm] % right matrix
\def\w{.8cm} % block size
\foreach \b in {0.5,2.5}{ %2x2 blocks
	\draw[fill=gray!30,xshift=-\msize+\b*\w,yshift=\msize-\b*\w] (0,0)--(\w,0)--(\w,-\w)--(0,-\w)--cycle;
	\node at (-\msize+\b*\w+.4cm,\msize-\b*\w-.4cm) {$U(2)$};
}
\foreach \b in {0,7}{ % small 1x1 blocks
	\draw[fill=gray!30,xshift=-\msize+\b*\w/2,yshift=\msize-\b*\w/2] (0,0)--(\w/2,0)--(\w/2,-\w/2)--(0,-\w/2)--cycle;
	\node[below,xshift=-\msize+\b*\w/2+.2cm,yshift=\msize-\b*\w/2] at (\w/4cm,\w/4cm) {\scriptsize$U(1)$};
}
\node at (0,0) {$\ddots$};
% left and right parenthesis
\node at (-\psize,0) {$\left(\vphantom{\rule{\psize}{\psize}}\right.$};
\node at (\psize,0) {$\left.\vphantom{\rule{\psize}{\psize}}\right)$};
\draw[decoration={brace,raise=3pt,aspect=.5,amplitude=4pt,mirror},decorate] (-\psize,-\psize)--(\psize,-\psize);
\node[below] at (0,-\psize-2mm) {$U^R$};
\end{scope}
\node at (-2.5,0) {\Large $U^{(i)}=$};
\end{tikzpicture}
\caption{Matrix structure of $U^{(i)}$. A single step $U^{(i)}$ is the product of two mismatching block matrices $U^{R,(i)}=\utwo_0\oplus\left(\bigoplus_{j=1}^{n/2-1}\utwo_j\right)\oplus \utwo_{0}'$ and $U^{L,(i)}=\bigoplus_{j=1}^{n/2}\utwo_j$, where all $\utwo_j$ are $2\times2$ unitary blocks except for $\utwo_0$ and $\utwo_0'$ which are $1\times1$ unitary.
The values of the blocks $\utwo_j^{(i)}$ can differ between different steps $U^{(i)}$. When considering random brickwall circuits, we will always take all $2\times2$ and $1\times1$ blocks to be independent and Haar distributed, and in this case we may formally write $U^{(i)}\dsim U^RU^L$.
}\label{fig:blocks}
\end{figure}

We can formally define a one-dimensional brickwall unitary matrix $U$ as follows.
Let $n\in2\N$. 
Given \( \utwo_1,\dots, \utwo_{n-1} \in \U(2) \), and $\utwo_0,\utwo_0'\in\U(1)$, construct the $n\times n$ matrices
\begin{align}
    U^L(\utwo_1,\dots, \utwo_{n/2}) &= \bigoplus_{j=1}^{n/2} \utwo_j,\\
    U^R(\utwo_1,\dots, \utwo_{\frac{n}{2}-1},\utwo_0,\utwo_0') &= \utwo_0 \oplus \parentheses{\bigoplus_{j=1}^{\frac{n}{2}-1} \utwo_j} \oplus \utwo_0',\\
    f(\utwo_1,\dots, \utwo_{n-1},\utwo_0,\utwo_0') &=U^R(\utwo_{\frac{n}{2}+1}, \dots, \utwo_{n-1},\utwo_0,\utwo_0') U^L(\utwo_1,\dots,\utwo_{n/2}) .\label{eqn:uprod}
\end{align}
Let \( d\in\N \).
For \( j=1,\dots,n-1 \) and \( d'=1,\dots, d \), choose \( \utwo_j^{(d')}\in\U(2) \) and \(\utwo_0^{(d')},\utwo_0'^{(d')}\in\U(1)\) and set
\begin{equation}\label{eqn:u-brickwall}
    U = f(\utwo_1^{(d)},\dots, \utwo_{n-1}^{(d)},\utwo_0^{(d)},\utwo_0'^{(d)})f(\utwo_1^{(d-1)},\dots, \utwo_{n-1}^{(d-1)},\utwo_0^{(d-1)},\utwo_0'^{(d-1)}) \cdots f(\utwo_1^{(1)},\dots, \utwo_{n-1}^{(1)},\utwo_0^{(1)},\utwo_0'^{(1)}).
\end{equation}
For notational convenience, we will usually write \eqref{eqn:u-brickwall} as $U=\prod_{i=d}^1U^{(i)}\equiv U^{(d)}U^{(d-1)}\cdots U^{(1)}$.
The matrix $U$ is a depth-$d$ brickwall unitary matrix, composed of steps $U^{(d')}=f(\utwo_1^{(d')},\dots, \utwo_{n-1}^{(d')},\utwo_0^{(d')},\utwo_0'^{(d')})$, for $d'=1,\ldots,d$.
We will mostly consider the \emph{random} brickwall circuit where all $1\times1$ and $2\times2$ blocks \( \utwo_j^{(d')} \) are drawn independently from the appropriate Haar measure. 
In this case, we will write $U^{(i)}\dsim U^RU^L$.

\subsection{Main results}

For entanglement dynamics, we consider the Gaussian boson sampling protocol which starts with an initial Gaussian state with all modes squeezed with some squeezing parameter $s>0$, and then evolves the state according to a linear optical unitary $U$. For a subsystem $\Gamma\subseteq\{1,\ldots,n\}$, the resulting R\'enyi-2 entropy is $S_2=-\log\Tr\rho^2$, where $\rho$ is the density matrix for the reduced (evolved) state corresponding to $\Gamma$.
The R\'enyi-2 entropy can also be written as \cite{Serafini2017Quantum-Continu}
\begin{align}\label{eqn:s2-cov}
S_2(U)&=\frac{1}{2}\log\det\sigma(U)=\frac{1}{2}\Tr\log\sigma(U),
\end{align}
where $\sigma(U)$ is the covariance matrix of $\rho$.
Since the state on $n$ modes is a pure state, the R\'enyi-2 entropy for the subsystem $\Gamma$ is the same as for $\Gamma^c$ \cite{nielsen2010quantum}. Thus, without loss of generality, we will take $|\Gamma|\le n/2$.

We are interested in the behavior of the R\'enyi-2 entropy $S_2(U)$ as a function of the circuit depth $d$.
We first prove bounds on its growth in the one-dimensional brickwall circuit.

\begin{thm}[brickwall upper bounds]\label{thm:up}
Let $U=\prod_{i=d}^1U^{(i)}$ be a depth-$d$ one-dimensional brickwall circuit as defined in \eqref{eqn:u-brickwall} or Figure~\ref{fig:bricks}.
Consider the R\'enyi-2 entropy $S_2(U)$ for the subsystem consisting of the first $k$ modes, i.e. $\Gamma=\{1,2,\ldots,k\}$.
\begin{enumerate}[(i)]
\item (Worst-case bound) For any $d\in\N$, 
\begin{align}\label{eqn:worst-bound}
S_2(U)\le 4d\log\cosh(2s).
\end{align}
\item (Average-case bound) Let all the $2\times2$ and $1\times1$ blocks in $U$ be independent and drawn from Haar measure. Then letting $\E$ denote the average over these $U$,
\begin{align}\label{eqn:exp-bound}
\E S_2(U)\le O(\sqrt{d})\sinh^2(2s).
\end{align}
Additionally, for any $\kappa\in\N$, there is a constant $C_\kappa$ so that 
\begin{align}\label{eqn:prob-bound}
\P\left[S_2(U)\le C_\kappa\sinh^2(2s)\sqrt{d\log d}\right]\ge 1-O(d^{-\kappa}).
\end{align} 
\end{enumerate}
\end{thm}

As we will see in Appendix~\ref{sec:brick-lower}, there is a matching-order average-case lower bound, so that for some $c$ and $d\le ck^2$, we have
\begin{align}
\E S_2(U)=\Theta_{s}(\sqrt{d}),
\end{align}
where $\Theta_{s}$ indicates implicit constants may depend on $s$.
We also have similar bounds for $D$-dimensional brickwork circuits through Corollary~\ref{cor:upperbound2} and Appendix~\ref{sec:brick-lower}, giving (for appropriate depths $d$) $\E S_2(U)=\Theta_{D,s}(A_\Gamma\sqrt{d})$, where $A_\Gamma$ denotes the subsystem boundary area for higher dimensions $D$ defined in Corollary~\ref{cor:upperbound2}.
We see entanglement in random brickwork bosonic circuits grows diffusively, in contrast to the ballistic entanglement growth in random brickwork qubit circuits \cite{nahum2017quantum}.

Since the expected number of photons in the system is $n\sinh^2(s)$, we see the bounds in Theorem~\ref{thm:up} are very much a random walk bound in the sense of Figure~\ref{fig:paths} below. 

\begin{rmk}
\begin{enumerate}
\item While the upper bounds are formulated for any depth $d$, the entropy $S_2(U)$ of a size $k$ subsystem is always bounded above by $k\log\cosh(2s)$ in this setup, as can be seen from \eqref{eqn:s2-tanh}.
Thus in practice one should only consider these bounds for depths $d$ before the upper bounds reach the trivial upper bound.

\item Note that the upper bounds do not depend on $k$ or $n$, because they only depend on the size of the ``boundary'', i.e. the modes in $\Gamma$ that can talk to the outside modes at time $d$ (or in the average-case bound, the modes that can talk to the outside modes with large enough probability at time $d$). This isn't affected by $k$ (or $n$) until the boundary reaches close to the entire subsystem size $k$, at which point the trivial bound in (1) starts to take over.

\item While the coefficient 4 in \eqref{eqn:worst-bound} is likely not optimal, we do have at least worst-case linear-in-$d$ behavior,
\begin{align}\label{eqn:sup}
\sup_{\substack{U\in\U(n):\text{depth-$d$}\\\text{brickwall}}}S_2(U)\ge d\log\cosh(2s),
\end{align}
for $d\le k \;(\le n/2)$. 
To see this, note that one can use \cite{clements2016optimal} to construct any $2d\times 2d$ unitary matrix in a brickwall circuit of depth $d$ (which has $2d$ layers), so we can construct the worst-case unitary from \cite[Appendix B]{iosue2023page} in depth $d$ and place it on the middle $2d$ modes around the subsystem cut. This gives the lower bound in \eqref{eqn:sup}. Additionally, we expect a brickwall network of fixed 50-50 complex beamsplitters to generate ballistic entanglement growth.

\item The probability bound \eqref{eqn:prob-bound} is not very useful for small depth $d$. However, as it numerically appears the variance of $S_2(U)$ may be roughly constant even at small depth $d$ (cf. Figure~\ref{fig:numerics}), we should perhaps not expect a very good typicality bound unless $d\to\infty$. 

\item The proofs or results of Theorem~\ref{thm:up}, as well as those of Theorems~\ref{thm:lower} and \ref{thm:gen}, can also apply to modified settings. For example, the same results (up to a constant factor) hold for the brickwall circuit with periodic boundary conditions.
Using \cite[Eq.~(B11)]{youm2025average}, one can also extend these results to R\'enyi-$\alpha$ entropies for $\alpha$ an integer $\ge2$.
\end{enumerate}
\end{rmk}

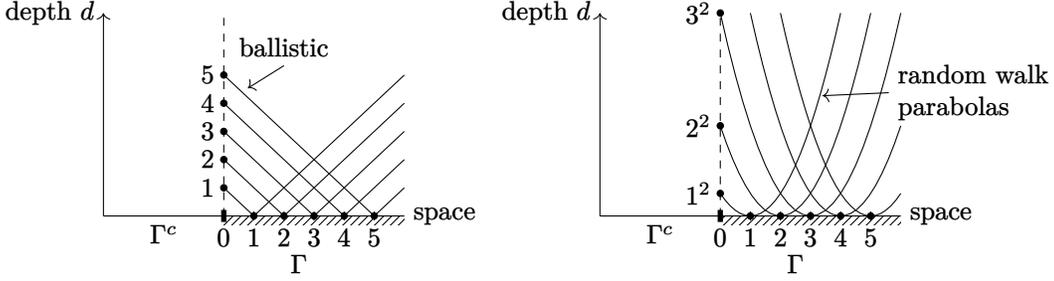
\begin{figure}[htb]
\begin{tikzpicture}[scale=.4]
\def\pscale{1.25}
\def\slopescale{.75} % make smaller for shallower lines
\def\height{.75*9/1.25} % put \pscale value in denominator, \psacle*\height from parabola in numerator
\draw (-4,0)--(6,0) node[right] {space};
\foreach \kmode in {1,...,5}{
\draw (\kmode,.15)--++(0,-.3) node[below] {$\kmode$};
\draw[scale =1, domain = 0:\kmode]  plot ({\x},{-\slopescale*\pscale*(\x-\kmode)});
\draw[scale =1, domain = \kmode:6]  plot ({\x},{\slopescale*\pscale*(\x-\kmode)});
\draw[fill=black] (0,\slopescale*\pscale*\kmode) circle (3pt) node[left] {$\kmode$};       
\draw[fill=black] (\kmode,0) circle (3pt);               
}
\draw[line width=2] (0,.2)--++(0,-.4);
\node[below] at (0,-.15) {$0$};
\draw[->] (-4,0)--++(0,\height*\pscale) node[left] {depth $d$};
\draw[dashed] (0,0)--++(0,\height*\pscale);
\node[below] at (-2,0) {$\Gamma^c$};
\node[below] at (2.5,-1) {$\Gamma$};
\fill[pattern = north east lines] (0,0) rectangle (6,-.3);
\draw[<-] (.8,3.4*\pscale)--(2,5) node[above] {ballistic};
\end{tikzpicture}
\begin{tikzpicture}[scale=.4]
\def\height{9}
\def\pscale{.75}
\draw (-4,0)--(6,0) node[right] {space};
\foreach \kmode in {1,...,5}{
\draw (\kmode,.15)--++(0,-.3) node[below] {$\kmode$};
\ifnum\kmode<4
\draw[scale =1, domain = 0:{sqrt(\height)+\kmode}]  plot ({\x},{\pscale*(\x-\kmode)*(\x-\kmode)});
\draw[fill=black] (0,\pscale*\kmode*\kmode) circle (3pt) node[left] {$\kmode^2$}; 
\else
\draw[scale =1, domain = {\kmode-sqrt(\height)}:6]  plot ({\x},{\pscale*(\x-\kmode)*(\x-\kmode)});
\fi 
\draw[fill=black] (\kmode,0) circle (3pt);
}
\draw[line width=2] (0,.2)--++(0,-.4);
\node[below] at (0,-.15) {$0$};
\draw[->] (-4,0)--++(0,\height*\pscale) node[left] {depth $d$};
\draw[dashed] (0,0)--++(0,\height*\pscale);
\node[below] at (-2,0) {$\Gamma^c$};
\node[below] at (2.5,-1) {$\Gamma$};
\fill[pattern = north east lines] (0,0) rectangle (6,-.3);
\draw[<-] (3.4,4)--++(2.2,.1) node[right] {\parbox[t]{2cm}{random walk parabolas}};
\end{tikzpicture}
\caption{Depiction of how single photons split into superpositions. The two halves of the superposition propagate in opposite directions, ballistically (left, worst-case) or diffusively like in a random walk (right, average-case). In the diffusive case, after depth $d$ only $\sim O(\sqrt{d})$ trajectories have crossed the cut between subsystems.}\label{fig:paths}
\end{figure}

The proof of Theorem~\ref{thm:up}(i) is via a light-cone argument, which can also be applied to arbitrary circuit geometries (Remark~\ref{rmk:worst}). 
For the proof of Theorem~\ref{thm:up}(ii), we will relate $U$ to a classical random walk on an interval, and use this to bound the R\'enyi-2 entropy in terms of classical random walk probabilities.
More generally, we prove an exact relation between expectation values of entries of $U$ for random circuits with arbitrary geometry, and associated classical random walks (Theorems~\ref{thm:boson-rw} and \ref{thm:uut}).

Before moving to large depth lower bounds for arbitrary circuit geometries, we note that Theorem~\ref{thm:up} can be adapted for $D$-dimensional brickwork circuits as follows. 
\begin{cor}[$D$-dimensional brickwork upper bounds]\label{cor:upperbound2}
Consider the $D$-dimensional brickwork circuit on $n=m^D$ modes $\{1,\ldots,m\}^D$ for $m\in2\N$, for which a single step (consisting of $2D$ layers) applies the 1D brickwall pattern sequentially in each dimension; see Remark~\ref{rmk:worst} for a more precise definition.
Consider the subsystem $\Gamma=\{(x_1,\ldots,x_D):1\le x_j\le k_j\}$ for some choice of $k_1,\ldots,k_D\in\intbrr{1:m}$, and define $A_\Gamma:=\sum_{j=1}^D\oneb_{k_j<m}\prod_{\substack{\ell=1\\\ell\ne j}}^Dk_\ell$ to be the ``boundary area'' of $\Gamma$. For example, if all $k_j=k<m$, then $A_\Gamma=Dk^{D-1}$.
 \begin{enumerate}[(i)]
\item (Worst-case bound) For any $d\in\N$,
\begin{align}
S_2(U)&\le 4d A_\Gamma\log\cosh(2s).
\end{align}
\item (Average-case bound) Let all the $2\times2$ and $1\times1$ blocks in the circuit be independent and drawn from Haar measure. Then
\begin{align}\label{eqn:dim-upper}
\E S_2(U)&\le O(\sqrt{d})A_\Gamma\sinh^2(2s).
\end{align}
\end{enumerate}
\end{cor}

In the opposite direction of Theorem~\ref{thm:up}, we can identify depths at which average subsystem entanglement becomes of maximal order, and prove a bound on the depth which ensures $U$ is $\varepsilon$-close to fully Haar random in $L^2$ Wasserstein distance (with Hilbert--Schmidt norm). (The $L^2$ Wasserstein distance between two probability measures $\mu$ and $\nu$ on $\U(n)$ equipped with the Hilbert--Schmidt norm is $W_2(\mu,\nu)=\inf_{X\dsim\mu,Y\dsim\nu}(\E\|X-Y\|_\hs^2)^{1/2}$.)
\begin{thm}[brickwall long-time lower bound]\label{thm:lower}
Consider the random one-dimensional brickwall circuit $U=\prod_{i=d}^1U^{(i)}$ where all $2\times2$ and $1\times1$ unitary blocks are drawn independently from Haar measure. Then there is a constant $C_1$ so that for depth $d\ge C_1n^2\log^2 n$ and any size $k\le n/2$ subsystem,
\begin{align}\label{eqn:s2-lower}
\E S_2(U)\ge\Theta(k)\log\cosh(2s).
\end{align} 
Furthermore, $U$ is $\varepsilon$-close to Haar random in $L^2$ Wasserstein distance at depth $d\ge C_2n^3(\log n)^2\log(n/\varepsilon)$. The $L^2$ Wasserstein convergence also implies convergence of $\E S_2(U)$ to Haar random values and weak typicality\footnote{this will be convergence in probability $S_2(U)/\E_{\mathcal U\dsim\mathrm{Haar}(n)} S_2(\mathcal U)\xrightarrow{p}1$; see Corollary~\ref{cor:haar}} of $S_2(U)$ at these depths.
\end{thm}

Note the bound $d\ge C_1n^2\log^2n$ matches with the diffusive upper bound \eqref{eqn:exp-bound} of Theorem~\ref{thm:up}(ii), up to logarithmic order terms.

An analogue of Theorem~\ref{thm:lower} holds for more general beamsplitter geometries, with a depth requirement depending on properties of an associated classical walk which is defined in 
Section~\ref{subsec:bosonrw}. In particular, the depth requirement depends on a classical mixing time $\tmixs$ (when the classical walk is close to the uniform equilibrium distribution) and meeting time $\tmeets$ (when two independent walks have nearly met with high probability). A brief statement is provided here, with the full definitions and statement given in Section~\ref{sec:lower} and Theorem~\ref{thm:lowerbound}.
\begin{thm}[general long-time lower bounds]\label{thm:gen}
Suppose we allow $2$-mode beamsplitters which may connect any two modes, and we construct a geometry and depth $d$ unitary matrix $U=\prod_{i=d}^1U^{(i)}$ with independent $2\times2$ and $1\times1$ Haar unitary blocks (described in more detail in Sections~\ref{subsec:bosonrw} and \ref{sec:lower}).
Then for depths larger than $\tmixs+\tmeets$ of the associated classical random walk, for any size $k\le n/2$ subsystem,
\begin{align}\label{eqn:s2-lowerbound}
\E S_2(U)\ge\Theta(k)\log\cosh(2s).
\end{align} 
Additionally, $\varepsilon$-closeness to Haar measure in $L^2$ Wasserstein distance holds with an additional $Cn\log(n/\varepsilon)$ factor on the depth. 
\end{thm}

For circuit geometries corresponding to graphs with bounded degree $M$, the condition for \eqref{eqn:s2-lowerbound} can be strengthened to involve only the spectral gap (related to the mixing time); see Appendix~\ref{sec:lower2}.

Finally, we use some intermediate results from the above to prove that the robust or approximate circuit complexity for one-dimensional random brickwall linear optical networks scales at most diffusively in the depth $d$ with high probability. This is again in contrast to the linear scaling for random qubit circuits \cite{chen2024incompressibility}.
In the other direction, we prove a depth condition which ensures that with high probability, a near-maximal-order number of gates are required to approximate the linear optical unitary $U$. Due to \cite{reck1994experimental,clements2016optimal}, any $n\times n$ unitary matrix can be exactly implemented using $n(n-1)/2$ beamsplitter-phaseshifter gates, and so $n^2$ is the maximum order for circuit complexity.
In what follows, we will use asymptotic notation $O_\kappa(\cdot)$ to indicate that the implicit constant may depend on $\kappa$.

\begin{thm}[approximate circuit complexity]\label{thm:complexity}
Consider depth-$d$ linear optical unitaries $U$ constructed from independent $2\times2$ and $1\times1$ Haar-random beamsplitter-phaseshifter gates.
\begin{enumerate}[(i),leftmargin=*]
\item (Brickwall upper bound) Let $U$ be an $n\times n$ depth-$d$ random one-dimensional brickwall linear optical unitary (Figure~\ref{fig:bricks}); this utilizes $\Theta(nd)$ nearest-neighbor gates.
Choose a $\kappa\ge1$. 
With probability $1-O(n^2d^{-\kappa})$ over $U$,
there is a unitary $\tilde U$ built from $O_\kappa(n\sqrt{d\log d})$  nearest-neighbor gates such that
\begin{align}\label{eqn:u2}
\|U-\tilde U\|_\hs&=O(n^{3/2}d^{-\kappa}),
\end{align}
where $\|\cdot\|_\hs$ denotes the Hilbert--Schmidt norm.
In particular, the probability this occurs is $1-o(1)$ and \eqref{eqn:u2} is $o(1)$ for $d=\omega(n^{2/\kappa})$, which is a mild condition for large $\kappa$.

The bound \eqref{eqn:u2} also implies the following: Let $\mathcal U$ and $\tilde{\mathcal U}$ denote the Gaussian unitaries acting on the (bosonic) Hilbert space which correspond to $U$ and $\tilde U$ respectively, and let $\mathscr{U}$ and $\tilde{\mathscr U}$ be the associated unitary channels, defined as conjugation by $\mathcal U$ or $\tilde{\mathcal U}$. Letting $N=\sum_{j=1}^na_j^\dagger a_j$ denote the total photon number Hamiltonian, then in the energy-constrained (EC) diamond norm \eqref{eqn:ec-diamond}, for $E>0$,
\begin{align}\label{eqn:diamond}
\|\mathscr{U}-\tilde{\mathscr U}\|_\diamond^{N,E}&=\sqrt{E+1}\,O(n^{5/4}d^{-\kappa/2}),
\end{align}
and the fidelity between $\mathcal U|\psi\rangle$ and $\tilde{\mathcal U}|\psi\rangle$ for pure states $|\psi\rangle$ with constrained expected photon-number satisfies
\begin{align}\label{eqn:fidelity}
\inf_{\langle\psi|N|\psi\rangle\le E}|\langle\psi|{\mathcal U}^\dagger\tilde{\mathcal U}|\psi\rangle|^2&\ge 1-(E+1)O(n^{5/2}d^{-\kappa}).
\end{align}
For $d=\omega(n^{5/(2\kappa)}(E+1)^{1/\kappa})$, both big-O error terms above are $o(1)$.

\item (Saturation bound) Let $0<\delta'<1/2$, and consider $n\times n$ depth-$d$ random linear optical unitaries $U$ corresponding to an arbitrary circuit geometry as described in Section~\ref{subsec:bosonrw}. 
Let $\tmeets$ and $\tmixs$ be certain mixing and meeting times of the associated classical random walk, as used in Theorem~\ref{thm:gen} and defined in Section~\ref{sec:lower}.
Then there are constants $C_1,C_2,c_1,c_2>0$ so that for $U$ of depth
\begin{align}\label{eqn:complexity-depth}
d\ge C_1(\tmeets+\tmixs)n^3\log n,
\end{align}
and any $\delta<(1/2-\delta')\sqrt{n}$,
\begin{align}\label{eqn:complex2}
\P[\exists \tilde U\text{ with $c_1n^2/\log n$ gates s.t. }\|\tilde U-U\|_\hs\le\delta]&\le C_2e^{-c_2n^2}.
\end{align}
The gates forming $\tilde U$ are allowed to be any 2-mode all-to-all unitary gates.
\end{enumerate}
\end{thm}

Part (i) gives a diffusive upper bound on the circuit complexity; it would be of much interest to obtain a matching diffusive lower bound.
In (ii), for the one-dimensional brickwall circuit, \eqref{eqn:complexity-depth} becomes the condition $d\ge Cn^5\log^3 n$, since we can take $\tmeets$ and $\tmixs$ to be $O(n^2\log^2n)$ (see Remark~\ref{rmk:lower}). 
However, we expect this depth condition could be improved, and that saturation should occur near depth of order $n^2$ to match the diffusive upper bound in (i). 
Note also that in \eqref{eqn:complex2}, since any unitary can be formed with $O(n^2)$ gates, the division by $\log n$ in the gate count $c_1n^2/\log n$ is key for obtaining the result.
The particular approximate implementations constructed in Section~\ref{sec:complexity} for the proof of part (i) have sequential components and end up having large depth;
however, the unitary $\tilde U$ can also be taken to have shallower depth $O_\kappa(\sqrt{d\log d})$; see Appendix~\ref{sec:rcc}.

\subsection{Previous work}

It is known from \cite{zhang2021entanglement} that bosons starting from a \emph{single} squeezed mode undergo random walk behavior on average.
In this case, \cite{zhang2021entanglement} showed that the system can effectively be viewed as a two-mode process, and so the entanglement entropy can be written exactly in terms of the symplectic eigenvalue of a single $2\times 2$ matrix. This eigenvalue can be expressed in terms of the associated classical walk transition probability between the subsystems, which implies a relation between the classical walk mixing time and the time required for maximal entanglement entropy starting from a single squeezed mode.
However, since experimental implementations of Gaussian boson sampling typically involve all or a large number of initial squeezed modes \cite{zhong2020quantum,zhong2021phase,madsen2022quantum,deng2023gaussian,liu2025robust}, we are interested in entanglement in the case where all initial modes are squeezed with some squeezing parameter $s>0$.
The method used in \cite{zhang2021entanglement} cannot be applied for more than one squeezed mode,\footnote{except in a sparse limit where all the photons in the subsystem are concentrated on a small number of modes} and moreover we do not expect there to be simple analytic expressions in depth-$d$ circuits for the entanglement or R\'enyi-$\alpha$ entropies in this case. 

Numerically, depth dependence of entanglement for linear optical systems has been studied for various set-ups in several works \cite{zhuang2019scrambling,zhou2021nonunitary,zhang2021entanglement,go2024exploring}.
In the infinite-depth Haar random limit, exact formulas for the entanglement entropy and R\'enyi-$\alpha$ entropies were obtained analytically in \cite{iosue2023page,youm2025average}, using the asymptotic Weingarten calculus on Taylor expansions of the R\'enyi-$\alpha$ entropies.

\subsection{Proof overviews}
We briefly describe some of the ideas in the proofs.
The starting point for the results on the random linear optical networks is the boson random walk property in Theorem~\ref{thm:boson-rw}, which gives an exact relation between second absolute moments of entries of the matrix $U$, and probabilities for a random walk $Z_t$ which performs a classical walk through the gates of the circuit,
\begin{align}\label{eqn:rw0}
\E|\langle x|U|y\rangle|^2=\P_y[Z_d=x].
\end{align}
The walk $Z_t$, while not possessing independent increments, can be seen for the 1D brickwall geometry to satisfy a central limit theorem and large deviation estimates similar to simple random walk. In particular, it behaves diffusively, and is unlikely to be found at a distance significantly larger than order $\sqrt{d}$ at time $d$.
While the R\'enyi-2 entropy involves fourth moments of $U$, the large deviation estimates for the second moments of $U$ in \eqref{eqn:rw0} will be sufficient to obtain the diffusive upper bound in Theorem~\ref{thm:up}(ii), using a combination of estimates and a light-cone-like argument (Section~\ref{sec:upperbounds}). 

The proof of the lower bound in Theorem~\ref{thm:gen} is more technically involved. 
We would like to calculate certain fourth moments of $U$ by iteratively pulling off unitary factors $U^{(i)}$ from $U=\prod_{i=d}^1U^{(i)}$. However, each iteration generates additional terms and factors which vary depending on how close two endpoint-conditioned classical independent random walks are to each other. 
When the walks are far away from each other, the factors are random walk transition probabilities. When the walks meet or are close to meeting, the factors change, making it difficult to keep track of the total factors over all possible pairs of endpoint-conditioned random walks. 
Due to this, we consider depths larger than a meeting $+$ mixing time of the associated classical walk $Z_t$ from \eqref{eqn:rw0}. We iteratively pull off unitary factors $U^{(i)}$ until the two associated classical walks meet (which occurs by the meeting time with high probability), at which point we stop the iteration. If the depth remaining is at least the mixing time, then we can use the mixing to effectively ignore the endpoint conditioning of the walks. This will be enough to provide a lower bound on the desired fourth moments of $U$, which is the main technical result of Section~\ref{sec:lower}, written as the random variable decoupling Lemma~\ref{lem:dec}. This will imply a lower bound on the R\'enyi-2 entropy at these depths.
The $L^2$ Wasserstein closeness will use the decoupling in Lemma~\ref{lem:dec} as a starting point, followed by a coupling construction applied with the coupling method of \cite{oliveira2009convergence}.

The proof of Theorem~\ref{thm:complexity} on approximate circuit complexity also relies on the relation \eqref{eqn:rw0}. For part (i), due to the large deviation bounds for the classical 1D brickwall walk $Z_t$, the entries of $U$ outside an effective band width $O(\sqrt{d\log d})$ are very small with high probability. 
We show these small entries can essentially be ignored in an approximate circuit, and that the remaining $O(n\sqrt{d\log d})$ matrix entries within the effective band can be implemented using only $O(n\sqrt{d\log d})$ gates using the zeroing procedure in \cite{hurwitz1897uber,reck1994experimental}, without significantly changing the Hilbert--Schmidt norm from $U$. The lift from approximate implementation of the linear optical unitary $U$ to results on the (infinite-dimensional) Hilbert space follows from results in \cite{becker2021energy}.
The proof of the long-time near-saturation in Theorem~\ref{thm:complexity}(ii) uses the Wasserstein closeness proved in Theorem~\ref{thm:gen} to compare the circuit complexity to that of a Haar random unitary.

\subsection{Outlook}\label{subsec:outlook}

In this paper, we prove Theorems~\ref{thm:up}--\ref{thm:complexity} on entanglement and circuit complexity growth in random linear optical networks as a function of the circuit depth. For random brickwork circuits, we show that average-case R\'enyi-2 entropy (starting from an initial Gaussian state with all modes equally squeezed) scales diffusively in the depth, in contrast to the linear scaling for their qubit counterparts. Additionally, for random 1D brickwork circuits, we show that approximate circuit complexity scales at most diffusively (up to a logarithmic factor) in the depth with high probability, again in contrast to the linear scaling for qubits. For arbitrary circuit geometries, we prove bounds on depths which guarantee (i) average subsystem entanglement reaches within a constant factor of its maximum value, (ii) $\varepsilon$-closeness of the linear optical unitary $U$ to Haar random in $L^2$ Wasserstein distance, or (iii) high robust circuit complexity with high probability.

We discuss some possible further directions. In this paper, we study entanglement starting from a Gaussian input state with all modes equally squeezed.
Subsequent work \cite{shou2026anticoncentration} can be used to extend to unequal squeezing regimes.
It would however be interesting to study the entanglement behavior for other types of input states, such as Fock input states or other non-Gaussian states, as well as to Gaussian states with some initial vacuum modes. 

Regarding the variance and fluctuations of the R\'enyi-2 entropy $S_2(U)$, numerical evidence suggests that for equal-squeezing input states and independent Haar random beamsplitter-phaseshifters in the brickwall geometry, that $S_2(U)\sim C\sqrt{d}+\xi$, for random fluctuations $\xi$ which are of roughly constant order regardless of the depth (Figure~\ref{fig:numerics}). This is in contrast to the KPZ behavior observed in random brickwall qubit circuits \cite{nahum2017quantum}, where the fluctuations are observed to grow as the KPZ scaling $\xi\propto d^{1/3}$ before the saturation time. 
Additionally, \cite{zhuang2019scrambling} considers random active elements in continuous-variable systems and finds linear growth of fluctuations, which they speculate could be due to fluctuations from the random squeezing overriding any KPZ fluctuations. 
It would be interesting to rigorously study the entanglement fluctuations for both passive linear optical networks and various models with active Gaussian optics. 

\begin{figure}[htb]
\includegraphics[height=2in]{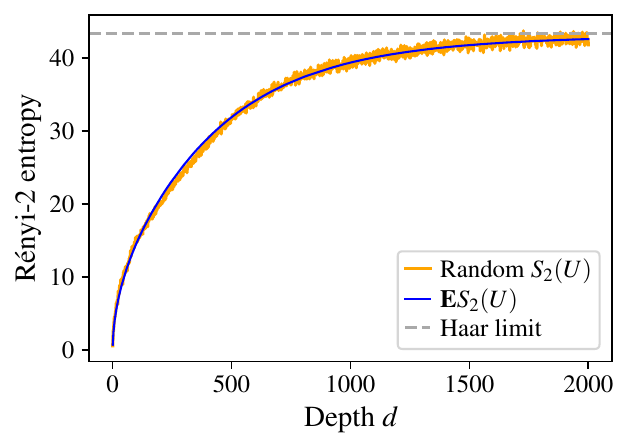}
\includegraphics[height=2in]{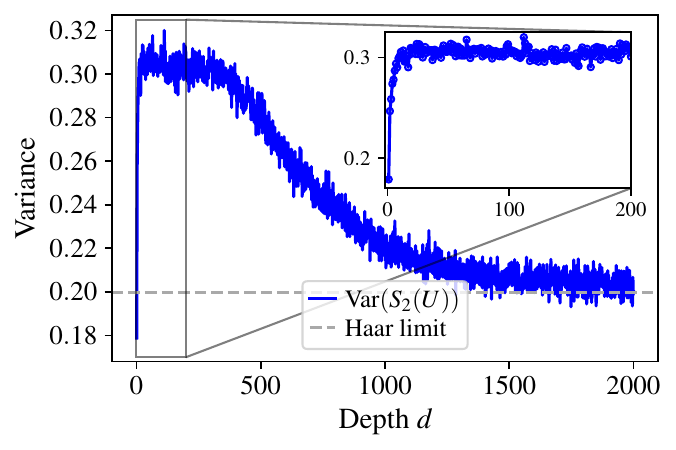}
\caption{(Left) Growth of the R\'enyi-2 entropy for the random brickwall geometry as a function of the depth, for $n=100$, $k=50$, $s=1$, shown for a single random realization (orange), and averaged over 5000 trials (blue). The early time shows a square-root diffusive growth. (Right) Variance of the R\'enyi-2 entropy for the setting in the left figure. The variance reaches a constant value nearly immediately (in contrast to the KPZ behavior of random brickwall \emph{qubit} circuits), then eventually saturates to the Haar-random limiting value.}\label{fig:numerics}
\end{figure}

While we consider some aspects of circuit complexity in this paper, a more detailed study would be of interest. One main question is obtaining lower bounds at low depths. We also mention the question of \emph{sampling} complexity. Despite the slower, diffusive entanglement growth in random brickwall linear optical networks compared to random brickwall qubit circuits, we still conjecture that classical hardness of sampling for 1D random brickwall may occur by any algebraic depth $d\gtrsim n^\gamma$, any $\gamma>0$. For worst-case hardness, one can embed an $n^\gamma\times n^\gamma$ Gaussian boson sampling instance in a depth-$O(n^\gamma)$ brickwall circuit using \cite{clements2016optimal}; however, average-case sampling hardness remains unknown.

\subsection{Outline and notation}

The rest of the paper is organized as follows.
\begin{itemize}
\item Section~\ref{sec:upperbounds}: We prove Theorem~\ref{thm:up} on upper bounds for the R\'enyi-2 entropy, assuming Proposition~\ref{prop:expbound} which is proved in the following section. 

\item Section~\ref{sec:rw}: We prove a boson random walk theorem, Theorem~\ref{thm:boson-rw}, which relates expectation values of $U$ to classical random walk probabilities. We then prove Proposition~\ref{prop:expbound}, which is used to finish the proof of the upper bound Theorem~\ref{thm:up}(ii). We also give the proof of Corollary~\ref{cor:upperbound2}(ii).

\item Section~\ref{sec:lower}: We prove Theorems~\ref{thm:lower} and \ref{thm:gen} on lower bounds for the R\'enyi-2 entropy.

\item Section~\ref{sec:complexity}: We prove Theorem~\ref{thm:complexity} on circuit complexity.

\item Appendix~\ref{appendix:sym}: Proof of a random walk bound. Appendix~\ref{sec:concentration}: Lipschitz constant estimate for the R\'enyi-2 entropy, concentration of Haar measure, and an application of concentration of measure to R\'enyi-2 typicality in the Haar random case.

\item Appendix~\ref{sec:lower2}: We prove the entanglement long-time lower bound Theorem~\ref{thm:lower-bound}.
Appendix~\ref{sec:brick-lower}: We prove the brickwork entanglement lower bound Theorem~\ref{thm:brickwall-lowerbound}.

\item Appendix~\ref{sec:rcc}: We show how to adapt the proof of Theorem~\ref{thm:complexity} to construct an approximate unitary $\tilde U$ with shallower depth $O_\kappa(\sqrt{d\log d})$.

\end{itemize}

For convenience, we collect some notation defined or used in various places in the paper:
\begin{itemize}

\item $\fakepair(n)$: permutations in $S_n$ which are a product of disjoint transpositions. These can be naturally associated with partitions of $\{1,\ldots,n\}$ into size 1 or 2 sets, which we may call a ``pairing'' (where self-pairing is allowed). Such a pairing corresponds to a single layer of beamsplitters and phaseshifters, where a linear optical element is applied across each pair.
For example, in the case of the 1D brickwall circuit in Figure~\ref{fig:bricks}, the first layer $U^L$ corresponds to $\pi_1=(1\,2)(3\,4)\cdots(n-1\,n)\in S_n$ (written in cycle notation) and the partition $\{ \{ 1,2\}, \{3,4 \}, \dots, \{n-1,n\} \}$, and the second layer $U^R$ to $\pi_2=(1)(2\,3)(4\,5)\cdots(n-2\,n-1)(n)$ and the partition $\{ \{ 1\}, \{2,3 \}, \{4,5 \},\dots, \{n-2,n-1 \}, \{ n\} \}$.

\item $\ui{j}$: $n\times n$ matrix which represents a single layer of beamsplitters with phaseshifters applied in parallel according to $\pi_j\in\fakepair(n)$.

\item $U^{(i)}$: $n\times n$ matrix for a single step (which consists of $M$ layers) of the circuit, $U^{(i)}\vsim\prod_{j=M}^1\ui{j}=\ui{M}\cdots\ui{1}$. For example, $M=2$ for brickwall since a single step $U^{(i)}$ consists of two layers $U^L$ and $U^R$.

\item $U$: $n\times n$ matrix for the depth $d$ circuit, formed as $U=\prod_{i=1}^dU^{(i)}=U^{(d)}\cdots U^{(1)}$.

\item $U^L$ and $U^R$: These describe the structure of the matrices $\ui{1}$ and $\ui{2}$ for the brickwall circuit; see also Figure~\ref{fig:blocks}. If we write $U^{(i)}\dsim U^RU^L$, we mean that we take all the blocks in $U^R$ and $U^L$ to be independent and Haar distributed.

\item $X\dsim\mu$: random variable $X$ has law given by $\mu$. $X\vsim Y$: abuse of notation, but we use it to mean that random variables $X$ and $Y$ have the same distribution
\item $\|\cdot\|_\hs$: Hilbert--Schmidt norm; for an $n\times n$ matrix $A$, $\|A\|_\hs=\left(\sum_{i,j=1}^n|A_{ij}|^2\right)^{1/2}$.

\item $\intbrr{a:b}=\{a,a+1,\ldots,b-1,b\}$ for $a,b\in\Z$, $a\le b$
\item $\N=\{0,1,2,\ldots\}$
\item We make use of the usual asymptotic notations $o,O,\Theta,\Omega,\omega$. We may also use $\tilde O$ or $\tilde \Theta$ to indicate there may be hidden logarithmic factors.
\end{itemize}

Throughout the paper, generic constants $C,C',c,c',c'',\ldots$ can change from line to line.
We generally make no attempt to optimize constants.

\section{Upper bounds on entanglement}\label{sec:upperbounds}
In this section, we prove the upper bounds in Theorem~\ref{thm:up}(i) and (ii), assuming Proposition~\ref{prop:expbound} which we will prove in the next section.

\subsection{Preliminaries}
We briefly review some useful expressions for the R\'enyi-2 entropy for Gaussian states in this setting; for further details see e.g. \cite{Serafini2017Quantum-Continu,iosue2023page}.
The initial Gaussian state with all modes squeezed with equal squeezing parameters $s_i=s>0$ is $|\psi\rangle=\bigotimes_{i=1}^ne^{\frac12s((a_i^\dagger)^2-a_i^2)}|0\rangle_i$, and it has covariance matrix $\sigma_0=e^{2s}I_n\oplus e^{-2s}I_n$.
After applying the linear optical unitary $U$, the covariance matrix is $\tilde\sigma(U)=\eta(U)\sigma_0\eta(U)^T$, where $\eta(U)=\begin{pmatrix}\Re U&\Im U\\-\Im U&\Re U\end{pmatrix}$ is the usual isomorphism between $\U(n)$ and $\Sp(2n)\cap\On(2n)$ (corresponding to passive Gaussian transformations), and where $\Sp(2n)$ is the real symplectic group of $2n\times2n$ matrices and $\On(2n)$ the orthogonal group of $2n\times2n$ matrices.
For a subsystem $\Gamma\subset\{1,\ldots,n\}$, the reduced state has covariance matrix $\sigma(U)=\hat P\tilde\sigma(U)\hat P^T$, where $\hat P=P\oplus P$ for $P:\C^n\to\C^\Gamma$ the projection onto $\Gamma$. Thus the covariance matrix of the reduced evolved state is
\cite{Fukuda2019Typical-entangl}
\begin{align}
\sigma(U)=\cosh(2s)I_{\Gamma\oplus\Gamma}+\sinh(2s)\begin{pmatrix}P\Re(\bar U\bar U^T)P^T&P\Im(\bar U\bar U^T)P^T\\
P\Im(\bar U\bar U^T)P^T&-P\Re(\bar U\bar U^T)P^T
\end{pmatrix}.
\end{align}
Letting $W=\Pi UU^T\Pi \bar U\bar U^T\Pi$ with $\Pi$ the $n\times n$ projection matrix onto $\Gamma$, \cite[(A32),(A33)]{iosue2023page} uses a power series expansion of $S_2(U)=\frac{1}{2}\Tr\log\sigma(U)$ to give the expressions
\begin{align}\label{eqn:s2-tanh}
S_2(U)&=|\Gamma|\log\cosh(2s)+\frac{1}{2}\Tr\log(1-\tanh^2(2s)W)\\
&=\sum_{\ell=1}^\infty\frac{\tanh^{2\ell}(2s)}{2\ell}\left(|\Gamma|-\Tr W^\ell\right).\label{eqn:s2}
\end{align}
These expressions will be used in the upper and lower bounds on entanglement.

\subsection{Light cone bound \texorpdfstring{$O(d)$}{O(d)}} 

In this section we will prove Theorem~\ref{thm:up}(i).
We will apply \eqref{eqn:s2-tanh}.
Let $\Pi:\C^n\to\C^n$ be the rank $k$ projector onto the first $k$ modes, and let $\tilde V=\Pi UU^T\Pi$ and $W=\tilde V\tilde V^\dagger$. 
The projectors $\Pi$ in $\tilde V$ effectively cut the matrix $UU^T$ at the subsystem boundaries. 
Note that we can write $\tilde V=V\oplus\mathbf{0}_{n-k}$ where $V=P UU^T P^T$ is $k\times k$. 
If the resulting two subsystems were completely non-interacting, i.e.~$U=U_{k\times k}\oplus U_{n-k\times n-k}$, then $V=U_{k\times k}U_{k\times k}^T$ would be unitary and $W=I_k\oplus\mathbf{0}_{n-k}$, so that \eqref{eqn:s2-tanh} implies $S_2(U)=0$. We see that \eqref{eqn:s2-tanh} and \eqref{eqn:s2} thus measure how ``non-unitary'' $V$ is after applying the projections on $UU^T$. For small depth circuits, which should have small entanglement, we want to show that $V$ is only a small perturbation away from a unitary matrix.

\begin{proof}[Proof of Theorem~\ref{thm:up}(i)]
We will allow slightly more general matrix structures of the steps $U^{(i)}$; instead of requiring a brickwall step we will just assume they have bounded band width.
Let $U=\prod_{i=d}^1U^{(i)}$ be a 1D depth $d$ circuit with each $U^{(i)}$ having band width $\le w$ for some fixed $w\in\N$, i.e.~$U^{(i)}_{xy}=0$ if $|x-y|>w$. For the brickwall circuit, we have $w=2$. The matrix representation of the depth $d$ matrix $U$ is then a band matrix with band width at most $wd$. The matrix $UU^T$ is also a band matrix with band width at most $2wd$. This permits at most $4wd+1$ nonzero entries in a row. When we form $V=P UU^TP^T$, the projectors cut off part of the band of $UU^T$ in the bottom right corner (Figure~\ref{fig:band}).
\begin{figure}[htb]
\begin{tikzpicture}
\def\psize{1.5cm}
\def\msize{1.35cm}
%structure inside
\def\w{.4cm} % band width * 2
\draw[fill=gray!20] (-\msize,\msize-\w)--(-\msize,\msize)--(-\msize+\w,\msize)--(\msize,-\msize+\w)--(\msize,-\msize)--(\msize-\w,-\msize)--cycle;
\draw[<->,xshift=.6cm,yshift=-.6cm] (-\w,0)--(\w,0);
\node[above] at (.5,-.65) {$4wd+1$};
% red triangle
\fill[red] (0,0)--(\w,0)--(0,\w)--cycle;
\draw[red,<-] (.3,.3) -- (.7,.7) node[above] {cut};
% box for V
\draw[line width=2] (-1.4,0)--(0,0)--(0,1.4)--(-1.4,1.4)--cycle;
\draw[dashed] (0,\w)--(-1.5,\w);
\draw[<-] (-.5,1.5)--(-.2,1.7) node[right] {$V$};
\draw (0,-1.5)--(0,0)--(1.5,0);
\draw[decoration={brace,raise=3pt,aspect=.5,amplitude=4pt},decorate,xshift=-1mm] (-1.5,\w)--(-1.5,1.4);
\node[left] at (-1.8,1) {orthonormal};
\draw[decoration={brace,raise=3pt,aspect=.5,amplitude=4pt},decorate,xshift=-1mm] (-1.5,0)--(-1.5,{\w-.5mm});
\node[left] at (-1.8,.2) {$2wd$}; 
\node[right] at (1.7,0.1) {\large $=UU^T$};
% left and right parenthesis
\node at (-\psize,0) {$\left(\vphantom{\rule{\psize}{\psize}}\right.$};
\node at (\psize,0) {$\left.\vphantom{\rule{\psize}{\psize}}\right)$};
\end{tikzpicture}
\caption{Forming $\tilde V$ or $V$ from the banded matrix $UU^T$ requires cutting out some entries from the last $O(d)$ rows in the subsystem, making those rows no longer orthonormal.}\label{fig:band}
\end{figure}
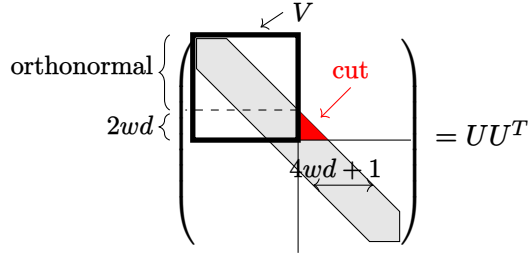
This cut only affects the last $2wd$ rows of $UU^T$ before the cut. (If $2wd\ge k$, the final bound $O(d)$ is trivial, so we will assume $2wd<k$.)
The top $k-2wd$ rows of $V$ are still orthonormal, and also orthogonal to the bottom $2wd$ rows of $V$ since the top rows don't see the cut part (Figure~\ref{fig:band}).
Therefore
\begin{align}\label{eqn:vvs}
VV^\dagger=I_{k-2wd}\oplus RR^\dagger,
\end{align} 
for a $2wd\times 2wd$ subunitary matrix $RR^\dagger$ corresponding to the bottom $2wd$ rows of $V$.

Applying \eqref{eqn:s2-tanh} and \eqref{eqn:vvs}, we then have
\begin{align*}
S_2(U)&=k\log\cosh(2s)+\frac{1}{2}\Tr\log(1-\tanh^2(2s)W)\\
&=k\log\cosh(2s)+\frac{1}{2}\Tr I_{k-2wd}\log(1-\tanh^2(2s)) + \frac{1}{2} \Tr\log(1-\tanh^2(2s)RR^\dagger)\\
&=2wd\log\cosh(2s)+\frac{1}{2}\Tr\log(1-\tanh^2(2s)RR^\dagger)\\
&\le 2wd\log\cosh(2s)-0.\numberthis
\end{align*}
For the brickwall circuit this gives $S_2(U)\le 4d\log\cosh(2s)$. 
\end{proof}

\begin{rmk}\label{rmk:worst}
The above worst-case light cone argument generalizes to other beamsplitter geometries and subsystems $\Gamma$ as follows. 
Let $\partial^-_{UU^T}\Gamma:=\{x\in\Gamma:\exists y\not\in\Gamma,\,\langle x|UU^T|y\rangle\ne0\}$. 
Recall $W=\Pi_\Gamma UU^T\Pi_\Gamma \bar U\bar U^T\Pi_\Gamma$ and that $V=P_\Gamma UU^T P_\Gamma$ which is $|\Gamma|\times|\Gamma|$, so that $W\cong VV^\dagger\oplus \mathbf{0}_{n-|\Gamma|}$. We claim that for any $x\in\Gamma\setminus \partial^-_{UU^T}\Gamma$, the $x$th row $\langle x|V$ of $V$ has unit norm and is orthogonal to every other row of $V$. 
This follows since $\langle x|UU^T|y\rangle=0$ for all $y\not\in\Gamma$, so that for $z\in\Gamma$,
\begin{align}
\nonumber \langle x|VV^\dagger|z\rangle&=\sum_{y\in\Gamma}\langle x|UU^T|y\rangle\langle y|(UU^T)^\dagger|z\rangle=\langle x|z\rangle.
\end{align}
Therefore, up to a permutation of the coordinate basis,
\begin{align}
\nonumber VV^\dagger&=I_{\Gamma\setminus\partial^-_{UU^T}\Gamma}\oplus RR^\dagger,
\end{align}
where $RR^\dagger$ is a $|\partial^-_{UU^T}|\times|\partial^-_{UU^T}|$ subunitary matrix corresponding to the rows in $\partial^-_{UU^T}$. As in the 1D band matrix case, by \eqref{eqn:s2-tanh} we obtain
\begin{align}\label{eqn:s2-bdy}
S_2(U)&\le|\partial^-_{UU^T}\Gamma|\log\cosh(2s),
\end{align}
for any geometry of $U$.

As a specific example, we can easily obtain a worst-case upper bound for the entanglement growth for a $D$-dimensional brickwork circuit as stated in Corollary~\ref{cor:upperbound2}(i).
For $m\in2\N$, let the $n=m^D$ modes be $\{(x_1,\ldots,x_D):x_j\in\intbrr{1:m}\}$, and let $U^{L,j}$ and $U^{R,j}$ be single bricklayer unitaries acting on the $j$th coordinate dimension, defined as follows: The matrix $U^{L,j}$ puts a $2\times 2$ unitary matrix across coordinates corresponding to modes $x=(x_1,\ldots,x_D)$ and $y=(y_1,\ldots,y_D)$ if $x_i=y_i$ for all $i\ne j$, and $\{x_j,y_j\}\in\pi^L$, where $\pi^L$ corresponds to the left brickwall permutation $(1\,2)\cdots(m-1\,m)\in\fakepair(m)$. Similarly, $U^{R,j}$ is defined in the same way but with $\{x_j,y_j\}\in\pi^R$, where $\pi^R$ corresponds to the right brickwall permutation $(1)(2\,3)\cdots(m-2\,m-1)(m)\in\fakepair(m)$. 
A single step $U^{(i)}$ will be the product of the $2D$ matrices $U^{L,1},\cdots, U^{L,D},U^{R,1},\cdots, U^{R,D}$, in any (fixed) order we choose, with different unitary blocks allowed for different $i$.
The depth $d$ $D$-dimensional brickwork unitary is then $U=\prod_{i=d}^1 U^{(i)}$. Note that a single step $U^{(i)}$ consists of $2D$ single layers, so the depth $d$ unitary $U$ consists of $2Dd$ single layers.

Consider the subsystem $\Gamma=\{(x_1,\ldots,x_D):1\le x_j\le k_j\}$ for some choice of $k_1,\ldots,k_D$, and define $A_\Gamma:=\sum_{j=1}^D\oneb_{k_j<m}\prod_{\substack{\ell=1\\\ell\ne j}}^D k_\ell$ to be the ``boundary area'' of $\Gamma$. We want to estimate $|\partial^-_{UU^T}\Gamma|$.
Note that for each coordinate $j$, $UU^T$  contains only $4d$ single layer matrices $U^{L,j}$, $U^{R,j}$, $(U^{L,j})^T$, $(U^{R,j})^T$, one set for each $U^{(i)}$, which act on this particular $j$th coordinate. If we view a nonzero $\langle x|U^{L/R,j}|y\rangle\ne0$ as allowing movement from mode $x$ to mode $y$ (and the same for the transposes $(U^{L/R,j})^T$), then each of these matrices only allows movement at most distance 1 in the $j$th coordinate. Then by expanding $UU^T$ in terms of its $4Dd$ factors, we see that for $\langle x|UU^T|y\rangle\ne0$ we must have $|x_j-y_j|\le 4d$ for each $j=1,\ldots,D$.
Therefore, letting $\dist_{\ell^\infty}$ denote the $\ell^\infty$ distance,
\begin{align*}
|\partial^-_{UU^T}\Gamma|\le |\{x\in\Gamma:\dist_{\ell^\infty}(x,\Gamma^c)\le 4d\}|
&\le \sum_{j=1}^D|\{x\in\Gamma:x_j> k_j-4d\}|\mathbf{1}_{k_j<m}\\
&\le \sum_{j=1}^D4d\mathbf{1}_{k_j<m}\prod_{\substack{\ell=1\\\ell\ne j}}^Dk_\ell = 4d A_\Gamma.\numberthis\label{eqn:Dbdy}
\end{align*}
Inserting this in \eqref{eqn:s2-bdy} gives 
\begin{align}
S_2(U)&\le 4dA_\Gamma\log\cosh(2s), \label{eqn:s2d-worst}
\end{align}
for the depth $d$, $D$-dimensional brickwork circuit. If all $k_j=k<m$, this becomes $4dDk^{D-1}\log\cosh(2s)$ since the boundary area of $\Gamma$ is $Dk^{D-1}$.
\end{rmk}

\subsection{Average case bound \texorpdfstring{$O(\sqrt{d})$}{O(sqrt(d))}} \label{sec:avgcase}

In this section, we prove Theorem~\ref{thm:up}(ii), assuming Proposition~\ref{prop:expbound} below, which we will prove in the next section.

When all of the $2\times2$ and $1\times1$ blocks in the brickwall matrix $U=\prod_{i=d}^1U^{(i)}$ are independent and Haar random, Theorem~\ref{thm:up}(ii) gives a slower growth of the R\'enyi-2 entropy than in the worst-case bound.
As we will show, this is because the entries $|\langle x|UU^T|y\rangle|^2$ decay rapidly with high probability when $|x-y|\gg \sqrt{d}$ (Figure~\ref{fig:band-numerics}). So even though $UU^T$ has a worst-case band width $\Theta(d)$, it effectively only has a band width approximately $\Theta(\sqrt{d})$, or perhaps just barely $\omega(\sqrt{d})$.
When we cut off the corners in the band of $UU^T$ to form $V$ as in Figure~\ref{fig:band}, we end up cutting out much smaller entries, many of which are exponentially small, than in the worst case bound. 
The light cone/cut argument then suggests we should be able to replace the R\'enyi-2 bound $O(d)$ with the effective band-width $O(\sqrt{d})$.

\begin{figure}[htb]
\includegraphics[height=2.7in]{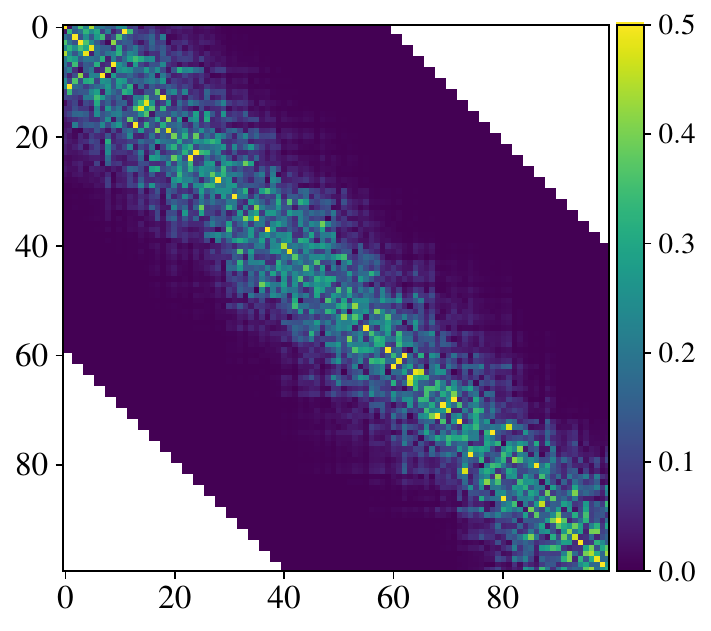}
\includegraphics[height=2.7in]{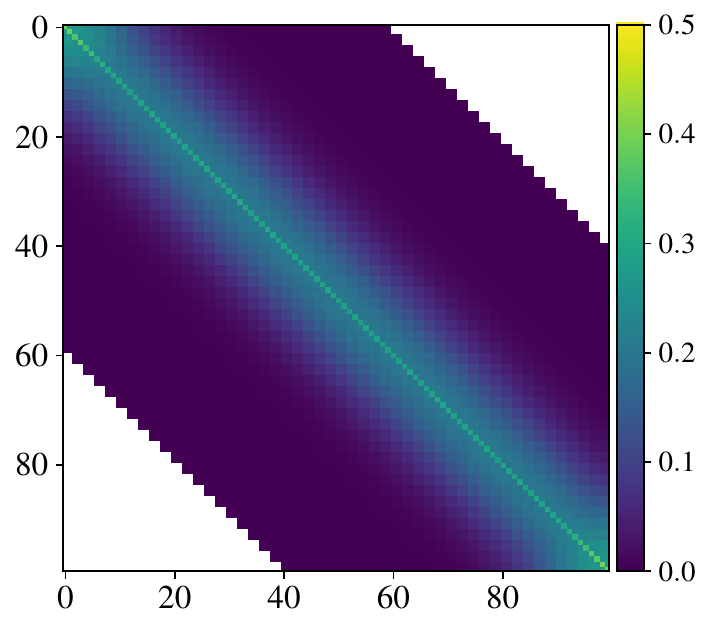}
\caption{
Plots of the matrix entries of $|UU^T|$ for $n=100$ and depth $d=15$, for a single realization (left) and averaged over $1000$ realizations (right). The color scale is truncated at $0.5$ (all entries $>0.5$ are plotted at the same color as $0.5$).
Entries near the edges of the band are much smaller than those near the center. This creates a smaller ``effective'' band width $\approx\Theta(\sqrt{d})$ compared to the actual band width $\Theta(d)$. Entries outside the actual $\Theta(d)$ band which are exactly zero are shown in white.
}\label{fig:band-numerics}
\end{figure}

The main estimate we need in order to prove Theorem~\ref{thm:up}(ii) is
\begin{prop}\label{prop:expbound}
Let $U=\prod_{i=d}^1 U^{(i)}$ be a random depth $d$ brickwall circuit as described above. 
\begin{enumerate}[(i)]
\item There exist numerical constants $C,c>0$ so that for all $x,y\in\intbrr{1:n}$, 
\begin{align}\label{eqn:uut-exp}
\E|\langle x|U|y\rangle|^2\le Ce^{-c(x-y)^2/d},\qquad
\E|\langle x|UU^T|y\rangle|^2\le Cde^{-c(x-y)^2/d}, 
\end{align}
and
\begin{align}\label{eqn:uut-prob}
\P\left[|\langle x|UU^T|y\rangle|^2>d^2e^{-2c(x-y)^2/d}\right] &\le Cde^{-c(x-y)^2/d}. 
\end{align}

\item For $\Gamma=\{1,\ldots,k\}$, we can also obtain the bound
\begin{align}\label{eqn:sum-exp}
\sum_{x\in\Gamma}\sum_{y\not\in\Gamma}\E|\langle x|UU^T|y\rangle|^2&\le C\sqrt{d}.
\end{align}
\end{enumerate}
\end{prop}

We will prove Proposition~\ref{prop:expbound} in Section~\ref{sec:rw} by relating $\E|\langle x|UU^T|y\rangle|^2$ and $|\langle x|UU^T|y\rangle|^2$ to classical random walk probabilities. Note that the quantity $e^{-c(x-y)^2/d}$ is reminiscent of the exponential decaying term in the probability for a simple random walk on $\Z$ to go from $x$ to $y$ in $d$ steps. 
For the brickwall circuits, we will have a random walk on $\{1,\ldots,n\}$ with edges between nearby sites determined by the brickwall structure, and which has similar exponential decay.

Assuming Proposition~\ref{prop:expbound} for now, we can prove Theorem~\ref{thm:up}(ii). The similar bound in Corollary~\ref{cor:upperbound2}(ii) for higher dimensional brickwork circuits will be proved in Section~\ref{subsec:higherdim}.

\begin{proof}[Proof of Theorem~\ref{thm:up}(ii)]
We plan to estimate $\Tr W^\ell$ and use \eqref{eqn:s2}.
Recall $\Tr W=\Tr VV^\dagger=\|V\|_\hs^2$ and that $V=P U U^TP^T$. Considering the top $k\times n$ part of $UU^T$ (Figure~\ref{fig:topk}) and using that the rows are normalized and that $UU^T$ is a banded matrix with band width at most $4d$, we see that for any $\gamma\in\intbrr{0:k}$, 
\begin{align}
\nonumber\|V\|_\hs^2&= \sum_{x=1}^k\sum_{y=1}^k|\langle x|V|y\rangle|^2= \sum_{x=1}^k \bigg(1-\sum_{y=k+1}^n|\langle x|UU^T|y\rangle|^2\bigg)\\
&\ge k-\gamma-\sum_{x=\max(1,k-4d)}^{k-\gamma}\sum_{y=k+1}^{\min(n,k+4d)}|\langle x|UU^T|y\rangle|^2,\label{eqn:VF}
\end{align}
where we are now only summing over the matrix entries of $UU^T$ pictured in gray in Figure~\ref{fig:topk}.
Note that the bound holds for all $\gamma\in\intbrr{0:k}$ with equality when $\gamma=0$.
The min and max bounds are inserted because for small $4d\ll n$, many of the entries in the gray region are actually zero. These min and max bounds imply there are at most $O(d^2)$ nonzero terms in the sum.
\begin{figure}[htb]
\begin{tikzpicture}
\def\rt{5cm};
\def\gh{.5cm};
\draw (0,0)--(\rt,0)--(\rt,2)--(0,2)--cycle;
\draw (0,0)--(2,0)--(2,2)--(0,2)--cycle;
\draw (0,\gh)--(\rt,\gh);
\node at (1,1) {{\Large $V$}};
\node[left] at (0,0) {$x=k$};
\node[left] at (0,\gh) {$x=k-\gamma$};
\node[left] at (0,1.9) {$x=1$};
\node[above] at (0.2,2) {$y=1$};
\node[above] at (2,2) {$y=k$};
\node[above] at (\rt,2) {$y=n$};
\draw[fill=gray!60] (2,\gh)--(\rt,\gh)--(\rt,2)--(2,2)--cycle;
\node[right] at (\rt+.2cm,1) {$=$ top $k$ rows of $UU^T$};
\end{tikzpicture}
\caption{Top $k$ rows of $UU^T$. Due to \eqref{eqn:VF}, we only need to upper bound the entries of $UU^T$ in the gray shaded region.}\label{fig:topk}
\end{figure}
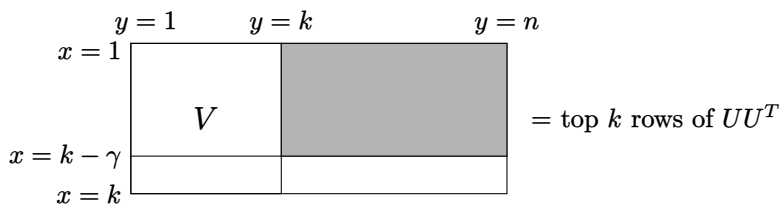

We first prove the average case bound \eqref{eqn:exp-bound}. 
In this case, we can take $\tlog=0$ and apply Proposition~\ref{prop:expbound}(ii) to obtain
\begin{align*}
\E \Tr W=\E\|V\|_\hs^2&\ge k-O(\sqrt{d}).\numberthis\label{eqn:trW}
\end{align*}

To handle higher traces of $W$, let $0\le \lambda_i\le 1$ ($i \in\intbrr{1:k} $) be the $k$ eigenvalues of $VV^\dagger$, so that
\begin{align}\label{eqn:trl}
\Tr W^\ell=\Tr((VV^\dagger)^\ell)&=\sum_{i=1}^k\lambda_i^\ell.
\end{align}
We then use
\begin{lem}\label{lem:moments}
Let $x_i$, $i=1,\ldots,k$ be random variables taking values in $[0,1]$. 
If $\E\sum_{i=1}^kx_i\ge k - \varepsilon$, then for any $\ell\ge1$,
\begin{align}\label{eqn:l-bound}
\E\sum_{i=1}^kx_i^\ell &\ge k- \ell \varepsilon.
\end{align}
\end{lem}
\begin{proof}[Proof of Lemma~\ref{lem:moments}]
By induction. But the intuitive reason is that if the sum of the $x_i$ is $\ge k-\varepsilon$ then even in the worst case, most of the $x_i$ should be at least $1-\varepsilon/k$ and raising to the $\ell$th power and summing gives $\approx k-\ell\varepsilon$.
Write
\begin{align*}
\sum_{i=1}^kx_i^{\ell+1}
&= \sum_{i=1}^k-x_i^\ell(1-x_i)+ \sum_{i=1}^kx_i^\ell\\
&\ge -\sum_{i=1}^k(1-x_i)+\sum_{i=1}^kx_i^\ell=-k+\sum_{i=1}^kx_i+\sum_{i=1}^kx_i^\ell.\numberthis\label{eqn:lp1}
\end{align*}
Taking expectation values and assuming \eqref{eqn:l-bound} holds for $\ell$, applying it and the property for $\ell=1$ gives the desired bound for $\ell+1$.
\end{proof}

Applying the lemma, we obtain $\E \Tr W^\ell\ge k-\ell O(\sqrt{d})$, which with \eqref{eqn:s2} implies 
\begin{align}
\E S_2(U)&\le O(\sqrt{d})\frac{\tanh^2(2s)}{2(1-\tanh^2(2s))}=O(\sqrt{d})\frac{1}{2}\sinh^2(2s),
\end{align}
which is \eqref{eqn:exp-bound}.

For the probability bound \eqref{eqn:prob-bound}, we start with \eqref{eqn:VF} and apply the second part of Proposition~\ref{prop:expbound}.
There are only $O(d^2)$ nonzero terms in the double sum in \eqref{eqn:VF}. 
All these terms have $|x-y|\ge\gamma$, where we take $\gamma=\sqrt{\frac{3+\kappa}{c}d\log d}$.
A union bound with Eq.~\eqref{eqn:uut-prob} of Proposition~\ref{prop:expbound} then shows that with probability $1-O(d^3e^{-c\gamma^2/d})=1-O(d^{-\kappa})$,
\begin{align*}
|\langle x|UU^T|y\rangle|^2&\le d^2e^{-2c(x-y)^2/d}\le \frac{1}{d^{4}},\quad\text{for all $(x,y)$ in \eqref{eqn:VF}}.
\end{align*}
Then with probability $1-O(d^{-\kappa})$,
\begin{align*}
\Tr W &\ge k-\gamma-O(d^2)\frac{1}{d^{4}}=k-O_\kappa(\sqrt{d\log d}). 
\numberthis\label{eqn:probbound-final}
\end{align*}
We can then apply Lemma~\ref{lem:moments} (with deterministic $x_i$) to derive a lower bound on $\Tr W^\ell$; combined with the formula \eqref{eqn:s2}, we thus obtain \eqref{eqn:prob-bound}.
\end{proof}

\section{Boson random walks and proof of Proposition~\ref{prop:expbound}}\label{sec:rw}

In this section, we prove the relationship between the random matrix $U$ and a classical random walk, and then use this to prove Proposition~\ref{prop:expbound}.

\subsection{Boson random walks}\label{subsec:bosonrw}

Every unitary matrix $U$ can be associated with a doubly stochastic matrix $P$ given by $P_{xy}=|U_{xy}|^2$. So it is natural to consider $\E|U_{xy}|^2$, and to determine which classical random walk it corresponds to.

We consider more general geometries than just the brickwall structure. Allowing for non-spatially-local 2-mode beamsplitters, we will consider $U=\prod_{i=d}^1U^{(i)}$, where the steps $U^{(i)}$ are independent and identically distributed, with each $U^{(i)}$ consisting of $M$ layers of beamsplitters/phaseshifters (Figure~\ref{fig:generalstep}). Each of the $M$ layers corresponds to a single ``pairing'' up of the $n$ modes, where each mode is either matched with a partner (corresponding to applying a combination of beamsplitter with phaseshifter on the two modes) or to itself (corresponding to applying a 1-mode phaseshifter). As will be useful in Section~\ref{sec:lower}, a step $U^{(i)}$ can also be associated with a graph whose vertices are the modes, and where two distinct vertices are connected by an edge iff there is a beamsplitter in some $u^{[i,j]}$ which connects the two modes (see e.g.~Figure~\ref{fig:generalstep}).

\begin{figure}[htb]
\begin{tikzpicture}[scale=.5,baseline=(current bounding box.center)]
\def\rcol{gray!30}
\def\spacing{1.4cm}
\begin{scope}[{Circle}-{Circle},shorten >=-2pt, shorten <=-2pt, line width =.8] % draw beamsplitters
\draw[fill=\rcol,color=\rcol] (-.25,.5) rectangle (.75,-5.5); % rectangle
\draw (0,0)--(0,-1);
\draw (0,-2)--(0,-4);
\draw[xshift=4mm] (0,-3)--(0,-5);
\begin{scope}[xshift=\spacing] % 2nd group
\draw[fill=\rcol,color=\rcol] (-.25,.5) rectangle (1,-5.5); % rectangle
\draw (0,0)--(0,-3);
\draw[xshift=4mm] (0,-1)--(0,-4);
\draw[xshift=8mm] (0,-2)--(0,-5);
\end{scope}
\begin{scope}[xshift={2*\spacing+3mm}] %3rd group
\draw[fill=\rcol,color=\rcol] (-.25,.5) rectangle (.75,-5.5); % rectangle
\draw (0,-1)--(0,-2);
\draw (0,-3)--(0,-4);
\draw[xshift=4mm] (0,0)--(0,-5);
\end{scope}
\begin{scope}[xshift={3*\spacing+3mm}] %4th group
\draw[fill=\rcol,color=\rcol] (-.25,.5) rectangle (.75,-5.5); % rectangle
\draw(0,0)--(0,-2);
\draw[xshift=4mm] (0,-1)--(0,-3);
\draw (0,-4)--(0,-5);
\end{scope}
\end{scope}
% lines
\foreach\mode in {0,...,5}{
\draw (-.5,-\mode)--++(6,0);
}
% labels
\foreach\layer in {1,...,4}{
\node[xshift=-6mm] at (1.5*\layer,-6) {$\ui{i,\layer}$};
}
\draw[decoration={brace,raise=3pt,aspect=.5,amplitude=4pt},decorate] (-.25,.5)--(5.25,.5);
\node [above] at (3,1) {$U^{(i)}=\ui{i,4}\ui{i,3}\ui{i,2}\ui{i,1}$};
\end{tikzpicture}
\qquad\qquad
\begin{tikzpicture}[scale=1.3, baseline=(current bounding box.center)]
\draw (0,0) -- (1,1.73205) -- (2,0) -- cycle;
\begin{scope}[xshift=.75cm,yshift=.75cm,scale=.25]
\draw (0,0) -- (1,-1.73205) -- (2,0) -- cycle;
\end{scope}
\draw (0,0)--(.75,.75);
\draw (0,0)--(1,.3);
\draw (2,0)--++(-1,.3);
\draw (2,0)--++(-.75,.75);
\draw (1,1.73205)--++(-.25,-1);
\draw (1,1.73205)--++(.25,-1);
\foreach \vertex in {(0,0),(2,0),(1,1.73205),(.75,.75),(1,.3),(1.25,.75)}{
\draw[fill=black] \vertex circle (2pt);
}
\end{tikzpicture}
\caption{(Left) An example of a single step $U^{(i)}=\ui{i,4}\ui{i,3}\ui{i,2}\ui{i,1}$ consisting of $M=4$ layers and containing non-local 2-mode beamsplitters. Beamsplitters are shown using connecting lines.
(Right) The 4-layer step $U^{(i)}$ shown in the left can be associated with the octahedral graph shown in the right.}\label{fig:generalstep}
\end{figure}
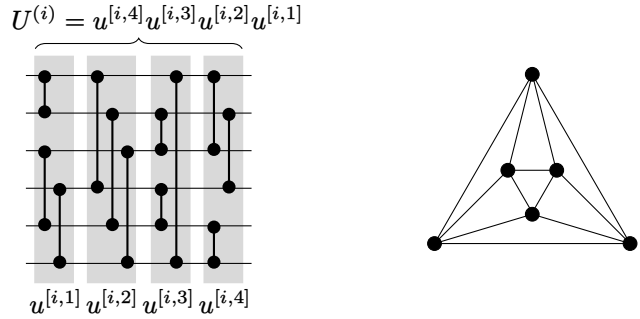

More formally, let $\fakepair(n)$ denote the set of permutations $\pi\in S_n$ which are a product of disjoint transpositions; equivalently it can be associated with the set of partitions of $\{1,\ldots,n\}$ into sets of size $1$ (corresponding to fixed points of $\pi$) or $2$ (corresponding to transpositions in $\pi$).
We may refer to $\pi$ as a ``pairing'', where we mean that we allow self-pairings.
Writing $U^{(i)}=\prod_{j=M}^1\ui{i,j}$ as a product of the $M$ layers, each of the layers $\ui{i,j}$ in the single step $U^{(i)}$ corresponds to pairing up indices $x,y\in\{1,\ldots,n\}$ (where self-pairings $x=y$ are allowed) according to some $\pi_j\in\fakepair(n)$. For each pair $x\ne y$, we put an independent Haar random $2\times2$ unitary matrix across the corresponding submatrix $\begin{bmatrix}\ui{i,j}_{xx}&\ui{i,j}_{xy}\\\ui{i,j}_{yx}&\ui{i,j}_{yy}\end{bmatrix}$ consisting of the rows and columns corresponding to $\{x,y\}$. For each self-paired $x$, we put an (independent) modulus one entry on the corresponding $(x,x)$ matrix entry. All other matrix entries are zero.

For example, for the brickwall circuit, $M=2$, and $\ui{i,1}$ and $\ui{i,2}$ are matrices with the structures $U^L$ and $U^R$ from Figure~\ref{fig:blocks}. The associated permutations in $\fakepair(n)$, which act on $\{1,\ldots,n\}$, are $\pi_1=(1\,2)(3\,4)\cdots(n-1\,n)$ and $\pi_2=(1)(2\,3)\cdots(n-2\,n-1)(n)$. 

\begin{thm}[boson random walk]\label{thm:boson-rw}
Let $U=\prod_{i=d}^1U^{(i)}$, where each $U^{(i)}$ is independent with the same distribution $U^{(i)}\vsim\prod_{j=M}^1\ui{j}$, where $(\ui{j})_j$ are independent and distributed as described above for some fixed $\pi_1,\ldots,\pi_M\in\fakepair(n)$.
Define the symmetric doubly stochastic matrix $\spi{j}$ by $\spi{j}_{xy}:=\E|\ui{j}_{xy}|^2$, which describes a lazy walk on vertices $\{1,\ldots,n\}$ which stays at its current vertex $x$ with probability $1/2$, and moves to $\pi_j(x)$ with probability $1/2$.
Letting $Z_t$ be the discrete-time random walk on $\{1,\ldots,n\}$ with transition probability matrix $P:=\prod_{j=1}^M\spi{j}$, then
\begin{align}
\E|\langle x|U|y\rangle|^2&=\P_y[Z_d=x].
\end{align}
\end{thm}
\begin{rmk}
Note that the order of the factors in $P=\prod_{j=1}^M\spi{j}$ is $\spi{1}\spi{2}\cdots\spi{M}$, in contrast to the reverse indexing in $U^{(i)}\vsim\ui{M}\cdots\ui{2}\ui{1}$. This is due to the convention that for transition probability matrices, $\langle y|P|x\rangle$ gives the probability for a random walk starting at $y$ to move to $x$, while for the unitary circuit, $U^{(i)}$ is viewed as acting as $U^{(i)}|y\rangle$ on an initial $|y\rangle$. 
\end{rmk}

\begin{rmk} 
The matrix $UU^T$ is unitary and so can also be associated with a classical random walk. But it appears hard to determine what the corresponding classical random walk should be since $\E|\langle x|UU^T|y\rangle|^2$ is difficult to compute.
\end{rmk}

The proof of Theorem~\ref{thm:boson-rw} is a quick computation using independence of all the $1\times1$ and $2\times2$ Haar-random blocks in the unitary factors making up $U^{(i)}$ and $U$. Before giving the proof, we note that
Theorem~\ref{thm:boson-rw} immediately implies the following random walk relation for the brickwall circuit.
\begin{cor}[brickwall random walk]\label{cor:rw}
Let $n\ge2$ be even.
Let $U=\prod_{i=d}^1 U^{(i)}$ be a random depth $d$ brickwall circuit on $\{1,\ldots,n\}$, with each $U^{(i)}\dsim U^RU^L$ with independent $2\times2$ and $1\times1$ Haar-random blocks as described in Figure~\ref{fig:blocks}. Then
\begin{align}
\E|\langle x|U|y\rangle|^2&=\P_y[Z_d=x],
\end{align}
where $Z_t$ is the discrete-time random walk on $\{1,\ldots,n\}$ defined as follows: Letting the boundary be $\partial_1=\{1,n\}$, 
\begin{align}\label{eqn:Zt}
\P[Z_t=x|Z_{t-1}=y]&=\begin{cases}
\frac{1}{4},&\text{ for $x\in\{y-1,y,y+1,y+2\}\cap \partial_1^c$ if $y$ is odd}\\
\frac{1}{4},&\text{ for $x\in\{y-2,y-1,y,y+1\}\cap\partial_1^c$ if $y$ is even}\\
\frac{1}{2},&\text{ if $x=1$ and $y\in\{1,2\}$, or $x=n$ and $y\in\{n-1,n\}$}
\end{cases}.
\end{align}
\end{cor}
\begin{proof}
Equation~\eqref{eqn:Zt} is the formula for the transition matrix $P=\spi{1}\spi{2}=\E |U^{L}|^2\E|U^{R}|^2$.
\end{proof}

The brickwall random walk in Corollary~\ref{cor:rw} is depicted in Figure~\ref{fig:brickwall_rw}.
\begin{figure}[htb]
\begin{tikzpicture}[scale=2]
% options for tree
\def\mrad{1mm} % vertex radius
\def\xcolor{blue} % x color
\def\ycolor{red} % y color
\def\bcolor{gray!60} % brick color
\def\brickmargin{.35cm} % size of bricks
\def\pathwidth{1.8} % line width for x, y paths
% tree parameters: depth, brick color, shift (w/ units), depth
\newcommand{\tree}[4]{\begin{scope}[rotate=270,scale=.25,xshift=-1cm,yshift=-1cm]
\draw (1,1)--++(-1,-1);
\draw (1,1)--++(0,-1);
% draw bricks
\foreach \height in {0,...,#4}{
\foreach \lr in {0,...,\height}{
\draw[draw=#2,fill=none,xshift=-\brickmargin,yshift=-\brickmargin+#3] (2*\lr-\height, -\height) rectangle ++({1cm+2*\brickmargin},{2*\brickmargin});
}}
% draw nodes
\draw[fill=black] (1,1) circle (\mrad);
\foreach \height in {0,...,#1}{
\foreach \lr in {0,...,\height}{
% Left step edges
\draw (2*\lr-\height,-\height)--++(-1,-1);
\draw (2*\lr-\height,-\height)--++(0,-1);
% right step edges
\draw (2*\lr-\height+1,-\height)--++(1,-1);
\draw (2*\lr-\height+1,-\height)--++(0,-1);
% draw nodes
\draw[fill=black] (2*\lr-\height,-\height) circle (\mrad);
\draw[fill=black] (2*\lr-\height+1,-\height) circle (\mrad);
}
}
\end{scope}}
\node at (1.4,1) {$\cdots$};

\begin{scope}[yshift=1cm,rotate=180] 
\tree{2}{\bcolor}{.5cm}{3}
\end{scope}
\node[left] at (0,1) {$y$};
\begin{scope}[yshift=1cm]
\draw[line width=\pathwidth,color=\ycolor] (0,0)--(.25,0)--(.5,0)--(.75,-.25)--(1,-.25);
\end{scope}
% labels
\begin{scope}[xshift=-.1cm]
\draw[decoration={brace,raise=3pt,aspect=.5,amplitude=4pt},decorate,xshift=-.36cm] (.5,1.8)--++(.45,0);
\draw[decoration={brace,raise=3pt,aspect=.5,amplitude=4pt},decorate,xshift=-.36cm] (1,1.8)--++(.45,0);
\node [above] at (.5,1.9) {$U^{(1)}$};
\node [above] at (1, 1.9) {$U^{(2)}$};
\begin{scope}[xshift=-0.15cm]
\foreach\i in {1,2,3,4}{
\ifodd\i \node[below,xshift=3mm] at (\i/4,-.1) {$U^L$};
\else \node[below,xshift=3mm] at (\i/4,-.11) {$U^R$};
\fi
}
\end{scope}
\end{scope}
\begin{scope}[xscale=-1,yshift=-.8cm]
% t axis
\def\theight{2mm}
\node[left] at (0,\theight) {$t$};
\draw [->](0,\theight)--++(-1.25,0);
\foreach \t in {0,...,4}{\draw[yshift=\theight] (-.25*\t,-.5mm)--++(0,1mm);}
\node[above] at (-.5,\theight+.5mm) {$1$};
\node[above] at (0,\theight+.5mm) {$0$};
\node[above] at (-1,\theight+.5mm) {$2$};
\end{scope}
\draw[->] (-.25,1.75)--++(0,-2) node[below] {$n$}; 
\end{tikzpicture}
\caption{A brickwall random walk $Z_t$ as defined in Corollary~\ref{cor:rw}. 
Modes are labeled from top ($1$) to bottom ($n$).}\label{fig:brickwall_rw}
\end{figure}
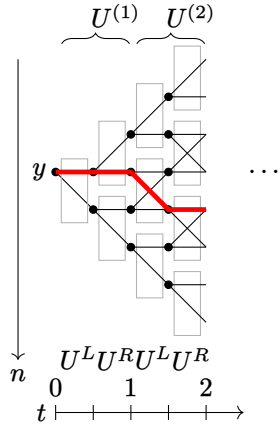

Theorem~\ref{thm:boson-rw} follows from the following rephrased version which makes writing the proof more convenient.
\begin{lem}[rephrased Theorem~\ref{thm:boson-rw}]\label{lem:prodexp}
Let $U=\prod_{j=K}^1\ui{j}$, where each $\ui{j}=\Ti{j}\Bi{j}(\Ti{j})^{-1}$ is an independent $n\times n$ matrix, with $\Bi{j}=\oplus_{k=1}^{p_j}\bj{j,k}$ where each $\bj{j,k}$ is a $1\times 1$ or $2\times 2$ independent Haar random unitary matrix, and $\Ti{j}$ a permutation matrix corresponding to a $\tau_j\in S_n$. Letting $\spi{j}$ be the defined as $\spi{j}_{xy}=\E|\ui{j}_{xy}|^2$, then
\begin{align}
\E|\langle x|U|y\rangle|^2&=\langle x|\spi{K}\cdots\spi{2}\spi{1}|y\rangle=\langle y|\spi{1}\spi{2}\cdots\spi{K}|x\rangle.
\end{align}
\end{lem}
\begin{proof}[Proof of Lemma~\ref{lem:prodexp}]
Expand
\begin{align*}
\E|\langle x|U|y\rangle|^2&=\E\sum_{\substack{i_1,\ldots,i_{K-1}=1\\i_1',\ldots,i_{K-1}'=1}}^n\ui{K}_{x {i_{K-1}}}\cdots\ui{2}_{i_2i_1}\ui{1}_{i_{1}y}\bui{K}_{xi_{K-1}'}\cdots\bui{2}_{i_2'i_1'}\bui{1}_{i_1'y}\\
&=\sum_{\substack{i_1,\ldots,i_{K-1}=1\\i_1',\ldots,i_{K-1}'=1}}^n\E[\ui{K}_{xi_{K-1}}\bui{K}_{xi_{K-1}'}]\cdots \E[\ui{2}_{i_2i_1}\bui{2}_{i_2'i_1'}]\E[\ui{1}_{i_1y}\bui{1}_{i_1'y}].\numberthis\label{eqn:uk}
\end{align*}
So we just need to calculate $\E[\ui{j}_{\alpha\beta}\bui{j}_{\alpha'\beta'}]$ for arbitrary indices $\alpha,\beta,\alpha',\beta'$.
But since the nonzero entries of $\ui{j}$ are simply distinct entries from $1\times1$ or $2\times2$ Haar random matrices, we see we must have $\alpha=\alpha'$ and $\beta=\beta'$ for the expectation $\E[\ui{j}_{\alpha\beta}\bui{j}_{\alpha'\beta'}]$ to be nonzero.
Thus $\E[\ui{j}_{\alpha\beta}\bui{j}_{\alpha'\beta'}]=\delta_{\alpha\alpha'}\delta_{\beta\beta'}\spi{j}_{\alpha\beta}$, and we obtain
\begin{align}
\E|\langle x|U|y\rangle|^2&=\sum_{i_1,\ldots,i_{K-1}=1}^n\spi{K}_{xi_{K-1}}\cdots\spi{2}_{i_2i_1}\spi{1}_{i_1y}=\langle x|\spi{K}\cdots\spi{2}\spi{1}|y\rangle.
\end{align}
This is equal to $\langle y|\spi{1}\spi{2}\cdots\spi{K}|x\rangle$ since all $\spi{j}$ are real and symmetric. 
\end{proof}

For the brickwall structure, the random walk $Z_t$ defined in \eqref{eqn:Zt} can be viewed as a reflected version of the associated random walk on $\Z$ ($n=\infty$). This will allow us to use central limit theorem or large deviation bounds for the walk on $\Z$.
\begin{lem}[reflected walk]\label{lem:reflect}
Let $S_t$ be the random walk on $\Z$ defined as 
\begin{align}\label{eqn:rwZ}
\P[S_t=x|S_{t-1}=y]&=\begin{cases}
\frac{1}{4},&\text{ for $x\in\{y-1,y,y+1,y+2\}$ if $y$ is odd}\\
\frac{1}{4},&\text{ for $x\in\{y-2,y-1,y,y+1\}$ if $y$ is even}
\end{cases}.
\end{align}
Set $[x]_{2n}:=x\mod 2n\in\intbrr{1:2n}$, and let $R:\Z\to\intbrr{1:n}$ be the reflecting map
\begin{align}\label{eqn:R}
R(x):=\begin{cases}
[x]_{2n},&1\le [x]_{2n}\le n\\
2n+1-[x]_{2n},&n+1\le [x]_{2n}\le 2n
\end{cases}.
\end{align}
Then for $n\ge 2$ even, the brickwall random walk $Z_t$ defined in \eqref{eqn:Zt} has the same law as the reflected random walk $Y_t:=R(S_t)$.
\end{lem}
\begin{proof}
The map $R$ defines an equivalence relation on $\Z$ via $x\sim y$ iff $R(x)=R(y)$, and one can check that the transition probabilities $\P[Y_t=x|Y_{t-1}=y]$ for $Y_t$ agree with \eqref{eqn:Zt}.
\end{proof}

\subsection{Random walk large deviations}
To apply Corollary~\ref{cor:rw}, we need to estimate the random walk probabilities $\P_y[Z_d=x]$. For this, we can apply some rough large deviation bounds.

The (non-reflected) random walk $S_t$ on $\Z$ defined in Lemma~\ref{lem:reflect} is similar to a symmetric random walk with a finite hopping distance on $\Z$. The individual steps are not independent since they depend on the parity of the current position of the walk, so the usual Central Limit Theorem (CLT) and local CLT for simple random walk do not directly apply.
The $\alpha$-mixing CLT \cite{ibragimov1971independent} does apply.
However, since we need quantitative estimates on the random walk probabilities for Proposition~\ref{prop:expbound}, we will instead relate $S_t$ to random walks with independent increments so that we can use standard quantitative large deviation bounds.

If we write $S_t=y+\sum_{i=1}^tX_i$, where $X_i\in\{-2,-1,0,1,2\}$ is an individual step (corresponding to two layers of the brickwall, i.e. one $U^{(i)}$ step), then the distribution for the next step $X_{t+1}$ depends on the parity of $S_t$: if $S_t$ is odd, then $X_{t+1}\in\{-1,0,1,2\}$ with probability $1/4$ of each outcome, while if $S_t$ is even, then $X_{t+1}\in\{-2,-1,0,1\}$ with probability $1/4$ of each outcome.
However, we can see that $X_{t+1}$ is independent of $X_{t-1},X_{t-2},\ldots,X_1$. 
One can check the marginal distribution of $X_{t+1}$ for $t\ge1$ is given by the measure $\frac{1}{8}\delta_{-2}+\frac{1}{8}\delta_2+\frac{1}{4}(\delta_{-1}+\delta_0+\delta_1)$; then a quick computation shows for $t\ge1$,
\begin{align*}
\P_y[&X_{t+1}=x,X_{t-1}=x_{t-1},\ldots,X_1=x_1]\\
&=\sum_{z\in\{-2,-1,0,1,2\}}\P_y[X_{t+1}=x,X_t=z,X_{t-1}=x_{t-1},\ldots,X_1=x_1]\\
&=\begin{multlined}[t]
\sum_{z\in \{-2,-1,0,1,2\}}\P_y[X_{t+1}=x|X_t=z,X_{t-1}=x_{t-1},\ldots,X_1=x_1]\\
\qquad\qquad\qquad\times\P_y[X_t=z|X_{t-1}=x_{t-1},\ldots,X_1=x_1]\P_y[X_{t-1}=x_{t-1},\ldots,X_1=x_1]
\end{multlined}\\
&=2\left(\frac{1}{4}\oneb_{x\in\{-1,0,1,2\}}+\frac{1}{4}\oneb_{x\in\{-2,-1,0,1\}}\right)\frac{1}{4}\P_y[X_{t-1}=x_{t-1},\ldots,X_1=x_1]\\
&=\P[X_{t+1}=x]\P_y[X_{t-1}=x_{t-1},\ldots,X_1=x_1],\numberthis
\end{align*}
where to obtain the second to last line we considered the cases where $x_{t-1}+\cdots+x_1$ is even or odd. 
Moreover, this calculation also implies (e.g. recursively) that the even steps $\{X_2,X_4,X_6,\ldots\}$ are mutually independent, and the odd steps $\{X_1,X_3,X_5,X_7,\ldots\}$ are mutually independent. 
Then we can write
\begin{align}\label{eqn:zsplit}
S_d&=y+\sum_{i=1}^d X_i=y+X_1+\sum_{i=1}^{\lceil d/2\rceil-1}X_{2i+1}+\sum_{i=1}^{\lfloor d/2\rfloor}X_{2i},
\end{align}
with each sum comprised of mutually independent random variables whose marginal distributions are all given by the measure $\frac{1}{8}\delta_{-2}+\frac{1}{8}\delta_2+\frac{1}{4}(\delta_{-1}+\delta_0+\delta_1)$. The entry $X_1$ is pulled out separately since it has a different distribution, as it depends on the parity of the starting position $y$ of the walk.
We can use a simple bound on \eqref{eqn:zsplit} to obtain
\begin{align*}
\P_y[S_d=x]&=\P\left[y+X_1+\sum_{i=1}^{\lceil d/2\rceil-1}X_{2i+1}+\sum_{i=1}^{\lfloor d/2\rfloor}X_{2i}=x\right]\\
&\le \P\left[\bigg|\sum_{i=1}^{\lceil d/2\rceil-1}X_{2i+1}\bigg|\ge |x-y|/2-1\right]+\P\left[\bigg|\sum_{i=1}^{\lfloor d/2\rfloor}X_{2i}\bigg|\ge |x-y|/2-1\right],\numberthis\label{eqn:split-indep}
\end{align*}
where we used that $|X_1|\le2$ regardless of the starting point $X_0=y$.

We can then apply Hoeffding's inequality,\footnote{Alternatively, if we cared about getting a better constant $c_1$ in the exponent, we could instead use a random walk large deviation result such as \cite[Appendix Corollary A.2.7]{lawlerlimic}.}  which states that for independent random variables $X_1,\ldots,X_n$ with $a\le X_i\le b$, that
\begin{align}
\P\left[\sum_{i=1}^n(X_i-\E X_i)\ge t\right]&\le e^{-\frac{2t^2}{n(b-a)^2}}.
\end{align}
Applying this to \eqref{eqn:split-indep} shows there are numerical constants $C,c_1$ such that for any $x,y\in\Z$ and $d\in\N$,
\begin{align}\label{eqn:rwprob2}
\P_y[S_d=x]&\le Ce^{-c_1(x-y)^2/d}.
\end{align}

\subsection{Proof of Proposition~\ref{prop:expbound}}
We can now give
\begin{proof}[Proof of Proposition~\ref{prop:expbound}(i)]
First note that the bounds are trivial if $d\ge n^2$ and $C\ge e^c$, so we may without loss of generality assume $d\le n^2$.
For \eqref{eqn:uut-exp}, we want to upper bound $\E|\langle x|U|y\rangle|^2$ and $\E|\langle x|UU^T|y\rangle|^2$, the latter of which we will bound in terms of $\E|\langle x|U|y\rangle|^2=\P_y[Z_d=x]$.

We start with the former, $\E|\langle x|U|y\rangle|^2$, though for notational convenience with the latter we will consider $\E|\langle x|U|\alpha\rangle|^2=\P[Z_d=x]$.
Using the reflected walk relation in Lemma~\ref{lem:reflect}, for $S_t$ the walk on $\Z$ defined in \eqref{eqn:rwZ} and $R$ the map in \eqref{eqn:R}, we have for any $x,\alpha\in\intbrr{1:n}$,
\begin{align*}
\P_\alpha[Z_d=x] &= \P_\alpha[R(S_d)=x]\\
&=\P_\alpha[S_d=x]+\P_\alpha[S_d=-x+1]+\P_\alpha[S_d=2n+1-x]+O\Bigg(\sum_{\substack{\beta\overset{R}{\sim} x\\|\beta-\alpha|>n}} \P_\alpha[S_d=\beta]\bigg).\numberthis
\end{align*}
Using the large deviation bound \eqref{eqn:rwprob2}, we first bound the remainder terms as
\begin{align*}
\sum_{\substack{\beta\overset{R}{\sim} x\\|\beta-\alpha|>n}} \P_\alpha[S_d=\beta]\le C\sum_{j=1}^\infty e^{-c_1n^2j^2/d}
&\le C\left(e^{-c_1n^2/d}+\int_1^\infty e^{-c_1n^2j^2/d}\,dj\right)\\
&\le Ce^{-c_1n^2/d}+\frac{Cd}{n^2}e^{-c_1n^2/d}\\
&\le C'e^{-c_1n^2/d},\numberthis
\end{align*}
using the inequality $\operatorname{erfc}(z)\le \frac{e^{-z^2}}{\sqrt{\pi}z}$ for $z\ge0$ and that $d\le n^2$. 
Applying \eqref{eqn:rwprob2} also to the other terms, we get
\begin{align}\label{eqn:Zdalpha}
\P_\alpha[Z_d=x]&\le Ce^{-c_1(x-\alpha)^2/d}+Ce^{-c_1(x+\alpha-1)^2/d}+Ce^{-c_1(\alpha-2n-1+x)^2/d}+O(e^{-c_1n^2/d}),
\end{align}
and similarly for $\P_\alpha[Z_d=y]$. 
Since $x,\alpha\in\intbrr{1:n}$, the largest of the four exponential terms in the bound \eqref{eqn:Zdalpha} is the one with $e^{-c(x-\alpha)^2/d}$, since $|x-\alpha|$ is the smallest of the four terms in the exponentials. Since $n^2>(x-y)^2$, we obtain the first bound in \eqref{eqn:uut-exp}.

Moving to $\E|\langle x|UU^T|y\rangle|^2$, we have
\begin{align}\label{eqn:uut-expand}
\E|\langle x|UU^T|y\rangle|^2&=\E\Big|\sum_{\alpha}\langle x|U|\alpha\rangle\langle y|U|\alpha\rangle\Big|^2
=\sum_{\alpha,\alpha'}\E U_{x\alpha}U_{y\alpha}\bar U_{x\alpha'}\bar U_{y\alpha'}.
\end{align}
Writing $U=\prod_{i=d}^1U^{(i)}=U^{(d)}\cdots U^{(1)}$ with each $U^{(i)}\dsim U^RU^L$, we can pull out the right-most $U^L$ (the one from $U^{(1)}$) to write $U=QU^L$, with $U^L$ independent from $Q$.
We see the terms in \eqref{eqn:uut-expand} with $\alpha\ne\alpha'$ are zero, since
\begin{align}
\E U_{x\alpha}U_{y\alpha}\bar U_{x\alpha'}\bar U_{y\alpha'}&=\sum_{i,j,i',j'=1}^n\E[Q_{xi}Q_{yj}\bar Q_{xi'}\bar Q_{yj'}]\E[U^L_{i\alpha}U^L_{j\alpha}\bar U^L_{i'\alpha'}\bar U^L_{j'\alpha'}],
\end{align}
and the expectation over the $U^L$ entries is zero if $\alpha\ne\alpha'$. (Recall $U^L$ simply contains independent $2\times2$ and $1\times1$ Haar unitary matrix entries, so
we need terms like $U^L_{\ell_1\alpha} \bar U^L_{\ell_1'\alpha}$ for this average to be nonzero.)
Then
\begin{align*}
\E|\langle x|UU^T|y\rangle|^2
&=\E \sum_{\alpha}|U_{x\alpha}|^2|U_{y\alpha}|^2\numberthis\label{eqn:uut-alpha}\\
&\le \sum_{\alpha}\E[|U_{x\alpha}|^4]^{1/2}\E[|U_{y\alpha}|^4]^{1/2},\numberthis
\end{align*}
applying Cauchy--Schwarz for the last line.
Using $|U_{ab}|^2\le1$ so that $|U_{ab}|^4\le|U_{ab}|^2$, the random walk relation in Corollary~\ref{cor:rw} then implies
\begin{align}\label{eqn:csbound}
\E|\langle x|UU^T|y\rangle|^2&\le \sum_\alpha \P_\alpha[Z_d=x]^{1/2}\P_\alpha[Z_d=y]^{1/2}.
\end{align}
From \eqref{eqn:Zdalpha}, we obtain
\begin{align}\label{eqn:Zdprob}
\P_\alpha[Z_d=x]\P_\alpha[Z_d=y]&\le Ce^{-c_1(x-\alpha)^2/d}e^{-c_1(y-\alpha)^2/d}+O(e^{-c_1n^2/d}).
\end{align}

The sum over $\alpha$ in \eqref{eqn:csbound} has at most $4d$ nonzero terms. Using $\sqrt{a+b}\le\sqrt{a}+\sqrt{b}$ and pulling out the remainder term $O(e^{-cn^2/d})$ with the $4d$ factor, we obtain
\begin{align}
\E|\langle x|UU^T|y\rangle|^2&\le \left(\sum_{\alpha\in\Z} Ce^{-c_1(x-y)^2/(4d)}e^{-c_1\left[\alpha-(x+y)/2\right]^2/d}\right)+O(de^{-c_1n^2/(2d)}).
\end{align}
Since
\begin{align}
\sum_{\alpha\in\Z} e^{-c_1\left[\alpha-(x+y)/2\right]^2/d}&\le 2+\int_\R e^{-c_1\left[\alpha-(x+y)/2\right]^2/d}\,d\alpha = 2+\sqrt{\frac{\pi d}{c_1}},
\end{align}
we get
\begin{align}
\E|\langle x|UU^T|y\rangle|^2&\le Cd^{1/2}e^{-c_1(x-y)^2/(4d)}+O(de^{-c_1n^2/(2d)}).
\label{eqn:uut-final}
\end{align}
Since $n^2>(x-y)^2$, we obtain \eqref{eqn:uut-exp}.

For the probability bound \eqref{eqn:uut-prob}, note that
\begin{align}
|\langle x|UU^T|y\rangle|^2&\le \left(\sum_{\alpha}|U_{x\alpha}||U_{y\alpha}|\right)^2,
\end{align}
so we want to bound $|U_{x\alpha}||U_{y\alpha}|$ with high probability.
By Markov's inequality, Cauchy--Schwarz, and Corollary~\ref{cor:rw} with the large deviation bound \eqref{eqn:Zdprob}, for any $\varepsilon>0$,
\begin{align}
\P[|U_{x\alpha}||U_{y\alpha}|>\varepsilon] &\le \frac{(\E|U_{x\alpha}|^2)^{1/2}(\E|U_{y\alpha}|^2)^{1/2}}{\varepsilon}\le \frac{Ce^{-c_1(x-\alpha)^2/(2d)}e^{-c_1(y-\alpha)^2/(2d)}+O(e^{-c_1n^2/(2d)})}{\varepsilon}.
\end{align}
We will take $\varepsilon=\frac{1}{4}e^{-c_1(x-y)^2/(8d)}$ to obtain for any $x,y\in\intbrr{1:n}$, 
\begin{align}
\P\left[|U_{x\alpha}||U_{y\alpha}|>\frac{1}{4}e^{-c_1(x-y)^2/(8d)}\right] &\le Ce^{-c_1(x-y)^2/(8d)}e^{-c_1\left(\alpha-\frac{x+y}{2}\right)^2/d}+O(e^{-3c_1n^2/(8d)}).
\end{align}
A union bound over at most $4d$ indices $\alpha$ where $|U_{x\alpha}|$ can be nonzero shows that
\begin{align*}
\P\left[\sum_{\alpha}|U_{x\alpha}||U_{y\alpha}|>de^{-c_1(x-y)^2/(8d)}\right] &\le \sum_\alpha\left[Ce^{-c_1(x-y)^2/(8d)}e^{-c_1\left(\alpha-\frac{x+y}{2}\right)^2/d}+O(e^{-3c_1n^2/(8d)})\right] \\
&\le C e^{-c_1(x-y)^2/(8d)}\left(2+\int_\R e^{-c_1\alpha^2/d}\,d\alpha\right)+O(de^{-3c_1n^2/(8d)})\\
&\le Cd^{1/2}e^{-c_1(x-y)^2/(8d)}+O(de^{-3c_1n^2/(8d)}).\numberthis
\end{align*}
Thus we obtain
\begin{align}
\P[|\langle x|UU^T|y\rangle|^2>d^2e^{-c_1(x-y)^2/(4d)}] &\le Cd^{1/2}e^{-c_1(x-y)^2/(8d)}+O(de^{-3c_1n^2/(8d)}),
\end{align}
for any $x,y\in\intbrr{1:n}$. Using that $n^2>(x-y)^2$ gives \eqref{eqn:uut-prob}.
\end{proof}

\begin{proof}[Proof of Proposition~\ref{prop:expbound}(ii)]
By \eqref{eqn:uut-alpha}, we have
\begin{align*}
\sum_{x\in\Gamma}\sum_{y\not\in\Gamma}\E|\langle x|UU^T|y\rangle|^2&= \sum_\alpha \E\left[\sum_{x\in\Gamma}|U_{x\alpha}|^2\sum_{y\not\in\Gamma}|U_{y\alpha}|^2\right]\\
&\le \sum_{\alpha\not\in\Gamma}\E\left[\sum_{x\in\Gamma}|U_{x\alpha}|^2\right]+\sum_{\alpha\in\Gamma}\E\left[\sum_{y\not\in\Gamma}|U_{y\alpha}|^2\right]\\ 
&=\sum_{\alpha\not\in\Gamma}\P_\alpha[Z_d\in\Gamma]+\sum_{\alpha\in\Gamma}\P_\alpha[Z_d\not\in\Gamma]\qquad\text{[by Theorem~\ref{thm:boson-rw}]}.\numberthis\label{eqn:walk-prob4}
\end{align*}
For $\Gamma=\{1,\ldots,k\}$, using Hoeffding's inequality as in \eqref{eqn:split-indep} and \eqref{eqn:rwprob2} to bound $\P_y[Z_d\ge x]\le Ce^{-c(x-y)^2/d}$ for $x>y$, we can estimate
\begin{align*}
\sum_{x\in\Gamma}\sum_{y\not\in\Gamma}\E|\langle x|UU^T|y\rangle|^2&\le  \sum_{\alpha=k+1}^n\P_\alpha[Z_d\le k]+\sum_{\alpha=1}^k\P_\alpha[Z_d\ge k+1]\\
&\le C\sum_{j=1}^\infty e^{-cj^2/d}\le C'\sqrt{d},\numberthis
\end{align*}
using e.g. $\sum_{z=1}^\infty e^{-cz^2/d}\le e^{-c/d}+\int_1^\infty e^{-cz^2/d}\,dz\le C\sqrt{d}$.
\end{proof}

\subsection{Higher-dimensional brickwork circuit}\label{subsec:higherdim}

Corollary~\ref{cor:upperbound2}(ii) follows from the above 1D bounds as follows.
Let $U=\prod_{i=d}^1U^{(i)}$ be the depth $d$ $D$-dimensional brickwork circuit as defined in Remark~\ref{rmk:worst}, with Haar random independent $1\times1$ and $2\times2$ blocks within the $2D$ layers of each $U^{(i)}$.
Index the $n=m^D$ modes as $\{1,\ldots,m\}^D$.

Starting from \eqref{eqn:walk-prob4}, which holds for any geometry, and using the 1D bounds, we can estimate e.g.
\begin{align*}
\sum_{\alpha\in\Gamma}\P_\alpha[Z_d\not\in\Gamma]&\le\sum_{\alpha\in\Gamma}\sum_{j=1}^D\oneb_{k_j<m}\P[(Z_d)_j>k_j]\\
&\le \sum_{j=1}^D\oneb_{k_j<m}\sum_{\alpha_1=1}^{k_1}\cdots\sum_{\alpha_D=1}^{k_D}Ce^{-c(k_j+1-\alpha_j)^2/d}\\
&\le\sum_{j=1}^D\oneb_{k_j<m}\Bigg(\prod_{\substack{\ell=1\\\ell\ne j}}^Dk_\ell\Bigg) C'\sqrt{d}=C'A_\Gamma\sqrt{d}. \numberthis
\end{align*}
A similar estimate holds for $\sum_{\alpha\not\in\Gamma}\P_\alpha[Z_d\in\Gamma]$ by considering the reversed walk. Then with \eqref{eqn:walk-prob4}, we obtain
\begin{align*}
\E\Tr W=\E\|V\|_\hs^2&=|\Gamma|-\sum_{x\in\Gamma}\sum_{y\not\in\Gamma}\E|\langle x|UU^T|y\rangle|^2
\ge|\Gamma|-C'A_\Gamma\sqrt{d},
\end{align*}
which by \eqref{eqn:s2} gives Corollary~\ref{cor:upperbound2}(ii). \qed

\section{Lower bounds beyond the classical mixing and meeting times}\label{sec:lower}

In this section, we prove the lower bound Theorem~\ref{thm:gen}. 

The relaxation time for a lazy 1D reflected simple random walk (SRW) on $\{1,\ldots,n\}$ is $\Theta(n^2)$ \cite[\S12]{levin2009book}. This implies that by time $O(n^2\log n)$, the distribution of the random walk is very close to the uniform measure on vertices. The brickwall walk $Z_t$ of Corollary~\ref{cor:rw} behaves similarly as the SRW by Corollary~\ref{cor:mixsrw} below, so that by depth $d=\Omega(n^2\log n)$, the probabilities $\E|\langle x|U|y\rangle|^2=\P_y[Z_d=x]$ are all around the same size $1/n$. To obtain a lower bound on the expected Renyi-2 entropy $\E S_2(U)$, we however need to lower bound $\E|\langle x|UU^T|y\rangle|^2$, which requires lower bounds on the fourth moments $\E[|U_{x\alpha}|^2|U_{y\alpha}|^2]$ by \eqref{eqn:uut-alpha}. To relate this to the entries $\E|U_{ij}|^2$, we will prove an entrywise decoupling property for depths larger than the classical walk mixing time plus meeting time (Lemma~\ref{lem:dec}). This will be the main technical lemma of this section.
We first review classical walk preliminaries and examples in Section~\ref{subsec:mixmeet}, and then state the main results Lemma~\ref{lem:dec}, Theorem~\ref{thm:lowerbound}, and Corollary~\ref{cor:haar} in Section~\ref{subsec:lower-main}. The remaining Sections~\ref{subsec:decoupling} and  \ref{subsec:proof-lower} contain the proofs of the lemma and theorem respectively.

\subsection{Classical walk preliminaries}\label{subsec:mixmeet}
In this section we will consider general circuit geometries and their associated classical random walks as defined in Section~\ref{subsec:bosonrw}.
Without loss of generality we will assume the circuit geometry is connected in the sense that the associated classical random walk of Theorem~\ref{thm:boson-rw} can move from any mode $x$ to any other mode $y$ in some number of steps (which may depend on $(x,y)$). 
Then the classical walk is irreducible and aperiodic (the latter immediate from the nonzero transition probability of staying at a single vertex), and so the walk distribution will converge to a unique stationary distribution, which in this case is $\phi=(1/n,\ldots,1/n)$ since all the transition matrices are doubly stochastic. 

Because a single time-step of the circuit involves applying $M$ layers of gates in a particular order (as in Figure \ref{fig:generalstep}), the classical walk of Theorem~\ref{thm:boson-rw} is in general not reversible.
We start with a well-known result for bounding the mixing time of a non-reversible walk in terms of one of its symmetrizations. 
Recall (see e.g.~\cite{levin2009book,aldous-fill-2014}) that the definition of mixing time for a finite state irreducible, aperiodic Markov chain is
\begin{align}\label{eqn:tmix}
\tmix(\varepsilon):=\min\{t:\max_y\|P^t(y,\cdot)-\phi\|_\tv\le\varepsilon\},
\end{align}
where $\phi$ denotes the stationary distribution, $P^t(x,\cdot)$ the distribution at time $t$ started from site $y$, and $\|\cdot\|_\tv$ the total variation distance, which for discrete distributions can be given as $\|\mu-\nu\|_\tv=\frac{1}{2}\sum_{x}|\mu(x)-\nu(x)|$.
A reversible Markov chain is one which satisfies $\phi(x)P(x,y)=\phi(y)P(y,x)$ for $\phi$ the stationary distribution.
For reversible ergodic chains, the rate of convergence to stationarity is governed by the spectral gap, which is the smallest nonzero eigenvalue of $I-P$. 
For (lazy) non-reversible walks, the rate of convergence can be bounded in terms of the spectral gap of a symmetrized Laplacian:
\begin{thm}[see e.g.~\cite{chung2005laplacians, fill1991eigenvalue}]\label{thm:nonrev}
Let $P$ be an irreducible transition matrix for a finite-state Markov chain, with holding probabilities ($=P_{xx}$) at least $\alpha>0$ for each site, and stationary vector $\phi$. Let $\Phi=\op{diag}(\phi)$ be the diagonal matrix whose diagonal entries are given by $\phi$.
Let $0=\lambda_0\le\lambda_1\le\cdots\le \lambda_{n-1}$ be the eigenvalues of the Laplacian $\mathscr L=I-\frac{1}{2}(\Phi^{1/2}P\Phi^{-1/2}+\Phi^{-1/2}P^T\Phi^{1/2})$.
Then
\begin{align}\label{eqn:thm41bound}
\max_y\|P^t(y,\cdot)-\phi\|_\tv&\le \frac{(1-2\alpha\lambda_1)^{t/2}}{2\sqrt{\min_x\phi(x)}}.
\end{align}
\end{thm}
\cite[Corollary 2.9]{fill1991eigenvalue} and \cite[Theorem 7.3]{chung2005laplacians} state convergence assuming the chain has holding probabilities at least $1/2$ for each site. For completeness, we write a proof of the version above with $\alpha>0$ in Appendix~\ref{appendix:sym}, following the proof of \cite[Theorems 7.2, 7.3]{chung2005laplacians}.
Note that if $P$ is doubly stochastic, then $\Phi$ is constant diagonal, and $\lambda_1$ is then the spectral gap of the additive symmetrization $\frac12(P+P^T)$ of $P$.

The classical walks generated by Theorem~\ref{thm:boson-rw} can be viewed as random walks on directed graphs. We would like to relate the mixing times of those walks to the spectral gap of simple random walk (SRW) on simpler non-directed graphs $G_n$. 
\begin{cor}\label{cor:mixsrw}
Let $G_n=(V_n,E_n)$ be a sequence of (non-directed) connected graphs with $|V_n|=n$. Suppose there is $M>0$ so that for every $n$, there exists a sequence of permutations $\pi_{n1},\ldots,\pi_{nM}\in \fakepair(n)$ such that for $x\ne y$, $(x,y)\in E_n$ iff there is some $j\in\intbrr{1:M}$ such that $\pi_{nj}(x)=y$.
Consider the random walk on a (directed) graph $G_n'=(V_n,E_n')$ whose transition probabilities are given by the doubly stochastic matrix $P=P(\pi_{n1})\cdots P(\pi_{nM})$, for $P(\pi)_{xy}:=\frac{1}{2}\delta_{x,y}+\frac{1}{2}\delta_{\pi(x),y}$. 
Then letting $\phi=(1/n,\ldots,1/n)$ be the stationary distribution for $P$, and $\Delta_n$ the spectral gap of the SRW walk on $G_n$, 
\begin{align}
\max_y\|P^t(y,\cdot)-\phi\|_\tv&\le \sqrt{n}(1-2^{-2M+1}\Delta_n)^{t/2},
\end{align}
and so
\begin{align}
\tmix(\varepsilon)&\le \frac{2^{2M}\log\left(\frac{\sqrt{n}}{\varepsilon}\right)}{\Delta_n}.
\end{align}
\end{cor}
\begin{proof}
Using Theorem~\ref{thm:nonrev}, it will suffice to compare the spectral gap $\lambda_1$ of the symmetrized walk with transition matrix $P'=\frac{P+P^T}{2}$, to the spectral gap $\Delta_n$ of the SRW on $G_n$, whose transition probability matrix we will denote by $P_{\mathrm{srw}}$ and is given by $P_{\mathrm{srw}}(u,v)=1/\deg(u)$ if $(u,v)$ is an edge and $0$ otherwise. To compare the spectral gaps, we can use the path method for reversible transition matrices \cite{diaconis1993comparison}, see e.g.~\cite[Theorem 13.23]{levin2009book}. 
Note that for an edge $(x,y)\in G_n$, by construction $P_{xy}$ and $P^T_{xy}$ are both at least $2^{-M}$ since the transposition such that $\pi_{nj}(x)=y$ shows up in each length $M$ sequence, and one can simply hold at a vertex for all other steps in the sequence. 
Using that the stationary distribution of the SRW is $\phi_\srw(x)=\deg(x)/(2|E|)$, we can check the \emph{congestion ratio} $B$ defined in \cite[Theorem 13.20]{levin2009book} for the simple $E'$-path choice $\Gamma_{xy}=(x,y)$ for $(x,y)\in E_\srw$ satisfies $B\le 2^Mn/(2|E_\srw|)$.
Then the path comparison method \cite[Theorem 13.23]{levin2009book} gives
\begin{align}
\lambda_1&\ge 2^{-M}\Delta_n.
\end{align}
Applying this to the conclusion of Theorem~\ref{thm:nonrev} for holding probabilities $\ge \alpha=2^{-M}$ gives the corollary.
\end{proof}

For a sequence of graphs $G_n$ in Corollary~\ref{cor:mixsrw} with uniformly bounded degree, i.e.~$\sup_n\max_{v\in V_n}\deg_{G_n}(v)\le D<\infty$, the constant $M$ can be taken to be at most $D+1$ due to Vizing's theorem on edge colorings of graphs (see e.g.~\cite[Theorem 5.3.2]{diestel2005book}).

\begin{ex}\label{ex:bmix}
In terms of the set-up described in Corollary~\ref{cor:mixsrw}, the non-directed graph $G_n$ with vertices $V_n=\{1,\ldots,n\}$ and edges between adjacent sites can be associated with the classical walk associated with the 1D brickwall circuit. The SRW on $G_n$ is the standard reflected SRW on a line segment, for which one can explicitly calculate the spectral gap, which is $\Delta_n=1-\cos(\frac{\pi}{n-1})=\frac{\pi^2}{n^2}(1+o(1))$ \cite[Example 12.10]{levin2009book}. From Corollary~\ref{cor:mixsrw} we then immediately obtain a bound on the mixing time $\tmix(\varepsilon)$ for the classical brickwall random walk of Corollary~\ref{cor:rw},
\begin{align}
\tmix(\varepsilon)&\le Cn^2\log\left(\frac{\sqrt{n}}{\varepsilon}\right).
\end{align}
We will typically want $\varepsilon=o(1/n)$, e.g.~$\varepsilon=1/n^\alpha$ for some $\alpha>1$, for which the above bound is $O(n^2\log n)$.
\end{ex}

We will also make use of the following standard lemma on Markov chain convergence.
\begin{lem}[see e.g.~{\cite[Exercise 4.3]{levin2009book}}]\label{lem:aftermixing}
Let $P$ be a Markov transition matrix on the state space $X$, and let $\mu$ and $\nu$ any two distributions on $X$. Then
\begin{align}
\|\mu P-\nu P\|_\tv&\le \|\mu-\nu\|_\tv.
\end{align}
\end{lem}
In particular, since we always work with doubly stochastic matrices which all have the same stationary distribution $\phi=(1/n,\ldots,1/n)$, taking $\nu=\phi$ implies that advancing the chain by any doubly stochastic transition matrix can only decrease the distance to $\phi$.

In addition to classical mixing times, we will need to consider classical meeting times. 
Recall from Theorem~\ref{thm:boson-rw} that we consider classical random walks $Z_t$ on $\{1,\ldots,n\}$ whose transition probability matrix for a single timestep $\delta t=1$ is given by $P=\prod_{j=1}^M\spi{j}$, where each $\spi{j}$ describes a lazy walk corresponding to a permutation $\pi_j\in\fakepair(n)$.
For meeting times, we will need to consider the expanded classical walk where $P=\prod_{j=1}^M\spi{j}$ is broken up into its $M$ component steps, i.e.~each individual $\spi{j}$ is considered as a time-step.
We will index these time-steps by fractional depths $t\in\frac{1}{M}\N$, where $t=\frac{K}{M}$ for $K\in\N$ corresponds to $K$ applications of $\spi{j}$ matrices (Figure~\ref{fig:meet}).
This way, $t\in\N$ still corresponds to $t$ big steps $P^t$. 

Now let $Z_t,Z_t'$ be two such independent, identically distributed walks indexed by $t\in\frac{1}{M}\N$. 
Each walk evolves in time-steps $\frac{1}{M}$ according to a sequence of transition probability matrices describing lazy walk steps according to the sequence $\pi_1,\ldots,\pi_M\in\fakepair(n)$ repeating. Starting the walks at $Z_0=x$ and $Z_0'=y$, let $T(Z_t,Z_t')\in\frac{1}{M}\N$ be the random variable indicating the first time that the walks ``$\pi$-meet'', which we define as the first time $T=T(Z_t,Z_t')\in\frac{1}{M}\N_{>0}$ such that $\{Z_T,Z_T'\}\in\pii{TM}$ (recall $\pii{TM}$ is a partition of $\{1,\ldots,n\}$), with the index on $\pii{TM}$ taken modulo $M$.
Equivalently, $T=T(Z_t,Z_t')$ is the first time such that $\{Z_T,Z_T'\}$ have just passed through the same beamsplitter/phaseshifter. 
We define the high-probability $\pi$-meeting time $\tmeet(\varepsilon)$ as
\begin{align}\label{eqn:tmeet}
\tmeet(\varepsilon):=\min\{s\ge0:\P_{xy}[T(Z_t,Z_t')>s]\le\varepsilon\,\,\,\forall x,y\},
\end{align}
where the subscript $xy$ in $\P_{xy}$ indicates that the walks start at $Z_0=x$ and $Z_0'=y$, 
and where $t,s\in\frac{1}{M}\N$.
Note that, if the circuit geometry were not connected, then $\tmeet(\varepsilon)=\infty$.

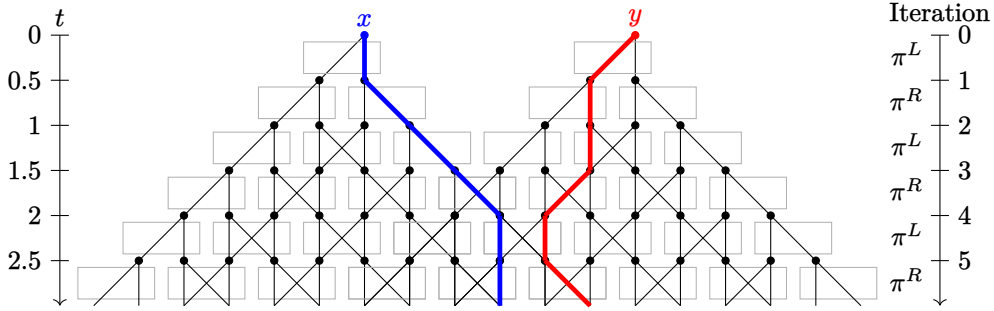
\begin{figure}[htb]
\begin{tikzpicture}[scale=.6]
\def\pdepth{4} % depth (actually depth-1)
\def\bdepth{5} % brick layer depth
\def\xpath{0,1,2,3,3,3} % x path coordinates
\def\ypath{-1,-1,-1,-2,-2,-1} % y path coordinates
% other options
\def\mrad{.8mm} % vertex radius
\def\xcolor{blue} % x color
\def\ycolor{red} % y color
\def\bcolor{gray!60} % brick color
\def\brickmargin{.35cm} % size of bricks
\def\timeshift{3.5} % amount to shift time axis to the left
\def\pathwidth{1.8} % line width for x, y paths
\def\yshift{6} % shift for y starting point
% define tree
\def\tree{\begin{scope}[xshift=-1cm,yshift=-1cm]
\foreach \height in {0,...,\bdepth}{
\foreach \lr in {0,...,\height}{
% draw bricks
\draw[\bcolor,fill=none,xshift=-\brickmargin,yshift=-\brickmargin+0.5cm] (2*\lr-\height, -\height) rectangle ++({1cm+2*\brickmargin},{2*\brickmargin});
}}
\foreach \height in {0,...,\pdepth}{
\foreach \lr in {0,...,\height}{
% Left step edges
\draw (2*\lr-\height,-\height)--++(-1,-1);
\draw (2*\lr-\height,-\height)--++(0,-1);
% right step edges
\draw (2*\lr-\height+1,-\height)--++(1,-1);
\draw (2*\lr-\height+1,-\height)--++(0,-1);
% draw nodes
\draw[fill=black] (2*\lr-\height,-\height) circle (\mrad);
\draw[fill=black] (2*\lr-\height+1,-\height) circle (\mrad);
}
}
\end{scope}}
%
% draw x tree
\draw[] (0,0)--++(-1,-1);
\draw[] (0,0)--++(0,-1);
\draw[\xcolor,fill=\xcolor] (0,0) circle (\mrad);
\node[above,\xcolor] at (0,0) {$x$};
\tree
% draw y tree
\begin{scope}[xshift={\yshift cm}]
\draw[] (0,0)--++(-1,-1);
\draw[] (0,0)--++(0,-1);
\draw[\ycolor,fill=\ycolor] (0,0) circle (\mrad);
\node[above,\ycolor] at (0,0) {$y$};
\tree
\end{scope}
% draw time axes on left
\tikzmath{\tdepth = (\pdepth+1)/2;} % height of axis
\tikzmath{\tshift= -\pdepth-\timeshift+.75;} % left shift
\draw[->] (\tshift,0) --++ (0,-\pdepth-2);
\node[above] at (\tshift,0) {$t$};
% draw iteration depth on right
\draw[->] (-\tshift+\yshift,0) --++ (0,-\pdepth-2);
\node[above] at (-\tshift+\yshift,1mm) {Iteration};
\foreach \t[evaluate={\tint=int(2*\t); \rlabel=int(2*\t-1);}] in {0,0.5,...,\tdepth}{
\draw[xshift={\tshift cm}] (2mm,-2*\t)--++(-4mm,0) node[left] {$\t$};
\draw[xshift={-\tshift cm+\yshift cm}] (-2mm,-2*\t)--++(4mm,0) node[right] {$\tint$};
% pi^L pi^R labels
\ifnum\tint > -1
\ifodd\tint
\node[left,xshift=-1mm,yshift=-0.25cm] at (-\tshift+\yshift,-2*\t) {$\pi^R$};
\else \node[left,xshift=-1mm,yshift=-0.25cm] at (-\tshift+\yshift,-2*\t) {$\pi^L$};
\fi
\fi
}
\def\xprev{0}
\def\yprev{0}
\foreach \x[count=\i, remember=\x as \xprev] in \xpath {
\draw[line width=\pathwidth,\xcolor] (\xprev,-\i+1)--(\x,-\i);}
\foreach \y[count=\i, remember=\y as \yprev] in \ypath {
\draw[line width=\pathwidth,xshift={\yshift cm},\ycolor] (\yprev,-\i+1)--(\y,-\i);}
\end{tikzpicture}
\caption{Tree illustration of all possible paths, starting at $x$ or $y$, in the brickwall circuit. The beamsplitter structure is superimposed on the tree to make indexing the path and $\pi$-meeting time easier. The two specific bolded color paths $\pi$-meet at time $t=2.5$ where they have just passed through the same beamsplitter/phaseshifter.
}\label{fig:meet}
\end{figure}

\begin{ex}\label{ex:meet}
For a sequence of graphs $G_n$ on $n$ vertices, let 
\begin{align}
\begin{aligned}
\tmixs&:=\tmix(o(1/n))\in\N\cup\{\infty\},\quad\text{ for some fixed choice of $o(1/n)$ function},\\
\tmeets&:=\tmeet(o(1/n^2))\in\frac{1}{M}\N\cup\{\infty\},\quad\text{ for some fixed choice of $o(1/n^2)$ function}.
\end{aligned}
\end{align}
\begin{enumerate}

\item For the brickwall geometry, we can compare $\tmeets$ and $\tmixs=\tmix(1/n^\alpha)$, $\alpha>1$,  using a simple argument. Suppose $Z_0=x>y=Z_0'$. Because the walks are in 1D, if at any time $Z_t\le Z_t'$, then the walks must have $\pi$-met (see e.g.~Figure~\ref{fig:meet}). At $\tmixs$, the walks are nearly equally likely to be anywhere, and so there is 
a probability $\frac{1}{2}+O(1/{n^{\min(1,\alpha-1)}})$ that $Z_{\tmixs}\le Z_{\tmixs}'$.
If we repeat this check at $2\tmixs,3\tmixs,\ldots$, then by time $t=(2+\delta)\tmixs \log_2 n$ for any $\delta>0$, there is probability $1-O(1/n^{2+\delta})$
that we have seen $Z_t\le Z_t'$, which means the walks have $\pi$-met. Using Example~\ref{ex:bmix} to take $\tmixs=\tmix(1/n^\alpha)=O(n^2\log n)$, this implies we can take $\tmeets=\tmeet(O(1/n^{2+\delta}))=O(n^2\log^2n)$.

We expect it is likely one of the $\log n$ factors in the bound could probably be removed by going through a more careful estimate, such as adapting \cite[Proposition 14.5]{aldous-fill-2014} or \cite[Proposition B.9]{kanade2023coalescence} to the brickwall walk.

\item For a general geometry, we can always have the bound $\tmeets\le (2+\delta)n\log(n)\,\tmixs$ using a similar type of simple argument.
If the walks have not $\pi$-met by time $\tmixs=\tmix(1/n^\alpha)$, $\alpha>1$, then they have probability at least $\frac{1+O(n^{-(\alpha-1)})}{n}$ of (exactly) meeting at time $2\tmixs$ since the transition probabilities for a $\tmixs$ time-step are all $\frac{1+O(n^{-(\alpha-1)})}{n}$. Thus
\begin{align}\label{eqn:tmeetmix}
\P[T(Z_t,Z_t')>\beta\tmixs]&\le\left(1-\frac{1+O(n^{-(\alpha-1)})}{n}\right)^{\beta},
\end{align}
which is $O(1/n^{2+\delta})$ for $\beta=(2+\delta)n\log n$. 
In general using this estimate will probably lead to a non-optimal bound in Lemma~\ref{lem:dec} and Theorem~\ref{thm:lowerbound} below, but we do not expect that the relation between $\tmeets$ and $\tmixs$ can be greatly improved in the general case due to comparison with hitting times. 
Based on results for similar, but reversible, random walks \cite{aldous1991meeting,aldous-fill-2014,kanade2023coalescence}, we expect that $\tmeets$ can be bounded in terms of the first hitting times. This could give better bounds than \eqref{eqn:tmeetmix}, particularly for slower mixing classical walks,
but for very fast mixing walks, the hitting times can be much larger than mixing times.

\end{enumerate}
\end{ex}

\subsection{Main lower bound results and consequences}\label{subsec:lower-main}

\begin{lem}[decoupling]\label{lem:dec}
Let $U=\prod_{i=d}^1U^{(i)}$ with each independent step $U^{(i)}$ a product of $M$ independent layers, $U^{(i)}\vsim\prod_{j=M}^1\ui{j}$, as introduced in Section~\ref{subsec:bosonrw}.
Let $\tmixs=\tmix(o(1/n))$ for some fixed $o(1/n)$ function be a mixing time for both the associated classical walk and its reversal\footnote{If we use Corollary~\ref{cor:mixsrw}, the mixing time bounds are the same since we compare both to the symmetrized walk. Also recall that $\tmixs\in\N$, while $\tmeets\in\frac{1}{M}\N$.}, and let $\tmeets=\tmeet(o(1/n^2))$ for some fixed $o(1/n^2)$ function be a high-probability meeting time for both the classical walk and its reversal.
Then for depth $d\ge\tmeets+ \tmixs$, 
\begin{align}
\E[|U_{\alpha x}|^2|U_{\alpha y}|^2]&\ge\frac{1}{3n^2}(1-o(1)),\label{eqn:alphamix}\\
\E[|U_{x\alpha}|^2|U_{y\alpha}|^2]&\ge\frac{1}{3n^2}(1-o(1)),\label{eqn:alphamix2}
\end{align}
for any $x,y,\alpha\in\{1,\ldots,n\}$, and where the error term $o(1)$ is uniform in $x,y,\alpha$, as $d,n\to\infty$.
\end{lem}

We will give the proof of Lemma~\ref{lem:dec} in Section~\ref{subsec:decoupling}.
The $o(1)$ error terms in \eqref{eqn:alphamix} and \eqref{eqn:alphamix2} can be made explicit by choosing explicit functions in $\tmixs$ and $\tmeets$.
Lemma~\ref{lem:dec} will be used to show the following theorem. To state part (ii) of the theorem, recall that jointly distributed random variables $(X,Y)$ are said to be a coupling 
of two distributions $\mu$ and $\nu$ if $X\dsim\mu$ and $Y\dsim\nu$.
\begin{thm}[lower bounds]\label{thm:lowerbound}
Let $U=\prod_{i=d}^1 U^{(i)}$ be as in Lemma~\ref{lem:dec}.
\begin{enumerate}[(i)]
\item For depth $d\ge \tmeets+\tmixs$  
and subsystem $\Gamma\subset\{1,\ldots,n\}$,
\begin{align}
\E S_2(U)&\ge \frac{|\Gamma|(n-|\Gamma|)}{3n}\log(\cosh(2s))(1-o(1)).
\end{align}

\item Let $W_{\dist,p}(\mu,\nu):=\inf\{ (\E[\dist(X,Y)^p])^{1/p}:(X,Y)\text{ is a coupling of }(\mu,\nu)\}$ denote the $L^p$ Wasserstein distance between probability measures $\mu$ and $\nu$ on a metric space $(M,\dist)$. Let $\mathcal{H}_n$ denote unit-normalized Haar measure on $\U(n)$. Consider the Hilbert--Schmidt norm $\|\cdot\|_\hs$ on $\U(n)$ and $p=2$, and let 
\begin{align}\label{eqn:tau2def}
\tau_2(\varepsilon):=\inf\{d\in\N:W_{\hs,2}(\mu K^d,\mathcal H_n)\le\varepsilon\text{ for all probability measures $\mu$ on $\U(n)$}\},
\end{align}
where $K$ is the transition kernel for a single step $U^{(i)}$. 
Then there is a numerical constant $C$ so that 
\begin{align}\label{eqn:tau2}
\begin{aligned}
\tau_2(\varepsilon)&\le C(\tmeets+\tmixs)n\log(n/\varepsilon).
\end{aligned}
\end{align}
\end{enumerate}
\end{thm}

\begin{rmk}\phantomsection\label{rmk:lower}
\begin{enumerate}[itemsep=1mm]
\item For the 1D brickwall circuit and $|\Gamma|=k\propto n$, the depth in part (i) can be taken to be $O(n^2\log^2n)$ since $\tmeets$ and $\tmixs$ can be taken $O(n^2\log^2n)$ by Examples~\ref{ex:bmix} and \ref{ex:meet}. This bound is sharp up to logarithmic factors due to the upper bound in Theorem~\ref{thm:up}(ii).
The bound \eqref{eqn:tau2} in part (ii) for the brickwall circuit can be taken to be $\tau_2(\varepsilon)=O(n^3\log^2 n\log(n/\varepsilon))$. We do not know if this is near-sharp; however we do expect that $\tau_2(\varepsilon)$, which involves closeness to Haar measure, could possibly occur at a later order time than when entanglement simply reaches $\Theta(k)\log\cosh(2s)$.

\item While the depth condition is sharp up to logarithmic factors for the 1D brickwall circuit discussed above, this is not the case in general. For fast mixing graphs, the meeting time $\tmeets$ may be much larger than the mixing time $\tmixs$. We expect the inclusion of the meeting time $\tmeets$ is however an artifact of the proof method.
In the subsequent \cite{shou2026log}, we use a different method, which does not use the relation to the classical random walk through the circuit, to prove that the hypercube geometry (which has growing degree $\log_2n$) generates large average-case subsystem entanglement in all subsystems in circuit depth $\Theta(\log n)$.
In later added Appendix~\ref{sec:lower2}, we consider circuit geometries corresponding to families of bounded degree graphs, and give a depth condition which depends on the classical graph spectral gap (related to the mixing time $\tmixs$) and not the mixing time.
\end{enumerate}
\end{rmk}

We will give the proof of Theorem~\ref{thm:lowerbound} in Section~\ref{subsec:proof-lower}.
Theorem~\ref{thm:lowerbound}(ii) gives a bound on when $\E S_2(U)$ is close to the Haar value by comparing to the $L^1$ Wasserstein distance $W_{\hs,1}(\mu,\nu)$, which is bounded above by $W_{\hs,2}(\mu,\nu)$ and has the dual characterization
\begin{align}\label{eqn:w1}
W_{\hs,1}(\mu,\nu)&=\sup\left\{\int_{\U(n)} f\,d(\mu-\nu):f:\U(n)\to\R\text{ is 1-Lipschitz w.r.t. } \|\cdot\|_\hs\right\}.
\end{align}
The Lipschitz constant for $S_2(U)$ with respect to the Hilbert--Schmidt or Frobenius norm is $L\le 2\sqrt{k}\sinh^2(2s)$ as calculated in Lemma~\ref{lem:lipschitz}. 
Letting $\E_{\mathcal U\dsim\mathrm{Haar}(n)}S_2(\mathcal U)$ denote the Haar value for the R\'enyi-2 entropy, we then have
\begin{cor}[closeness to Haar values]\label{cor:haar}
Let $U=\prod_{i=d}^1U^{(i)}$ as in Lemma~\ref{lem:dec} or Theorem~\ref{thm:lowerbound}.
For $d\ge \tau_2(\varepsilon)$ as defined in \eqref{eqn:tau2},
\begin{align}
&|\E S_2(U)-\E_{\mathcal U\dsim\mathrm{Haar}(n)}S_2(\mathcal U)| \le 2\sqrt{k}\sinh^2(2s)\varepsilon,\label{eqn:es}\\
\text{and}\quad&\P[|S_2(U)-\E_{\mathcal U\dsim\mathrm{Haar}(n)}S_2(\mathcal U)|>x]\le \frac{8k^{3/2}\varepsilon\log(\cosh(2s))\sinh^2(2s)+48\sinh^4(2s)\frac{k}{n}}{x^2}.\label{eqn:ps}
\end{align}
Taking e.g.~$\varepsilon=1/n$, then there is a constant $C$ so that for depth
\begin{align}\label{eqn:s2-typicality}
d\ge C(\tmeets+\tmixs)n\log n,
\end{align}
there is weak typicality of $S_2(U)$ with respect to Haar averages.\footnote{that is, $\lim\limits_{n\to\infty}\P\left[\left|\frac{S_2(U)}{\E_{\mathcal U\dsim\mathrm{Haar}(n)} S_2(\mathcal U)}-1\right|<\epsilon\right]=1$ for any $\epsilon>0$}

We also have estimates for convergence of moments of $U$ to Haar moments. Letting $f_m(U):=\prod_{j=1}^mU^\pm_{a(j)b(j)}$, where $U^\pm$ indicates the term may be $U$ or $\bar U$ and $a(j),b(j)\in\{1,\ldots,n\}$, then
\begin{align}\label{eqn:fm}
|\E f_m(U)-\E_{\mathcal U\dsim\mathrm{Haar}(n)}f_m(\mathcal U)|&\le \varepsilon,
\end{align}
for depth $d\ge C(\tmeets+\tmixs)n\log(nm/\varepsilon)$ for a constant $C$.
\end{cor}
\begin{proof}[Proof of Corollary~\ref{cor:haar}]
Equation~\eqref{eqn:es} follows immediately from \eqref{eqn:w1} and the bound on the Lipschitz constant. For \eqref{eqn:ps}, let $S_{2,H}:=\E_{\mathcal U\dsim\mathrm{Haar}(n)}S_2(\mathcal U)$, and define $F(U):=|S_2(U)-S_{2,H}|^2$. Then $F$ is $8k^{3/2}\log(\cosh(2s))\sinh^2(2s)$-Lipschitz, as
\begin{align*}
|F(U)-F(V)|&\le|S_2(U)^2-S_2(V)^2|+2S_{2,H}|S_2(U)-S_2(V)|\\
&\le 4\|S_2\|_\infty|S_2(U)-S_2(V)|\\
&\le 8k^{3/2}\log(\cosh(2s))\sinh^2(2s)\|U-V\|_\hs. \numberthis
\end{align*}
Then for $d\ge \tau_2(\varepsilon)$ and $U=\prod_{i=d}^1U^{(i)}$, using \eqref{eqn:w1} we have 
\begin{align}
|\E F(U)-\E_{\mathcal U\dsim\mathrm{Haar}(n)}F(\mathcal U)|&\le \varepsilon8k^{3/2}\log(\cosh(2s))\sinh^2(2s).
\end{align}
Noting that $\E_{\mathcal U\dsim\mathrm{Haar}(n)}F(\mathcal U)=\Var_{\mathcal U\dsim\mathrm{Haar}(n)}(S_2(\mathcal U))$, we then obtain
\begin{align*}
\P[|S_2(U)-\E_{\mathcal U\dsim\mathrm{Haar}(n)}S_2(\mathcal U)|>x]&\le\frac{\E|S_2(U)-S_{2,H}|^2}{x^2}\\
&\le \frac{\varepsilon8k^{3/2}\log(\cosh(2s))\sinh^2(2s)+\Var_{\mathcal U\dsim\mathrm{Haar}(n)}(S_2(\mathcal U))}{x^2},
\end{align*}
and \eqref{eqn:ps} follows from using the variance proxy $\sigma^2=6L^2/n$ with $L=2\sqrt{k}\sinh^2(2s)$ for $S_2(\mathcal U)$ with $\mathcal U$ Haar-random, as expressed in \eqref{eqn:concentration}. Weak typicality follows from using that the Haar-random expectation values of $S_2$ scale as $\Theta(k)$ \cite{iosue2023page}.

The moments bound \eqref{eqn:fm} follows from checking that the function $f_m$ is $m$-Lipschitz with respect to the Hilbert--Schmidt norm.
\end{proof}

For the brickwall geometry, \eqref{eqn:s2-typicality} becomes $d\ge Cn^3\log^3n$ for some constant $C>0$. 
For the moments of $U$, since $m$th moments of an $n\times n$ Haar random unitary are of order $O(n^{-m/2})$, we will typically want to take $\varepsilon=o(n^{-m/2})$. For brickwall unitaries we then have $o(n^{-m/2})$-closeness of $m$th moments to Haar moments by depth $d\ge\Omega(m n^3\log^3n)$.

\subsection{Decoupling proof}\label{subsec:decoupling}
\begin{proof}[Proof of Lemma~\ref{lem:dec}]
For this proof, we will consider the discrete-time expanded classical random walk $\spl{t}$, indexed by $t\in\frac{1}{M}\N$, whose transition matrix $\pt{t}$ for the step at time $t\in \frac{j}{M}+\N$ is given by $\E|\ui{j}|^2$. 
We also introduce notation $\udm{r}{d}$ to denote the depth $d$ matrix $U$ but with the first $r-1$ layers in the circuit drawing (which are the last $r-1$ of the factors $\ui{j}$ in the product for $U$) removed, e.g.~for $r-1< M$,
\begin{align}\label{eqn:Usd}
\udm{r}{d}&= U^{(d)}\cdots U^{(2)}\prod_{j=M}^{r}\ui{j},
\end{align}
so that the first step applied in $\udm{r}{d}$ is $\ui{r}$.
From Lemma~\ref{lem:prodexp}, the corresponding classical walk for $\udm{r}{d}$, which we denote by $\splm{r}{t}$, is the same as for $U$, except that it starts by walking through the $r$th layer in the circuit instead of the first layer. The first nonzero time $\frac{1}{M}$ for the walk $\splm{r}{t}$ corresponds to time $\frac{r}{M}$ for $\spl{t}$. We then have
\begin{align}
\P_\alpha[\splm{r}{d-\frac{r-1}{M}}=\ell]&=\E|\udm{r}{d}_{\ell\alpha}|^2.
\end{align}

To estimate $\E[|U_{\alpha x}|^2|U_{\alpha y}|^2]$, or more generally $\E[|\udm{r}{d}_{\alpha x}|^2|\udm{r}{d}_{\alpha y}|^2]$, we start by pulling out the right-most matrix $\ui{r}$ in \eqref{eqn:Usd}, which is independent of the other factors, to get
\begin{align*}
\E[|\udm{r}{d}_{\alpha x}|^2|\udm{r}{d}_{\alpha y}|^2]&=\E\sum_{\ell,m,\ell',m'=1}^n\udm{r+1}{d}_{\alpha \ell}\ui{r}_{\ell x}\,\budm{r+1}{d}_{\alpha \ell'} \bui{r}_{\ell'x}\udm{r+1}{d}_{\alpha m}\ui{r}_{m y}\,\budm{r+1}{d}_{\alpha m'}\bui{r}_{m'y}\\
&=\sum_{\ell, m, \ell',m'=1}^n\E[\udm{r+1}{d}_{\alpha \ell}\udm{r+1}{d}_{\alpha m}\budm{r+1}{d}_{\alpha \ell'}\budm{r+1}{d}_{\alpha m'}]\E[\ui{r}_{\ell x}\ui{r}_{m y}\bui{r}_{\ell'x}\bui{r}_{m'y}].\numberthis\label{eqn:dmhalf}
\end{align*}
To evaluate the expectation value over the $\ui{r}$ terms, recall that each matrix $\ui{r}$ is constructed from $1\times 1$ or $2\times 2$ independent Haar random blocks spread out over singlets or pairs of indices according to a permutation $\pi_r\in\fakepair(n)$. 
In a slight abuse of notation, we will let $\pii{r}$ denote both the associated partition in $\fakepair(n)$ and the function $\pii{r}(x)$ which maps $x$ to the set $\{x,\pi_r(x)\}$.
Then we obtain from second and fourth moment calculations for the $1\times1$ and $2\times 2$ Haar unitaries (which are all independent),
\begin{align}
\E[\ui{r}_{\ell x}\ui{r}_{m y}\bui{r}_{\ell' x}\bui{r}_{m'y}]&=\oneb_{\substack{\ell,\ell'\in\pii{r}(x)\\m,m'\in\pii{r}(y)}}\begin{cases}
\delta_{\ell\ell'mm'},&x=y\text{ and }|\pii{r}(x)|=1\\
\frac{1}{6}(\delta_{\ell\ell'}\delta_{mm'}+\delta_{\ell m'}\delta_{m \ell'}),&x=y\text{ and }|\pii{r}(x)|=2\\
\frac{1}{3}\delta_{\ell\ell'}\delta_{mm'}-\frac{1}{6}\delta_{\ell m'}\delta_{m \ell'},&\{x,y\}\in\pii{r},\;x\ne y\\
2^{-(\delta_{|\pii{r}(x)|,2}+\delta_{|\pii{r}(y)|,2})}\delta_{\ell\ell'}\delta_{mm'},&\{x,y\}\not\in\pii{r}
\end{cases}.
\end{align}
This implies
\begin{multline}\label{eqn:22cases}
\E[|\udm{r}{d}_{\alpha x}|^2|\udm{r}{d}_{\alpha y}|^2]
\\=\sum_{\substack{\ell_1\in\pii{r}(x)\\m_1\in\pii{r}(y)}}\E[|\udm{r+1}{d}_{\alpha\ell_1}|^2|\udm{r+1}{d}_{\alpha m_1}|^2]
\times\begin{cases}
1,&x=y\text{ and }|\pii{r}(x)|=1\\
\frac{1}{3},&x=y\text{ and }|\pii{r}(x)|=2\\
\frac{1}{6},&\{x,y\}\in\pii{r},\;x\ne y\\
2^{-(\delta_{|\pii{r}(x)|,2}+\delta_{|\pii{r}(y)|,2})},&\{x,y\}\not\in\pii{r}\\
\end{cases}.
\end{multline}

In the first three cases, when $\{x,y\}\in\pii{r}$ so that $\pii{r}(x)=\pii{r}(y)$, the sum over $\ell_1$ and $m_1$ is bounded below by the sum over just the terms where $\ell_1=m_1$. 
Since for $|\pii{r}(x)|=2$, there are 2 possible values for $\ell_1=m_1$, we obtain in any of the first three cases of \eqref{eqn:22cases},
\begin{align*}
\E[|\udm{r}{d}_{\alpha x}|^2|\udm{r}{d}_{\alpha y}|^2]&\ge \frac{1}{3}\min_{\ell_1\in\pii{r}(x)}\E[|\udm{r+1}{d}_{\alpha\ell_1}|^4]\\
&\ge \frac{1}{3}\min_{\ell_1\in\pii{r}(x)}(\E|\udm{r+1}{d}_{\alpha\ell_1}|^2)^2
\ge
\min_\ell \frac{1}{3}\P_{\ell}[\splm{r+1}{d-\frac{r}{M}}=\alpha]^2,\numberthis\label{eqn:4bound}
\end{align*}
where we used Jensen's inequality in the second inequality.
If the total walk time $d-\frac{r}{M}$ in the walk $\splm{r+1}{d-\frac{r}{M}}$ is larger than the classical mixing time $\tmixs$, then this will be $\approx\frac{1}{3n^2}$.
More precisely, recalling that $\tmixs:=\tmix(o(1/n))\in\N$ and using Lemma~\ref{lem:aftermixing}, for $t\ge\tmixs$, 
\begin{align*}
\bigg|\P_\ell[\spl{t}=\alpha]^2-\frac{1}{n^2}\bigg|\le \bigg|\P_\ell[\spl{t}=\alpha]-\frac{1}{n}\bigg|\left(\bigg|\P_\ell[\spl{t}=\alpha]-\frac{1}{n}\bigg|+\frac{2}{n}\right)
&= o\left(\frac{1}{n^2}\right).\numberthis\label{eqn:pmixbound}
\end{align*}
In the last case of \eqref{eqn:22cases} where $\{x,y\}\not\in\pii{r}$, we iterate, pulling out the next layer $\ui{r+1}$ from $\udm{r+1}{d}$. This produces another summation symbol like in \eqref{eqn:22cases}, this time over indices $\ell_2\in\pii{r+1}(\ell_1)$ and $m_2\in\pii{r+1}(m_1)$. 

Starting from $\E[|U_{\alpha x}|^2|U_{\alpha y}|^2]$ which corresponds to \eqref{eqn:dmhalf} with $r=1$, we can associate the iteration process with steps in a pair of random walks $Z_t$ and $Z_t'$ (Figure~\ref{fig:iterate}). One path in this walk starts at $x$, then moves to an $\ell_1\in\pii{1}(x)$, then to an $\ell_2\in\pii{2}(\ell_1)$, and so on. The other path starts at $y$, then moves to an $m_1\in\pii{1}(y)$, then to an $m_2\in\pii{2}(m_1)$, and so on. Each term $\E[|\udm{r+1}{d}_{\alpha\ell_r}|^2|\udm{r+1}{d}_{\alpha m_r}|^2]$ in the resulting sum after $r$ steps corresponds to one of these pairs of paths, going from $x\to\ell_r$ and $y\to m_r$ in $r$ steps.

We continue iterating each term like $\E[|\udm{r+1}{d}_{\alpha\ell_r}|^2|\udm{r+1}{d}_{\alpha m_r}|^2]$ in the sum until we can apply one of the first three cases of \eqref{eqn:22cases}. At that point, we apply \eqref{eqn:4bound} and stop iterating that branch. 
In summary, this process considers all possible pairs of classical walks $(\sigma_1,\sigma_2)$, where each walk is described by the classical walk $\spl{t}$, with $\sigma_1$ starting at $x$, $\sigma_2$ starting and $y$, and both ending at $\alpha$ at time $d$ (Figure~\ref{fig:iterate}). As soon as $\sigma_1$ and $\sigma_2$ coincide enough to use one of the first three cases of \eqref{eqn:22cases}, we apply \eqref{eqn:4bound} at that time and stop iterating that branch.

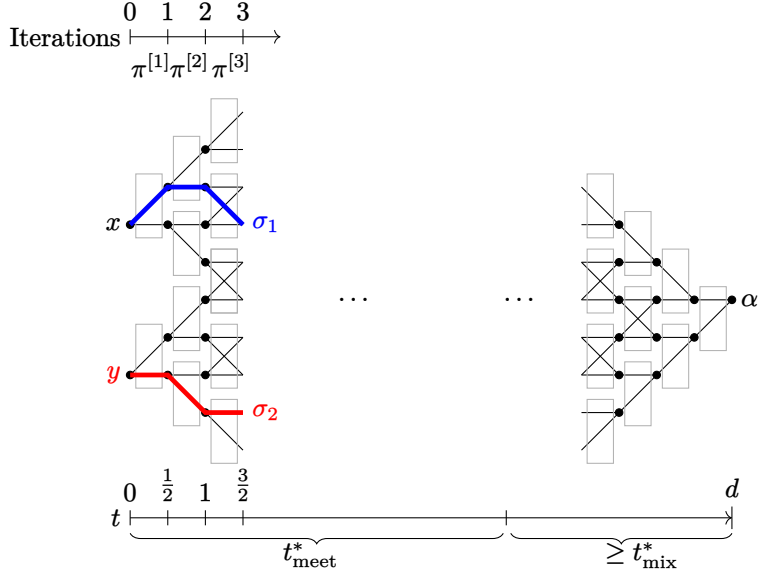
\begin{figure}[htb]
\begin{tikzpicture}[scale=2,xscale=-1]
\def\height{1mm};
\begin{scope}[yshift=-6.5mm] % bottom axes and labels
% t axis
\def\theight{2mm}
\node[left] at (4,\theight) {$t$};
\draw [->](4,\theight)--++(-4,0);
\foreach \t in {0,...,3}{\draw[yshift=\theight] (4-.25*\t,-.5mm)--++(0,1mm);}
\node[above] at (4-.25,\theight+.5mm) {$\frac{1}{2}$};
\node[above] at (4-.5,\theight+.5mm) {$1$};
\node[above] at (4-.75,\theight+.5mm) {$\frac{3}{2}$};
\node[above] at (4,\theight+.5mm) {$0$};
\draw[yshift=\theight] (0,-.75mm)--++(0,1.5mm) node[above] {$d$};
\draw[yshift=\theight] (1.5,-.5mm)--++(0,1mm);
\draw[decoration={brace,raise=3pt,aspect=.5,amplitude=4pt},decorate,yshift=-.5mm] (1.52,\theight)--(4,\theight);
\node[below] at (2.8,\theight-1mm) {$\tmeets$};
\draw[decoration={brace,raise=3pt,aspect=.5,amplitude=4pt},decorate,yshift=-.5mm] (0,\theight)--(1.47,\theight);
\node[below] at (.6,\theight-1mm) {$\ge\tmixs$};
\end{scope}

% options for tree
\def\mrad{1mm} % vertex radius
\def\xcolor{blue} % x color
\def\ycolor{red} % y color
\def\bcolor{gray!60} % brick color
\def\brickmargin{.35cm} % size of bricks
\def\pathwidth{1.8} % line width for x, y paths
% tree parameters: depth, brick color, shift (w/ units), depth
\newcommand{\tree}[4]{\begin{scope}[rotate=270,scale=.25,xshift=-1cm,yshift=-1cm]
\draw (1,1)--++(-1,-1);
\draw (1,1)--++(0,-1);
\draw[fill=black] (1,1) circle (\mrad);
\foreach \height in {0,...,#1}{
\foreach \lr in {0,...,\height}{
% Left step edges
\draw (2*\lr-\height,-\height)--++(-1,-1);
\draw (2*\lr-\height,-\height)--++(0,-1);
% right step edges
\draw (2*\lr-\height+1,-\height)--++(1,-1);
\draw (2*\lr-\height+1,-\height)--++(0,-1);
% draw nodes
\draw[fill=black] (2*\lr-\height,-\height) circle (\mrad);
\draw[fill=black] (2*\lr-\height+1,-\height) circle (\mrad);
}
}
% draw bricks
\foreach \height in {0,...,#4}{
\foreach \lr in {0,...,\height}{
\draw[draw=#2,fill=none,xshift=-\brickmargin,yshift=-\brickmargin+#3] (2*\lr-\height, -\height) rectangle ++({1cm+2*\brickmargin},{2*\brickmargin});
}}
\end{scope}}

\begin{scope}[xshift=4cm,yshift=1.5cm]
\tree{1}{\bcolor}{0.5cm}{2}
\draw[line width=\pathwidth,color=\xcolor] (0,0)--(-.25,.25)--(-.5,.25)--(-.75,0);
\node[right,color=\xcolor] at (-.75,0) {$\sigma_1$};
\node[left] at (0,0) {$x$};
\end{scope}

\begin{scope}[xshift=4cm,yshift=.5cm]
\tree{1}{\bcolor}{0.5cm}{2}
\draw[line width=\pathwidth,color=\ycolor] (0,0)--(-.25,0)--(-.5,-.25)--(-.75,-.25);
\node[right,color=\ycolor] at (-.75,-.25) {$\sigma_2$};
\node[left,\ycolor] at (0,0) {$y$};
\end{scope}

\node at (2.5,1) {$\cdots$};
\node at (1.4,1) {$\cdots$};
% alpha branch
\begin{scope}[yshift=1cm,rotate=180] 
\tree{2}{\bcolor}{0.5cm}{3}
\end{scope}
\node[right] at (0,1) {$\alpha$};

% draw iteration depth axis
\def\height{2.75cm}
\draw[->] (4,\height)--++(-1,0);
\node[left] at (4,\height) {Iterations};
\foreach \t in {0,...,3}{
\draw[yshift=\height] (4-.25*\t,-.5mm)--++(0,1mm) node[above] {$\t$};
}
\begin{scope}[xshift=0cm]
\node[xshift=1.8mm,below] at (4-.05,\height-.5mm) {$\pii{1}$};
\node[xshift=2.5mm,below] at (4-.27,\height-.5mm) {$\pii{2}$};
\node[xshift=2.5mm,below] at (4-.55,\height-.5mm) {$\pii{3}$};
\end{scope}
\end{tikzpicture}
\caption{Illustration for brickwall iteration. 
We start by peeling off unitaries $\ui{1},\ui{2},\ldots$ from the right side of $U$, which corresponds to advancing a pair of random walks starting at $x$ and $y$ in the left of the above circuit illustration.
We need the resulting random walk paths $\sigma_1$ and $\sigma_2$ to $\pi$-meet in order to apply \eqref{eqn:4bound} and stop the iteration. With high probability over the choice of paths, this happens by time $\tmeets$. By assumption we took $d$ large enough so that $\udmnom{\tmixs}{d}$ still has depth at least the classical mixing time $\tmixs$. This allows us to apply the classical mixing property \eqref{eqn:pmixbound}, and to handle the conditioning that the walks end at $\alpha$: if we still have at least time $\tmixs$ to go, then we know the probability is $\frac{1+o(1)}{n}$ that a random walk, regardless of its current location, ends at $\alpha$.}\label{fig:iterate}
\end{figure}

Throughout this process, we need to keep track of the fourth moment factors from the right side of \eqref{eqn:22cases} in each iteration. We stop the iteration as soon as we reach any of the first three cases, so we only need to consider the last case of \eqref{eqn:22cases}. In this case, the iteration of a single term $\E[|\udm{j+1}{d}_{\alpha\ell_{j}}|^2|\udm{j+1}{d}_{\alpha m_{j}}|^2]$  generates 
\begin{align}
|\pii{j}(\ell_{j}) \times \pii{j}(m_{j})| =2^{(\delta_{|\pii{j}(\ell_{j})|,2}+\delta_{|\pii{j}(m_{j})|,2})}
\end{align}
new terms.
This number is conveniently exactly the reciprocal of the fourth moment factor produced in the right side of \eqref{eqn:22cases}.
If we ignore the condition that $\sigma_1$ and $\sigma_2$ end at $\alpha$ for now, this fourth moment factor thus gives the correct normalization to view the iteration process as a pair of independent random walks according to $\spl{t}$, as we explain in more detail below. 
In what follows, we will use the assumption $d\ge\tmeets+\tmixs$ to show we can effectively ignore the requirement that $\sigma_1$ and $\sigma_2$ end at $\alpha$.

Consider a pair of paths $(\sigma_1,\sigma_2)$ such that $(\sigma_1)_0=x$, $(\sigma_2)_0=y$, and $(\sigma_1)_d=(\sigma_2)_d=\alpha$. 
Let $T(\sigma_1,\sigma_2):=\min\{t\in\frac{1}{M}\N_{>0}:\{(\sigma_1)_t,(\sigma_2)_t\}\in\pii{Mt}\}$ be the first time the paths $\pi$-meet, equivalently the first time the paths have just passed through the same beamsplitter/phaseshifter. 
If we stop the iteration of a pair of paths $(\sigma_1,\sigma_2)$ at time $T(\sigma_1,\sigma_2)$ to apply the bound \eqref{eqn:4bound}, the factor due to the $2^{-(\delta_{|\pii{r}(x)|,2}+\delta_{|\pii{r}(y)|,2})}$ terms is exactly
$\P[(Z_t,Z_t')=(\sigma_1,\sigma_2)\,\forall\,0\le t\le T(\sigma_1,\sigma_2)]$ where $Z_t,Z_t'$ are independent copies of the random walk, with $Z_t$ starting at $x$ and $Z_t'$ at $y$, but with no conditioning on the ending value at time $d$.
To relate this probability to independent walks $\tilde Z_t,\tilde Z_t'$ which are conditioned to end at $\alpha$, consider $d-T\ge\tmixs$. Then letting the subscript in $\Pxy$ indicate starting point $(x,y)$ for a pair of random walks, write (recall $t\in\frac{1}{M}\N$),
\begin{align*}\label{eqn:condalpha}
\Pxy&[(\tilde Z_t,\tilde Z_t')=(\sigma_1,\sigma_2)\,\forall\,0\le t\le T]\\
&=\frac{\Pxy[(Z_t,Z_t')=(\sigma_1,\sigma_2)\,\forall\,0\le t\le T,\text{ and }Z_d=Z_d'=\alpha]}{\Pxy[Z_d=Z_d'=\alpha]}\quad\text{[by conditioning]}\\
&=\frac{\Pxy[Z_d=Z_d'=\alpha|Z_T=(\sigma_1)_T,Z_T'=(\sigma_2)_T]\Pxy[(Z_t,Z_t')=(\sigma_1,\sigma_2)\,\forall\,0\le t\le T]}{\Pxy[Z_d=Z_d'=\alpha]}\quad\text{[Markov property]}\\
&=\frac{\frac{1+o(1)}{n^2}\Pxy[(Z_t,Z_t')=(\sigma_1,\sigma_2)\,\forall\,0\le t\le T]}{\frac{1+o(1)}{n^2}}\quad\text{[for $d-T\ge\tmixs = \tmix(o(1/n))$]}\\
&=\Pxy[(Z_t,Z_t')=(\sigma_1,\sigma_2)\;\;\forall\,0\le t\le T](1+o(1)).\numberthis
\end{align*}
To apply the Markov property in the third line above, we noted that the individual steps of the walks $Z_t,Z_t'$ only depend on $t\mod 1$ and the current location of the walk, but not on previous locations of the walk.

Next, to handle the counting in the iteration process, define equivalence classes $[(\sigma_1,\sigma_2)]\in\mathcal E$ over the set of all pairs of paths $(\sigma_1,\sigma_2)$ which start at $(x,y)$ and end at $(\alpha,\alpha)$ as follows. Two pairs $(\sigma_1,\sigma_2)$ and $(\sigma_1',\sigma_2')$ are in the same equivalence class if they have the same first $\pi$-meeting time $T$ and are equal up to that time. 
We apply the iteration procedure \eqref{eqn:22cases} until we can use the bound \eqref{eqn:4bound}, which occurs when $\frac{r}{M}=T(\sigma_1,\sigma_2)$. 
Since we stop the iteration process at $T(\sigma_1,\sigma_2)$, we only ever consider a single representative $[(\sigma_1,\sigma_2)]$ from its equivalence class.
Using \eqref{eqn:condalpha} and \eqref{eqn:pmixbound}, we then obtain for $d\ge T+\tmixs$,
\begin{align*}
\E[&|U_{\alpha x}|^2|U_{\alpha y}|^2]\\
&\ge\sum_{[(\sigma_1,\sigma_2)]\in\mathcal E}\oneb_{T(\sigma_1,\sigma_2)\le d-\tmixs}\Pxy[(Z_t,Z_t')=(\sigma_1,\sigma_2)\,\forall\,0\le t\le T]\frac{1}{3}\min_\ell\P_\ell[\splmnom{T+\frac{1}{M}}{d-T(\sigma_1,\sigma_2)
,r}=\alpha]^2\\
&\ge\sum_{[(\sigma_1,\sigma_2)]\in\mathcal E}\oneb_{T(\sigma_1,\sigma_2)\le d-\tmixs}\Pxy[(\tilde Z_t,\tilde Z_t')=(\sigma_1,\sigma_2)\,\forall\,0\le t\le T]\frac{1-o(1)}{3n^2}. \numberthis
\end{align*}
To convert the sum over equivalence classes to one over pairs $(\sigma_1,\sigma_2)$, note that for any equivalence class $[(\sigma_1,\sigma_2)]$,
\begin{equation}
\P[(\tilde Z_t,\tilde Z_t')=(\sigma_1,\sigma_2)\,\forall\,0\le t\le T]=
    \sum_{(\sigma_1',\sigma_2')\in[(\sigma_1,\sigma_2)]} 
    \P[(\tilde Z_t,\tilde Z_t')=(\sigma_1',\sigma_2')\,\forall\,0\le t\le d],
\end{equation}
by using that $\oneb_{[(\sigma_1,\sigma_2)]}=\sum_{(\sigma_1',\sigma_2')\in[(\sigma_1,\sigma_2)]}\oneb_{(\sigma_1',\sigma_2')}$.
Applying this, we can sum over both the equivalence classes and the elements of the equivalence classes to get
\begin{align*}
\E[|U_{\alpha x}|^2|U_{\alpha y}|^2]&\ge \sum_{\substack{\sigma_1:x\to\alpha\\\sigma_2:y\to\alpha}}\oneb_{T(\sigma_1,\sigma_2)\le d-\tmixs}\Pxy[(\tilde Z_t,\tilde Z_t')=(\sigma_1,\sigma_2)\,\forall\,0\le t\le d]\frac{1-o(1)}{3n^2}\\
&= \Pxy[T(\tilde Z_t,\tilde Z_t')\le d-\tmixs]\frac{1-o(1)}{3n^2}.\numberthis\label{eqn:decp}
\end{align*}

It remains to estimate $\P[T(\tilde Z_t,\tilde Z_t')\le d-\tmixs]$. 
As before, let $Z_t,Z_t'$ be independent random walks starting at $x$ and $y$ respectively, with no end time condition. Then applying the definition of conditional probability followed by Cauchy--Schwarz, we obtain a bound on the complementary event
\begin{align*}
\Pxy[T(\tilde Z_t,\tilde Z_t')> d-\tmixs]
&=\Pxy[T( Z_t,  Z_t')>d-\tmixs| Z_d= Z_d'=\alpha]\\
&\le \frac{\Pxy[T(Z_t,Z_t')>d-\tmixs]^{1/2}}{\Pxy[Z_d=Z_d'=\alpha]^{1/2}}.\numberthis\label{eqn:tbound}
\end{align*}
The denominator is $\frac{1+o(1)}{n}$ for $d\ge \tmixs$.
By definition of $\tmeets=\tmeet(o(1/n^2))$, we have
\begin{align}
\Pxy[T(Z_t,Z_t')>\tmeets]&=o(1/n^2),
\end{align}
so that \eqref{eqn:tbound} is $o(1)$ if $d-\tmixs\ge\tmeets$.
Applying this to \eqref{eqn:decp} we thus obtain for $d\ge\tmeets+\tmixs$,
\begin{align}
\E[|U_{\alpha x}|^2|U_{\alpha y}|^2]&\ge \frac{1-o(1)}{3n^2},
\end{align}
proving \eqref{eqn:alphamix}.

The second inequality \eqref{eqn:alphamix2} follows by considering $U^T$ instead of $U$, since we took $\tmixs$ and $\tmeets$ to work for both the forward and reversed classical random walks.
\end{proof}

\subsection{Proof of Theorem~\ref{thm:lowerbound}}\label{subsec:proof-lower}

\begin{proof}[Proof of Theorem~\ref{thm:lowerbound}(i)]
Let the subsystem be $\Gamma\subset\{1,\ldots,n\}$ with size $|\Gamma|=k$, and let $P_\Gamma:\C^n\to\C^\Gamma$ be the $k\times n$ projection matrix onto the $k$ coordinates in $\Gamma$. Applying Lemma~\ref{lem:dec} to the equation \eqref{eqn:uut-alpha}, for $d\ge\tmeets+\tmixs$ and any $x,y$ we obtain
\begin{align}
\E|\langle x|UU^T|y\rangle|^2&=\E \sum_{\alpha=1}^n|U_{x\alpha}|^2|U_{y\alpha}|^2\ge \frac{1-o(1)}{3n}.
\end{align}
Then for $W_\Gamma=V_\Gamma V_\Gamma^\dagger$ with $V_\Gamma=P_\Gamma UU^T P_\Gamma^T$,
\begin{align*}
\E \Tr W_\Gamma=\E\|V_\Gamma\|_\hs^2&=|\Gamma|-\sum_{x\in\Gamma}\sum_{y\not\in\Gamma}\E|\langle x|UU^T|y\rangle|^2\\
&\le |\Gamma|-\sum_{x\in\Gamma}\sum_{y\not\in\Gamma}\frac{1-o(1)}{3n}\\
&\le |\Gamma|-\frac{k(n-k)}{3n}(1-o(1)).\numberthis
\end{align*}
Since the eigenvalues $\lambda_i$ of $W_\Gamma$ are between 0 and 1, for any $\ell\ge1$,
\begin{align}
\E \Tr W_\Gamma^\ell&=\sum_{i=1}^{|\Gamma|} \E\lambda_i^\ell \le \sum_{i=1}^{|\Gamma|} \E\lambda_i=\E \Tr W_\Gamma.
\end{align}
Then from \eqref{eqn:s2} we get for $d\ge\tmixs+\tmeets$,
\begin{align}
\E S_2(U)&\ge\frac{|\Gamma|(n-|\Gamma|)}{3n}\log(\cosh(2s))(1-o(1)).
\end{align}
\end{proof}

Next we prove Theorem~\ref{thm:lowerbound}(ii). 
This will be an application of the coupling method of \cite{oliveira2009convergence} which is related to the path coupling argument of \cite{bubley1997path}. We note that \cite{oliveira2009convergence}'s coupling result was also applied in the proofs of approximate unitary $t$-designs for qudit circuits in \cite{brandao2016local,haferkamp2022random}.
Section 5.2 of \cite{oliveira2009convergence} proves $L^2$ Wasserstein convergence for Kac's random walk \cite{kac1956foundations} extended to the unitary group $\U(n)$. 
This walk on $\U(n)$ is the walk which chooses two coordinate indices $i,j$ at random, then applies a random $2\times2$ Haar matrix on the subspace $\operatorname{span}(e_i,e_j)$. In terms of linear optical networks, we can think of Kac's random walk as a sequential circuit where in each step, we choose random modes $i,j$, and apply a random beamsplitter plus phaseshifters across those two modes. To apply the coupling method, one needs to show that the walk is locally contracting \emph{on average}. Due to the uniform random choice of $i,j$ in Kac's random walk, this property holds on average as shown in \cite[\S5.2]{oliveira2009convergence}.
However, for a specific non-random beamsplitter geometry, such as the brickwall geometry, this does not work---essentially this is because each mode $i$ can only interact with $M$ fixed other modes in a single step, rather than any mode $j$ as in Kac's walk.
As a result, we will have to pre-run our circuit until we obtain the decoupling property in Lemma~\ref{lem:dec}, and then we will construct a different coupling to show the local contraction. 

\begin{proof}[Proof of Theorem~\ref{thm:lowerbound}(ii)]
Recall we only work with the Hilbert--Schmidt norm. To prove $L^2$ Wasserstein convergence, we need an estimate on the diameter of $\U(n)$, plus a local contraction estimate of the form $\E[\|X-Y\|_\hs^2]\le (1-\kappa)r^2+O(r^3)$ for some coupling $(X,Y)$ of the chain and small distance $r=\|x-y\|_\hs$ between starting points $x$ and $y$.
The diameter $2\sqrt{n}$ for $\U(n)$ with Hilbert--Schmidt norm is immediate, so we just need the local contraction estimate. 
To obtain this estimate, we will need to work with more than just a single layer $U^{(i)}$. We will need to take a product $U=\prod_i U^{(i)}$ until $U$ is well spread out on average, in particular satisfying the decoupling property of Lemma~\ref{lem:dec}. 

Let $Q=\prod_{i=d}^1U^{(i)}$ with $d\ge \tmeets+\tmixs$ so that Lemma~\ref{lem:dec} applies.
A single step of our walk $L$ on $\U(n)$ will be an application of an independent copy of the depth-$(d+1)$ unitary $U:=U^{(d+1)}Q$, for a step $U^{(d+1)}$ independent of $Q$. Letting $P=K^{d+1}$ be the Markov transition kernel of $L$, the chain $L$ is said to be $\xi$-contracting in the $W_{\hs,2}$ metric if for all probability measures $\mu,\nu$,
\begin{align}
W_{\hs,2}(\mu P,\nu P)\le \xi W_{\hs,2}(\mu,\nu),
\end{align}
for some $\xi<1$. So we want to construct a coupling $(X,Y)$ of one step of the chain which satisfies an average $L^2$ local contraction estimate, i.e.~letting $L_z$ be the chain starting at $z\in\U(n)$, we want to construct a coupling $(X,Y)$ of $(L_x,L_y)$ such that $\E[\|X-Y\|_\hs^2]\le (1-\epsilon)r^2+O(r^3)$, where $r=\|x-y\|_\hs$.

Suppose $x,y\in\U(n)$ with $\|x-y\|_\hs=r$. 
Then as in \cite[\S5.2]{oliveira2009convergence}, $\|yx^\dagger-\id-h\|_\hs=O(r^2)$ for $h=\Pi_T(yx^\dagger-\id)$, with $\Pi_T$ the orthogonal projection onto the tangent space $T=T_{\id}(\U(n))=\{h\in M(n,\C):h=-h^\dagger\}$. This also implies $\|h\|_\hs=r+O(r^2)$.
Choose the coupling
\begin{align}
X=U^{(d+1)}V_1QV_2x,\quad Y=Uy,
\end{align}
where $V_2$ will be a diagonal unitary matrix depending only on $x,y$, and $V_1$ will be a diagonal unitary matrix depending only on $x,y,Q$. Since $U^{(d+1)}$ and $Q$ are made up of products of matrices which consist of independent $1\times1$ and $2\times 2$ Haar-random blocks, their distributions are invariant under multiplication by any fixed diagonal unitary matrix. 
We then see that $U^{(d+1)}V_1\dsim U^{(d+1)}$ after conditioning on $Q$ since $V_1$ depends on $Q$ but not on $U^{(d+1)}$, and also that $U^{(d+1)}V_1QV_2\vsim U^{(d+1)}Q=U$.

If $\sum_{j=1}^n|h_{jj}|^2\ge(\frac{1}{3}-\delta)\frac{1}{n}\|h\|_\hs^2$ for some $\delta>0$ to be chosen later, take $V_1=\id$ and write $\|X-Y\|_\hs^2=\|V_2-yx^\dagger\|_\hs^2$, using that the Hilbert--Schmidt norm is invariant under unitary multiplication. Then taking $V_2=\operatorname{diag}(e^{h_{jj}})_{j=1}^n$, which is unitary since $h$ is skew-hermitian so $h_{jj}$ is purely imaginary, we get
\begin{align*}
\|X-Y\|_\hs^2=\|V_2-yx^\dagger\|_\hs^2&=\|\operatorname{diag}(h_{jj})_{j=1}^n-h\|_\hs^2+O(r^3)\\
&=\|h\|_\hs^2-\sum_{j=1}^n|h_{jj}|^2+O(r^3)\\
&\le\left(1-\frac{1-3\delta}{3n}\right)\|h\|_\hs^2+O(r^3).\numberthis\label{eqn:diag1}
\end{align*}

If instead $\sum_{j=1}^n|h_{jj}|^2<(\frac{1}{3}-\delta)\frac{1}{n}\|h\|_\hs^2$, take $V_2=\id$ to write
\begin{align}
\|X-Y\|_\hs^2&=\|V_1-Q(yx^\dagger)Q^\dagger\|_\hs^2.
\end{align}
We took $Q$ to have a depth $d$ so that for most matrix structures of $yx^\dagger$, on average $h':=QhQ^\dagger$ has entries of roughly comparable size:
For any $\alpha$, write
\begin{align*}
\E|h'(\alpha,\alpha)|^2&=\sum_{j,j',k,k'}\E[Q_{\alpha j}h_{jj'}Q^\dagger_{j'\alpha}\bar Q_{\alpha k'}\bar h_{k'k}\bar Q^\dagger_{k\alpha}]\\
&=\sum_{j,j',k,k'}h_{jj'}\bar h_{k'k}\E[Q_{\alpha j}Q_{\alpha k}\bar Q_{\alpha j'}\bar Q_{\alpha k'}]\\
&=\sum_{j,k}(h_{jj}\bar h_{kk}+|h_{jk}|^2-\delta_{jk}|h_{jj}|^2)\E[|Q_{\alpha j}|^2|Q_{\alpha k}|^2],\numberthis
\end{align*}
where to get the third line we used that we must have $j=j'$ and $k=k'$, or $j=k'$ and $k=j'$ for the expectation in the second line to be nonzero, by the same argument used to obtain \eqref{eqn:uut-alpha}. 
Using decoupling Lemma~\ref{lem:dec} with depth $d\ge \tmeets+\tmixs$, then 
\begin{align*}
\E|h'(\alpha,\alpha)|^2&\ge \|h\|_\hs^2\frac{1-o(1)}{3n^2}-\sum_{j\ne k}|h_{jj}||h_{kk}|\E[|Q_{\alpha j}|^2|Q_{\alpha k}|^2]\\
&\ge \|h\|_\hs^2\frac{1-o(1)}{3n^2}-\left(\sum_{j, k}|h_{jj}|^2\E[|Q_{\alpha j}|^2|Q_{\alpha k}|^2]\right)^{1/2}\left(\sum_{j, k}|h_{kk}|^2\E[|Q_{\alpha j}|^2|Q_{\alpha k}|^2]\right)^{1/2}\\
&=\|h\|_\hs^2\frac{1-o(1)}{3n^2}-\sum_{j=1}^n|h_{jj}|^2\E[|Q_{\alpha j}|^2]\\
&\ge \|h\|_\hs^2\frac{\delta-o(1)}{n^2},\numberthis 
\end{align*}
where we have used the assumption $\sum_{j=1}^n|h_{jj}|^2<(\frac{1}{3}-\delta)\frac{1}{n}\|h\|_\hs^2$ and the property $\E[|Q_{\alpha j}|^2]=\P_j[Z_d=\alpha]=\frac{1+o(1)}{n}$ since the depth of $Q$ is $\ge\tmixs$ to obtain the last line.
Choose $V_1=\operatorname{diag}(e^{h'(j,j)})_{j=1}^n$, which is unitary since $h'=QhQ^\dagger$ is skew-hermitian, and which depends only on $x,y$, and $Q$. Then
\begin{align*}
\E\|X-Y\|_\hs^2&=\E\|\operatorname{diag}(h'(j,j))-h'\|_\hs^2+O(r^3)\\
&=\|h'\|_\hs^2-\sum_{j=1}^n\E|h'(j,j)|^2+O(r^3)\\
&\le \|h\|_\hs^2\left(1-\frac{\delta-o(1)}{n}\right)+O(r^3),\numberthis\label{eqn:diag2}
\end{align*}
where the $o(1)$ term is as $n\to\infty$, separate from the $O(r^3)$ term.
Taking $\delta=1/6$ to optimize \eqref{eqn:diag1} and \eqref{eqn:diag2} gives
\begin{align*}
\E\|X-Y\|_\hs^2&\le \left(1-\frac{1-o(1)}{6n}\right)r^2+O(r^3),\numberthis
\end{align*}
so the chain is $\sqrt{1-\frac{1-o(1)}{6n}}$-locally contracting.

Taking depth $d=\lceil\tmeets+\tmixs\rceil$ to use the estimates above, applying \cite[Theorem 3, Corollary 1]{oliveira2009convergence} with the diameter bound $2\sqrt{n}$ then gives
\begin{align}
\tau_2(\varepsilon)\le (\lceil\tmeets+\tmixs\rceil+1)\,\left\lceil 12n\log\left(\frac{2\sqrt{n}}{\varepsilon}\right)(1+o(1))\right\rceil.
\end{align}
\end{proof}

\section{Circuit complexity}\label{sec:complexity}
In this section, we prove Theorem~\ref{thm:complexity} on robust or approximate circuit complexity for random one-dimensional brickwall linear optical networks.

To prove part (i), we first observe that the constructions in \cite{reck1994experimental,clements2016optimal} (see also \cite{hurwitz1897uber}), which give linear optical networks implementing any $n\times n$ unitary matrix using $n(n-1)/2$ beamsplitter-phaseshifter gates, imply the following:
Let $U$ be an $n\times n$ unitary matrix such that $U_{ij}=0$ if $i-j>w$, i.e. $U$ has zeros in the bottom left corner triangle of height and width $n-1-w$. Then $U$ can be exactly expressed using $nw-\frac{w(w+1)}{2}$ nearest-neighbor beamsplitter-phaseshifter gates.
This follows from the constructions in \cite{reck1994experimental,clements2016optimal}, which iteratively use a beamsplitter-phaseshifter gate to zero out an entry in the lower left of $U$, starting from the bottom-left corner and proceeding either along rows or along diagonals to zero out a growing region of entries below the main diagonal. In the above case, we start with a lower-left triangle of zeros, so can skip the gates corresponding to those entries and only zero out the $nw-\frac{w(w+1)}{2}$ nonzero entries below the diagonal.

In the case where $U$ is an $n\times n$ depth-$d$ random 1D brickwall linear optical unitary, the matrix $U$ has band width $\Theta(d)$, but only an effective band width $w=O(\sqrt{d\log d})$ with high probability (cf. Section~\ref{sec:avgcase}). We would like to ignore the small entries of $U$ outside this effective band width $w$, and apply the above observation using \cite{reck1994experimental,clements2016optimal} with only band width $w=O(\sqrt{d\log d})$.
However, due to the small but nonzero entries which we ignore in this procedure, we will need to keep track of their cumulative effects to avoid large accumulation of error bounds.

\begin{proof}[Proof of Theorem~\ref{thm:complexity}(i)]
Let $U$ be an $n\times n$ depth-$d$ random 1D brickwall linear optical unitary. 
By Theorem~\ref{thm:boson-rw} and the large deviation estimate \eqref{eqn:Zdalpha},
\begin{align}
\E|\langle x|U|y\rangle|^2&=\P_y[Z_d=x]\le Ce^{-c_1(x-y)^2/d},
\end{align}
which by Markov's inequality implies
\begin{align}
\P[|\langle x|U|y\rangle|>e^{-c_1(x-y)^2/(4d)}]&\le Ce^{-c_1(x-y)^2/(2d)}.
\end{align}
Then for any $\kappa>0$, there is $C_\kappa$ so that with probability at least $1-Cn^2d^{-\kappa}$, we have
\begin{align}\label{eqn:sqrtband}
|\langle x|U|y\rangle|\le d^{-\kappa}\,\text{ for all }|x-y|>C_\kappa\sqrt{d\log d}. 
\end{align}
Let $\varepsilon:=\sqrt{n}d^{-\kappa}$, chosen so the $\ell^2$ norm of the entries outside the $C_\kappa\sqrt{d\log d}$-band in any row is $\le\varepsilon$ in the event above.
To construct a $\tilde U$ with $O(n\sqrt{d\log d})$ gates which approximates $U$, we will essentially ignore the entries of $U$ outside the $C_\kappa\sqrt{d\log d}$-band, using that these are small by \eqref{eqn:sqrtband}.

To this end, perform the zeroing algorithm from \cite{reck1994experimental}, but only zero out entries below the diagonal of $U$ which are inside the effective $w=C_\kappa\sqrt{d\log d}$-band. We start from the bottom row and left-most entry in the effective band and move to the right, always using the matrix entry one column to the right to zero out the current entry $(r,m)$.
Zeroing out the entry at $(r,m)$ is performed by multiplication by an appropriate complex rotation on the $m,m+1$-coordinate plane, which we will denote by $T_{r,m}$ and which represents a beamsplitter-phaseshifter applied to modes $m$ and $m+1$. This mixes the $m$th and $(m+1)$th columns of the current matrix.
The \cite{reck1994experimental} procedure then gives
\begin{align}\label{eqn:reck}
UT_{n,n-w}\cdots T_{n,n-1}T_{n-1,n-w-1}\cdots T_{n-1,n-2}\cdots T_{3,1}T_{3,2} T_{2,1} =: \tilde{D},
\end{align}
which involves $nw-\frac{w(w+1)}{2}$ nearest-neighbor beamsplitter-phaseshifter gates for $w=C_\kappa\sqrt{d\log d}$. 
(For later use, note that each row $r=n-j$ requires at most $w$ applications of $T_{r,m}$'s, whose applications can only affect columns $r-w$ through $r$.)
Unlike the \cite{reck1994experimental} procedure which produces a diagonal unitary matrix on the right of \eqref{eqn:reck}, the unitary matrix $\tilde D$ is in general not diagonal, since we did not zero out the entries of $U$ outside the $C_\kappa\sqrt{d\log d}$-band, and also some of those entries can mix with entries that were previously zeroed out.
We will however show $\tilde D$ is close to diagonal; in particular, we claim
\begin{align}\label{eqn:close-diag}
|\tilde D_{n-j,n-j}|^2&\ge 1-(j+1)\varepsilon^2,\quad j=0,\ldots,n-1.
\end{align}

To show this, we will show the entries off of the diagonal in each row are small. First we observe that the entries to the left of the diagonal in any row satisfy $\sum_{i=1}^{n-j-1}|\tilde D_{n-j,i}|^2\le\varepsilon^2$: 
Right after we zero out the effective $C_\kappa\sqrt{d\log d}$-band in row $r=n-j$, the entries to the left of the diagonal in row $r$ have total $\ell^2$ norm $\le\varepsilon$, since they are zero in the $w=C_\kappa\sqrt{d\log d}$-band and have total $\ell^2$ norm $\le\varepsilon$ outside the $w$-band (these entries outside the $w$-band have not yet been acted on by any $T_{r',m}$).
Subsequent applications of $T_{r',m}$'s to zero out the $w$-band in higher rows $<r$ only affect columns with index $<r$, so they preserve the $\ell^2$ norm\footnote{note the subsequent $T_{r',m}$ may overwrite previously zeroed out entries, but the total $\ell^2$-norm mass to move around in each row is $\le\varepsilon$}, which is $\le\varepsilon$, of entries to the left of the diagonal in row $r$.

Using row normalization and the above observation, we see that for any $\ell=0,\ldots,n-1$,
\begin{align}
|\tilde D_{n-\ell,n-\ell}|^2&\ge 1-\varepsilon^2-\sum_{i=1}^\ell|\tilde D_{n-\ell,n-\ell+i}|^2.
\label{eqn:dmatrix-diag}
\end{align}
For row $r=n-j$, we can estimate (as illustrated in Figure~\ref{fig:diagcancel})
\begin{align*}
\sum_{i=1}^j|\tilde D_{n-j,n-j+i}|^2&\le \sum_{i=1}^{j-1}\left[1-\sum_{k=1}^i|\tilde D_{n-j+k,n-j+i}|^2\right]+|\tilde D_{n-j,n}|^2\quad\text{[by column normalization]}\\
&\le \sum_{i=1}^{j-1}\left[\varepsilon^2+\sum_{i'=1}^{j-i}|\tilde D_{n-j+i,n-j+i+i'}|^2-\sum_{k=1}^{i-1}|\tilde D_{n-j+k,n-j+i}|^2\right]+|\tilde D_{n-j,n}|^2,\numberthis\label{eqn:diagcancel}
\end{align*}
using the diagonal estimate \eqref{eqn:dmatrix-diag} on the $k=i$ ($\le j-1$) terms.
Comparing the terms in the sums over $i,i'$ and $i,k$, we see all terms cancel except for the $i'=j-i$ terms, which gives
\begin{align}
\sum_{i=1}^j|\tilde D_{n-j,n-j+i}|^2&\le (j-1)\varepsilon^2+\sum_{i=1}^{j-1}|\tilde D_{n-j+i,n}|^2+|\tilde D_{n-j,n}|^2\le j\varepsilon^2,
\end{align}
since $|\tilde D_{n,n}|^2\ge1-\epsilon^2$ by \eqref{eqn:dmatrix-diag} and so the $\ell^2$ norm of the rest of the entries in the $n$th column must be $\le\varepsilon$. This proves \eqref{eqn:close-diag}. 

\begin{figure}[htb]
\begin{tikzpicture}[scale=.5]
\def\size{4}
\def\lw{1.5}
\foreach \row in {0,...,\size}{
\foreach \col in {0,...,\row}{
	\ifnum \col<\row
		\draw (-\col,\row) rectangle ++(1,1);
	\else
		\filldraw[fill=gray!50] (-\col,\row) rectangle ++(1,1); % gray diagonal
	\fi
}
}
\foreach \row in {2,...,\size}{\draw[color=blue, line width=\lw] (-\row+1.5,\size+.5) --++ (0,\row-\size-1);}
\foreach \row in {1,...,\numexpr\size-1\relax}{
\draw[color=red, line width=\lw] (-\row+.5,\row+.5) --++ (\row,0);
\filldraw[color=red] (.5,\row+.5) circle (3pt);
}

\foreach \col in {1,...,\size}{
\fill[color=black] (-\col+1.5,\size+.5) circle (3pt); % draw dots for top row
}
\filldraw[black] (-\size,\size) rectangle ++(1,1);
\node at (.5,.5) {{\small $\tilde D_{n,n}$}};
\node[left] at (-\size,\size+.5) {$n-j$};
\draw[decoration={brace,raise=3pt,aspect=.5,amplitude=4pt,mirror},decorate] (1,1)--(1,\size+1);
\node[right,yshift=.5cm, xshift=.2cm] at (1,{\size/2}) {$\|\tilde{D}_{\bullet\ne n,n}\|_2\le\varepsilon$};
\end{tikzpicture}
\caption{Illustration of the bounds and cancellation in \eqref{eqn:diagcancel} for $j=4$. The squares represent matrix elements of $\tilde{D}$, with the diagonal entries of $\tilde{D}$ shown as shaded squares. The bound using column orthonormalization in the first line of \eqref{eqn:diagcancel} is represented in blue vertical lines, and the application of the diagonal estimate \eqref{eqn:dmatrix-diag} on the $k=i$ term is shown in red horizontal lines. After canceling terms, the remaining entries involving $\tilde D$ are all in the last column of $\tilde D$, which, excluding the $(n,n)$ entry, has total $\ell^2$-norm mass at most $\varepsilon$.}\label{fig:diagcancel}
\end{figure}
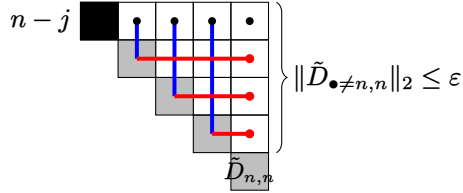

By applying phaseshifters, we may assume $\tilde D$ has real, nonnegative diagonal entries. 
For a unitary matrix $\tilde D$ with nonnegative diagonal and which satisfies \eqref{eqn:close-diag}, we have
\begin{align}
\|\tilde D-I\|_\hs^2
&\le 2\sum_{j=0}^{n-1}(j+1)\varepsilon^2
=O(n^2\varepsilon^2).
\end{align}
Thus taking $\tilde U=(T_{n,n-w}\cdots T_{n,n-1}T_{n-1,n-w-1}\cdots T_{n-1,n-2}\cdots T_{2,1}\Phi)^{-1}$ in \eqref{eqn:reck}, where $\Phi$ represents phaseshifters so that $\tilde D$ has nonnegative diagonal, we obtain
\begin{align}
\|U-\tilde U\|_\hs &=\|\tilde D\tilde U-\tilde U\|_\hs =\|\tilde D-I\|_\hs=O(n\varepsilon),
\end{align}
which proves \eqref{eqn:u2}.

The consequences \eqref{eqn:diamond} and \eqref{eqn:fidelity} follow from Eqs.~(2) and (4) of \cite{becker2021energy}. Given a Hamiltonian $H$ with ground state energy $0$ and an $E>0$, the energy-constrained (EC) diamond norm \cite{shirokov2018energy,winter2017energy} for a superoperator $\mathscr L_A$ which acts on trace-class operators and preserves self-adjointness, is
\begin{align}\label{eqn:ec-diamond}
\|\mathscr L_A\|_\diamond^{H,E}&:=\sup_{\Psi_{AA'}:\Tr\Psi_AH_A\le E}\|(\mathscr L_A\otimes \operatorname{Id}_{A'})(\Psi_{AA'})\|_1,
\end{align}
where $\|\cdot\|_1$ denotes the trace norm and $\Psi_{AA'}$ is a state on $AA'$ for an ancilla $A'$, and $\Psi_A$ is the reduced state on $A$.
Then \cite{becker2021energy} shows that for unitary channels $\mathscr U(\cdot)=\mathcal U(\cdot)\mathcal U^\dagger$ and $\tilde{\mathscr U}(\cdot)=\tilde{\mathcal U}(\cdot)\tilde{\mathcal U}^\dagger$,
\begin{align}
\|\mathscr U-\tilde{\mathscr{U}}\|_\diamond^{H,E}&=2\sqrt{1-\inf_{\langle\psi|H|\psi\rangle\le E}|\langle\psi|\mathcal U^\dagger\tilde{\mathcal U}|\psi\rangle|^2},
\end{align}
and also that for $N=\sum_{j=1}^n a_j^\dagger a_j$ the total photon number Hamiltonian and $\mathscr U$ and $\tilde{\mathscr U}$ corresponding to passive Gaussian unitaries,
\begin{align}
\|\mathscr U-\tilde{\mathscr U}\|_\diamond^{N,E}&\le C\sqrt{(c+n)(E+1)}\sqrt{\|U-\tilde U\|_\hs},
\end{align}
for explicit numerical constants $C,c$. With \eqref{eqn:u2}, this immediately implies \eqref{eqn:diamond} and \eqref{eqn:fidelity}.
\end{proof}

To prove Theorem~\ref{thm:complexity}(ii), recall that by \cite{reck1994experimental,clements2016optimal}, any $n\times n$ unitary matrix can be implemented exactly using $n(n-1)/2$ beamsplitter-phaseshifters in a 1D brickwall geometry. 
We first verify that with high probability, this order $\Theta(n^2)$ of gates is essentially (up to a logarithmic factor) required for implementing Haar random unitaries approximately in Hilbert--Schmidt norm.

\begin{lem}[Haar approximate circuit complexity]\label{lem:haar-complexity}
Let $\mathcal H_n$ denote the normalized Haar measure on $\U(n)$, and let $0<\delta'<1/2$.
Then there are constants $C,c_1,c_2>0$ such that for any $\delta\le(1/2-\delta')\sqrt{n}$,
\begin{align}
\underset{U\dsim\mathcal H_n}{\P}[\exists V\text{ with $c_1n^2/\log n$ gates s.t. }\|V-U\|_\hs\le\delta]&\le Ce^{-c_2n^2}.
\end{align}
\end{lem}
\begin{proof}
This will follow from concentration of Haar measure on $\U(n)$ combined with the counting argument of \cite{shannon1949synthesis} and \cite[\S6]{chen2024incompressibility}.
For the counting argument, following \cite[\S6]{chen2024incompressibility}, let $G_\epsilon$ denote an $\epsilon$-net of $\U(2)$ in Hilbert--Schmidt norm, i.e. for any $u\in\U(2)$ there is $g\in G_\epsilon$ such that $\|u- g\|_\hs\le\epsilon$. Since $\U(2)$ is compact and has dimension 4, for bounded $\epsilon\le O(1)$ we can construct $G_\epsilon$ with $\le C\varepsilon^{-4}$ elements. 
Let $M_{G_{\epsilon},R}$ denote all $n\times n$ linear optical unitaries which can be built from at most $R$ all-to-all gates taken from $G_\epsilon$.
Then
\begin{align}
|M_{G_{\epsilon},R}|\le \left(\binom{n}{2}|G_{\epsilon}|\right)^R\le (Cn^2\varepsilon^{-4})^R.
\end{align}

The $\epsilon$-net property implies for any $V$ constructed from $R$ gates, there is $V'\in M_{G_\epsilon,R}$ such that $\|V-V'\|_\hs\le R\epsilon$.
In particular, taking $\epsilon=\frac{\delta'\sqrt{n}}{2R}$, we then have
\begin{align*}
\P_U[\exists V\text{ with $R$ gates s.t. }\|V-U\|_\hs\le\delta]&\le \P_U\bigg[\exists V'\in M_{G_{\epsilon},R}\text{ s.t. }\|V'-U\|_\hs\le\delta+\frac{\delta'\sqrt{n}}{2}\bigg]\\
&\le |M_{G_{\epsilon},R}|\sup_{V'}\P_U[\|V'-U\|_\hs\le{\sqrt{n}}(1-\delta')/2].\numberthis\label{eqn:rprob}
\end{align*}
For a fixed unitary $V'$ and unitary $U$, the function $U\mapsto \|V'-U\|_\hs$ is $1$-Lipschitz and bounded from above by $2\sqrt{n}$.
For $U$ random with $\E U_{ij}=0$, we can estimate
\begin{align*}
\E\|V'-U\|_\hs \ge \frac{1}{2\sqrt{n}}\E\|V'-U\|_\hs^2
&=\frac{1}{2\sqrt{n}}\sum_{i,j=1}^n\E|V'_{ij}-U_{ij}|^2\\
&\ge \frac{1}{2\sqrt{n}}\sum_{i,j=1}^n\E|U_{ij}|^2=\frac{\sqrt{n}}{2},\numberthis\label{eqn:meanhs}
\end{align*}
where we used that $\E|U_{ij}-V'_{ij}|^2\ge\E|U_{ij}-\E U_{ij}|^2=\E|U_{ij}|^2$.
Then using concentration of Haar measure on $\U(n)$ (Theorem~\ref{thm:haar-concentration}), we obtain for any unitary $V'$,
\begin{align*}
\underset{U\dsim\mathcal H_n}{\P}[\|V'-U\|_\hs\le {\sqrt{n}}(1-\delta')/2]&\le \underset{U\dsim\mathcal H_n}{\P}[\E\|V'-U\|_\hs-\|V'-U\|_\hs\ge \delta'\sqrt{n}/2]\\
&\le e^{-cn(\delta'\sqrt{n}/2)^2}= e^{-c'n^2}.\numberthis\label{eqn:haar-exp}
\end{align*}
Inserting the above in \eqref{eqn:rprob} with $R=c_1n^2/\log n$ for an appropriate (small) choice of $c_1$ completes the proof.
\end{proof}

We can obtain a similar property for random linear optical unitaries which converge in Wasserstein distance to Haar measure.
\begin{proof}[Proof of Theorem~\ref{thm:complexity}(ii)]
We follow the same argument as in the proof of Lemma~\ref{lem:haar-complexity}, but we need to replace \eqref{eqn:haar-exp} with an analogous bound for $U$ a depth-$d$ random circuit. Since \eqref{eqn:haar-exp} is a tail bound, we will do this by comparing moments of $\|V'-U\|_\hs$, for a fixed unitary $V'$, using the Wasserstein convergence of $U$ to Haar measure.
We let $\mathcal L(U)$ denote the law of the depth-$d$ random $U$, and denote expectations and probabilities with respect to this measure using $\E$ and $\P$. The notation $\E_{\mathcal H_n}$ will denote expectation with respect to Haar measure $\mathcal H_n$ on $\U(n)$.

Since $\|V'-U\|_\hs$ is bounded from above by $2\sqrt{n}$ for unitary $V'$ and $U$, we see from the dual characterization \eqref{eqn:w1} of $L^1$ Wasserstein distance, that for $W_{\hs,1}(\mathcal L(U),\mathcal{H}_n)\le\varepsilon$,
\begin{align}
\left|e^{t\E\|V'-U\|_\hs}-e^{t\E_{\mathcal H_n}\|V'-U\|_\hs}\right|&\le \varepsilon t e^{2t\sqrt{n}},
\end{align}
using that $x\mapsto e^{tx}$ is $te^{2t\sqrt{n}}$-Lipschitz on $[0,2\sqrt{n}]$ and that $U\mapsto\|V'-U\|_\hs$ is $1$-Lipschitz.
For $t>0$, the map $U\mapsto e^{-t\|V-U\|_\hs}$ is $t$-Lipschitz, so for $W_1(\mathcal L(U),\mathcal{H}_n)\le\varepsilon$, we can compare the moment-generating functions for $t>0$ as
\begin{align}\label{eqn:mgf}
\left|\E[e^{t(\E\|V'-U\|_\hs-\|V'-U\|_\hs)}] - \E_{\mathcal H_n}[e^{t(\E_{\mathcal H_n}\|V'-U\|_\hs-\|V'-U\|_\hs)}]\right|&\le 2\varepsilon te^{2t\sqrt{n}}.
\end{align}
For $U$ Haar random, we know from concentration of Haar measure (Theorem~\ref{thm:haar-concentration}) that $\|V'-U\|_\hs$ is subgaussian with variance proxy $\sigma^2=6/n$, and so has moment-generating function bounded as $\E_{\mathcal H_n}[e^{t(\E_{\mathcal H_n}\|V'-U\|_\hs-\|V'-U\|_\hs)}]\le e^{ct^2/n}$. We will take $t=\delta' n^{3/2}/(4c)$ and $\varepsilon=\frac{1}{t}e^{ct^2/n-2t\sqrt{n}}=Cn^{-3/2}e^{-c''n^2}$ in \eqref{eqn:mgf}. 
By Theorem~\ref{thm:lowerbound}(ii), an appropriate constant in the depth condition \eqref{eqn:complexity-depth} implies $W_1(\mathcal L(U),\mathcal{H}_n)\le\varepsilon$ for such $\varepsilon$.
Then using that \eqref{eqn:meanhs} still holds since $\E U_{ij}=0$, followed by exponential Markov inequality and \eqref{eqn:mgf}, gives
\begin{align*}
\P[\|V'-U\|_\hs\le\sqrt{n}(1-\delta')/2]&\le \P[\E\|V'-U\|_\hs-\|V'-U\|_\hs
\ge \delta'\sqrt{n}/2]\\
&\le e^{-t\delta'\sqrt{n}/2}\E \big[e^{t(\E\|V'-U\|_\hs-\|V'-U\|_\hs)}\big]\\
&\le e^{-t\delta'\sqrt{n}/2} (2t\varepsilon e^{2t\sqrt{n}}+e^{ct^2/n})=3e^{-c'n^2},\numberthis
\end{align*}
for a $c'>0$, which implies \eqref{eqn:complex2} by the same argument as in the proof of Lemma~\ref{lem:haar-complexity}.
\end{proof}

\appendix

\section{Proof of Theorem~\ref{thm:nonrev}}\label{appendix:sym}

We follow the proof of \cite[Theorems 7.2, 7.3]{chung2005laplacians} to obtain Theorem~\ref{thm:nonrev} on bounding the rate of convergence of a nonreversible Markov chain in terms of its additive symmetrization which has transition matrix given by $\frac{P+P^*}{2}$ for $P^*(x,y):=\frac{\phi(y)}{\phi(x)}P(y,x)$. Recall $\Phi=\operatorname{diag}(\phi)$ for stationary vector $\phi$ of $P$.

\begin{proof}
Let $\langle f|g\rangle:=\sum_x\bar f(x)g(x)$ be the usual inner product.
We can write $P=\alpha I+(1-\alpha)Q$ for some transition probability matrix $Q$, and then \begin{align}
\mathscr L=I-\frac{\Phi^{1/2}P\Phi^{-1/2}+\Phi^{-1/2}P^T\Phi^{1/2}}{2}=(1-\alpha)\left[I-\frac{\Phi^{1/2}Q\Phi^{-1/2}+\Phi^{-1/2}Q^T\Phi^{1/2}}{2}\right]=:(1-\alpha)\mathscr{L}_Q.
\end{align}
Letting $M:=\Phi^{1/2}P\Phi^{-1/2}$, we can estimate for any vector $f\in\C^n$,
\begin{align*}
\frac{\|f M\|_2^2}{\|f\|_2^2}=\frac{\langle f|MM^T|f\rangle}{\|f\|_2^2}&=\frac{\langle f|\Phi^{1/2}(\alpha+(1-\alpha)Q)\Phi^{-1}(\alpha+(1-\alpha)Q^T)\Phi^{1/2}|f\rangle}{\|f\|_2^2}\\
&=\alpha^2+(1-\alpha)^2\frac{\langle f|\Phi^{1/2}Q\Phi^{-1}Q^T\Phi^{1/2}|f\rangle}{\|f\|_2^2}-2\alpha(1-\alpha)\frac{\langle f|\mathscr{L}_Q|f\rangle}{\|f\|_2^2}+2\alpha(1-\alpha)\\
&\le 1-2\alpha\frac{\langle f|\mathscr{L}|f\rangle}{\|f\|_2^2},\numberthis
\end{align*}
where in the last inequality we used that $\|\Phi^{1/2}Q\Phi^{-1}Q^T\Phi^{1/2}\|\le\|\Phi^{1/2}Q\Phi^{-1/2}\|^2\le 1$, since  using Jensen's inequality and that $\phi$ is also stationary for $Q$,
\begin{align*}
\|\Phi^{1/2}Q\Phi^{-1/2}f\|_2^2&=\sum_x\phi(x)\left|\sum_y Q_{xy}\phi^{-1/2}(y)f(y)\right|^2
\le \sum_{x,y}\phi(x)Q_{xy}\phi^{-1}(y)|f(y)|^2=\sum_y |f(y)|^2.
\end{align*}
The zero eigenvector of $\mathscr L$ is the vector $\sqrt{\phi}$, so we can write $\lambda_1=\inf\limits_{f:\,\sum_x f(x)\phi(x)=0}\frac{\langle g,\mathscr{L}g\rangle}{\|g\|_2^2}$ for $g=f\Phi^{1/2}$ \cite[Theorem 4.1]{chung2005laplacians}. Then
for $\langle f\Phi^{-1/2}|\phi\rangle=\langle f|\phi^{1/2}\rangle=\sum_x \bar f(x)\phi(x)^{1/2}=0$, 
\begin{align}\label{eqn:lambda1}
\frac{\|fM\|_2^2}{\|f\|_2^2}\le1-2\alpha\lambda_1.
\end{align}
Let $d_2(t):=\max_y\left(\sum_x\frac{(P^t(y,x)-\phi(x))^2}{\phi(x)}\right)^{1/2}$, and note that $\max_y\|P^t(y,\cdot)-\phi\|_\tv\le\frac{1}{2}d_2(t)$ by the Cauchy--Schwarz inequality. We have
\begin{align*}
d_2(t)^2&=\max_y\|(\delta_yP^t-\phi)\Phi^{-1/2}\|_2^2\\
&=\max_y\|(\delta_y-\phi)P^t\Phi^{-1/2}\|_2^2
=\max_y\|(\delta_y-\phi)\Phi^{-1/2}M^t\|_2^2.\numberthis
\end{align*}
Since $\langle(\delta_y-\phi)\Phi^{-1/2}|\phi^{1/2}\rangle=0$, applying \eqref{eqn:lambda1} $t$ times followed by the estimate $\|(\delta_y-\phi)\Phi^{-1/2}\|_2^2=\phi(y)^{-1}-1\le\phi(y)^{-1}$ yields
\begin{align*}
d_2(t)^2&\le (1-2\alpha\lambda_1)^t\max_y\|(\delta_y-\phi)\Phi^{-1/2}\|_2^2\\
&\le (1-2\alpha\lambda_1)^t\max_y\phi(y)^{-1},\numberthis
\end{align*}
which gives \eqref{eqn:thm41bound}.
\end{proof}

\section{Lipschitz constant and concentration}\label{sec:concentration}

In this section we review some standard estimates, and apply concentration of measure to answer a question mentioned in \cite{iosue2023page} concerning the fully Haar random case.
\begin{lem}\label{lem:lipschitz}
For any subsystem $\Gamma$ of size $k$, the R\'enyi-2 entropy $S_2(U)$ is $2\sqrt{k}\sinh^2(2s)$-Lipschitz with respect to the Hilbert--Schmidt (or Frobenius) norm on $n\times n$ matrices.
\end{lem}
\begin{proof}
Recall from \eqref{eqn:s2-tanh} that $S_2(U)$ can be written in terms of the eigenvalues $(\lambda_j)_j$ of $W=VV^\dagger$ where $V=P_\Gamma UU^TP_\Gamma^T$, as 
\begin{align}
S_2(U)=\Tr\log\cosh(2s)I_k+\frac{1}{2}\sum_{j=1}^k\log(1-\tanh^2(2s)\lambda_j).
\end{align}
Let $g(x_1,\ldots,x_k)=\frac{1}{2}\sum_{j=1}^k\log(1-\tanh^2(2s)x_j)$, which by taking derivatives is  seen to be Lipschitz with respect to the Euclidean norm on $[0,1]^k$, with Lipschitz constant $L=\frac{\sqrt{k}}{2}\frac{\tanh^2(2s)}{1-\tanh^2(2s)}=\frac{\sqrt{k}}{2}\sinh^2(2s)$.
(The Lipschitz constant for $g$ follows directly from $\sum_i |x_i - y_i | = \lVert x-y \rVert_1 \leq \sqrt k \lVert x-y \rVert_2$ and the Lipschitz constant for log defined on the domain $[x_0,\infty)$ being $1/x_0$.) 
Let $\tilde g(VV^\dagger):=g(\lambda_1(VV^\dagger),\ldots,\lambda_k(VV^\dagger))$. For unitary matrices $U$ and $U'$ with corresponding $V=P_\Gamma UU^TP_\Gamma^T$ and $V'=P_\Gamma U'U'^TP_\Gamma^T$, we then have $|S_2(U)-S_2(U')|=|\tilde g(VV^\dagger)-\tilde g(V'V'^\dagger)|$.
To relate the Lipschitz constant of $S_2(U)$ to that of $g$, we use
\begin{lem}[Hoffman--Wielandt]
Let $A,B$ be $n\times n$ self-adjoint matrices, with eigenvalues $\lambda_1^A\le\cdots\le \lambda_n^A$ and $\lambda_1^B\le\cdots\le\lambda_n^B$. Then
\begin{align}
\sum_{i=1}^n(\lambda_i^A-\lambda_i^B)^2&\le \Tr\left[(A-B)^2\right] = \lVert A - B \rVert_\hs^2.
\end{align}
\end{lem}
See e.g.~\cite[Lemma 2.1.19]{AGZ-book} for a proof.
Applying this to $A=VV^\dagger$ and $B=V'V'^\dagger$, we obtain
\begin{align*}
|S_2(U)-S_2(U')|&=|\tilde g(VV^\dagger)-\tilde g(V'V'^\dagger)|\\
&\le \frac{\sqrt{k}}{2}\sinh^2(2s)\|VV^\dagger-V'V'^\dagger\|_\hs\\
&= \frac{\sqrt{k}}{2}\sinh^2(2s)\|VV^\dagger-VV'^\dagger+VV'^\dagger-V'V'^\dagger\|_\hs\\
&\le \sqrt{k}\sinh^2(2s)\|V-V'\|_\hs\\
&\le 2\sqrt{k}\sinh^2(2s)\|U-U'\|_\hs,\numberthis
\end{align*}
where we used that $\|AB\|_\hs\le \|A\|_\mathrm{op}\|B\|_\hs$ and that $V, V',P_k,P_k^T$ all have operator norm $\le1$.
\end{proof}

\subsection{Haar limit typicality}
In the case where $U_n$ is fully Haar random from $\U(n)$, the Lipschitz constant in Lemma~\ref{lem:lipschitz} and concentration properties of the Haar measure on $\U(n)$ can be used to answer a question about weak vs strong typicality of $S_2(U_n)$ mentioned in \cite{iosue2023page}. 
\begin{thm}[concentration of Haar measure]\label{thm:haar-concentration}
Let $F:\U(n)\to\R$ be $L$-Lipschitz with respect to the Hilbert--Schmidt norm. Then for $t>0$,
\begin{align}\label{eqn:concentration}
\P[F(U_n)\ge \E F(U_n)+t]&\le e^{-nt^2/(12L^2)}.
\end{align}
\end{thm}
For a proof, see e.g.~{\cite[Theorems 5.5, 5.16]{meckes2019book}}.
The bound \eqref{eqn:concentration} is a subgaussian tail bound, and implies that $F(U_n)$ is a \emph{subgaussian random variable}, which is equivalent to other properties such as control on the growth of moments (see e.g.~\cite[Prop. 2.5.2]{vershynin2018book}).
The variance proxy from \eqref{eqn:concentration} is $\sigma^2=6L^2/n$ for $L=2\sqrt{k}\sinh^2(2s)$, which is similar to (though must be larger than) the exact variance calculated in \cite{iosue2023page}.
Theorem~\ref{thm:haar-concentration} then implies for any $t>0$,
\begin{align}\label{eqn:subgaussian}
\P\left[\left|S_2(U_n)-\E S_2(U_n)\right|>t\right] &\le 2e^{-cnt^2/(k\sinh^4(2s))}.
\end{align}
Combined with the size of $\E S_2(U_n)$, this can be used to recover typicality results from \cite[\S4]{iosue2023page}. In \cite{iosue2023page}, it was shown that for $k=rn$, $r\in(0,1)$ fixed, and fixed squeezing parameter $s$, that there is weak typicality: for any $\varepsilon>0$,
\begin{align}
\lim_{n\to\infty}\P[|S_2(U)-\E S_2(U)|<\varepsilon\E S_2(U)]=1.
\end{align}
Existence or absence of strong typicality, which would be $\lim_{n\to\infty}\P[|S_2(U)-\E S_2(U)|<\varepsilon]=1$ (convergence in probability), was left open in this case.
We can use the control of higher moments implied by \eqref{eqn:subgaussian} to show absence of strong typicality.
For fixed $r$ and $s$, the variance $\sigma^2$ as determined in \cite{iosue2023page} is constant.
Intuitively, the only way to have constant variance and strong typicality for $Y_n:=S_2(U_n)-\E S_2(U_n)$ is if $Y_n$ takes very large values to cancel out the vanishing (by strong typicality assumption) probability $\P[|Y_n|>\varepsilon]$. But this forces $Y_n$ to have large moments, which will be incompatible with the subgaussian moment property. Uniformly bounded fourth moment (or any $2+\delta$ moment) is enough to rule out strong typicality: We have
\begin{align*}
\sigma^2\leftarrow\E|Y_n|^2&=\E[|Y_n|^2\oneb_{|Y_n|^2\le \sigma^2/2}]+\E[|Y_n|^2\oneb_{|Y_n|^2> \sigma^2/2}]\\
&\le \frac{\sigma^2}{2}+(\E|Y_n|^4)^{1/2}\P[|Y_n|>\sigma/\sqrt{2}]^{1/2},
\end{align*}
where we applied Cauchy--Schwarz to obtain the fourth moment bound.
Then if the limiting variance $\sigma^2$ is nonzero and $\limsup \E|Y_n|^4<\infty$ (which is implied by the subgaussian tail bound since the variance proxy is a constant),
\begin{align}
\liminf\P[|Y_n|>\sigma/\sqrt{2}]&\ge \frac{\sigma^4}{4\limsup\E|Y_n|^4}>0,
\end{align}
so $Y_n$ is not strongly typical.

\section{Entanglement lower bounds at large depth}\label{sec:lower2}

In this section, we derive an exact 2-boson random walk relation for the fourth moments $\E|\langle x|UU^T|y\rangle|^2$ (Theorem~\ref{thm:uut}) for arbitrary circuit geometries, and use this to prove the lower bound Theorem~\ref{thm:lower-bound}. 
In this section only, for notational convenience circuit depth is indexed by the number of individual layers, rather than the number of complete cycles.

As a consequence of \eqref{eqn:22cases}, we can express $\E|\langle x|UU^T|y\rangle|^2$ in terms of a \emph{2-boson} classical random walk.

\begin{thm}[2-boson random walk]\label{thm:uut}
For any circuit geometry, let $Y_t=(Y_t^1,Y_t^2)\in[n]\times[n]$ be a 2-particle classical random walk through the circuit defined as follows. If $Y_{t-1}=(v,w)$ and $v,w$ are not in the same 2-mode beamsplitter in the layer described by $\pi_t$, then $Y_t^1$ and $Y_t^2$ evolve independently according to the 1-particle circuit walk. If $Y_{t-1}^1,Y_{t-1}^2$ enter the same beamsplitter connecting modes $\{a,b\}\in\pi_t$, then for any $x,y,v,w\in\{a,b\}$, the transition probabilities are
\begin{align}\label{eqn:walk2}
\P[Y_t=(x,y)|Y_{t-1}=(v,w)]&=\begin{cases}
\frac13,&x=y\\
\frac16,&x\ne y
\end{cases}.
\end{align}
Let $U$ be a random $t$-layer linear optical unitary according to some circuit geometry.
Then in terms of the associated classical 2-boson walk $Y_t$,
\begin{align}\label{eqn:uut-2boson}
\E|\langle x|UU^T|y\rangle|^2&=\sum_{\alpha=1}^n\P_{(\alpha,\alpha)}[Y_t=(x,y)]. 
\end{align}
\end{thm}
As a consequence of Theorem~\ref{thm:uut}, \eqref{eqn:s2}, and Lemma~\ref{lem:moments}, we immediately obtain
\begin{cor}\label{cor:s2-walk}
Let $U$ be a $t$-layer random linear optical unitary, and let $Y_t$ be the associated correlated 2-particle random walk as defined in Theorem~\ref{thm:uut}. Then
\begin{align}\label{eqn:s2-lowerbound-walk}
\E S_2(U)&\ge \frac12\sum_{z=1}^n\P_{(z,z)}[Y_t\text{ split across }\Gamma|\Gamma^c]\log\cosh(2s),
\end{align}
and 
\begin{align}
\E S_2(U)&\le \frac14\sum_{z=1}^n\P_{(z,z)}[Y_t\text{ split across }\Gamma|\Gamma^c]\sinh^2(2s).
\end{align}
\end{cor}
\begin{proof}[Proof of Theorem~\ref{thm:uut}]
Recall we have 
$
\E|\langle x|UU^T|y\rangle|^2=\sum_{\alpha=1}^n\E[|U_{x\alpha}|^2|U_{y\alpha}|^2].
$
For $t\in\N$, let $F_t(x,y):=\E[|U^{\langle t\rangle}_{x\alpha}|^2|U^{\langle t\rangle}_{y\alpha}|^2]$, where $U^{\langle t\rangle}$ denotes the linear optical unitary for the first $t$ layers of the circuit. 
Then pulling off the left-most factor of each $U^{\langle t\rangle}$, \eqref{eqn:22cases} gives (note it is the transposed matrix there but the calculation is the same),
\begin{align}
F_t(x,y)&=\sum_{\substack{\ell\in\pi_t(x)\\m\in\pi_t(y)}}F_{t-1}(\ell,m)\times\begin{cases}
1,&x=y\text{ and }|\pi_t(x)|=1\\
\frac{1}{3},&x=y\text{ and }|\pi_t(x)|=2\\
\frac{1}{6},&\{x,y\}\in\pi_t,\;x\ne y\\
2^{-(\delta_{|\pi_t(x)|,2}+\delta_{|\pi_t(y)|,2})},&\{x,y\}\not\in\pi_t\\
\end{cases}.
\end{align}
But this is the forwards Kolmogorov relation for the 2-particle walk $Y_t$ described in the present theorem statement. Thus $F_t(x,y)=\P_{(\alpha,\alpha)}[Y_t=(x,y)]$, which gives \eqref{eqn:uut-2boson}.
\end{proof}
\begin{rmk}\label{rmk:sym}
Alternatively, one can view the 2-boson random walk directly through the 2-particle symmetric subspace $\mathrm{Sym}^2(\mathcal H):=\{v\in\mathcal H\otimes\mathcal H:\tau v=v\}$, where $\mathcal H=\C^n$ and $\tau(x\otimes y)=y\otimes x$. In terms of the orthonormal basis consisting of vectors $|x,x\rangle:=|x\rangle\otimes|x\rangle$ and $|x,y\rangle_s:=\frac{1}{\sqrt{2}}(|x\rangle\otimes|y\rangle+|y\rangle\otimes|x\rangle)$ for $x<y$, $x,y\in\{1,\ldots,n\}$, a symmetric matrix $Q$ can be represented as a vector in $\mathrm{Sym}^2(\mathcal H)$,
\begin{align}
\Psi(Q)&:=\sum_x Q_{xx}|x,x\rangle+\sqrt2\sum_{x<y}Q_{xy}|x,y\rangle_s,
\end{align}
with $\|\Psi(Q)\|^2=\|Q\|_\hs^2$.
Then one can check symmetric conjugation $Q\mapsto VQV^T$ acts as
\begin{align}
\Psi(VQV^T)&= (V\otimes V)\Psi(Q). 
\end{align}
For $Q=UU^T$ with $U=\prod_{i=d}^1U^{(i)}$, new layers correspond to symmetric conjugation of $Q$, and thus to evolution of the vector $\Psi(Q)\in \mathrm{Sym}^2(\mathcal H)$. This can then be used to iteratively calculate $\E|\langle x|UU^T|y\rangle|^2$ for $x\ne y$ using e.g. for $U=\prod_{i=d}^1U^{(i)}=:U^{(d)}U_{d-1}$,
\begin{align}
\E|\langle x|UU^T|y\rangle|^2&=\frac12\E|\langle x,y|_sU^{(d)}\otimes U^{(d)}|\Psi(U_{d-1}U_{d-1}^T)\rangle_{\mathrm{Sym}^2(\mathcal H)}|^2,
\end{align}
and calculating the action of a beamsplitter $g$ on a 2-mode symmetric space, $\E_g|\langle \eta|g\otimes g|\xi\rangle|^2$ for basis states $\eta,\xi\in\{|x,x\rangle,|y,y\rangle,|x,y\rangle_s\}$.
The same-beamsplitter transition \eqref{eqn:walk2} is seen to correspond to the walk on the symmetric subspace with uniform $1/3$ probability for each of the possible outputs $|x,x\rangle,|y,y\rangle,|x,y\rangle_s$.
\end{rmk}

We now consider circuit geometries as in Section~\ref{sec:lower}, which consist of a sequence $\pi_1,\ldots,\pi_M$ repeated periodically. We will generally consider families of circuit geometries corresponding to graphs with bounded degree as $n\to\infty$.
Using Theorem~\ref{thm:uut}, we prove the following Theorem~\ref{thm:lower-bound} 
(cf. Theorem~\ref{thm:lowerbound}). It involves only the spectral gap, related to the mixing time, of the simple random walk (SRW) on the associated graph. In particular, it does not involve a meeting time $\tmeets$, which can be large in $n$. 
While the meeting time argument was sufficient to obtain a good depth condition at depth $\tilde\Theta(n^2)$ for one-dimensional random brickwall circuits (Section~\ref{sec:lower}), it was not ideal for faster mixing geometries.
The improvement here will allow us to use other circuit geometries corresponding to graphs with bounded degree, to obtain large average-case entanglement in all subsystems at sublinear, including logarithmic, depth. Since the spectral gaps for simple random walk on many graphs are known or can be estimated \cite[\S12]{levin2009book}, we immediately obtain useful bounds for many circuit geometries corresponding to bounded degree families of graphs, such as brickwork and expanders.

\begin{thm}[saturation bound]\label{thm:lower-bound}
Let $U$ be a $T$-layer random linear optical unitary with a connected circuit geometry consisting of $M$ layers per cycle, and let $\Delta_n$ be the spectral gap for the simple random walk (SRW) on the associated graph, i.e. the first nonzero eigenvalue of $I-P_\srw$, where $P_\srw$ is the Markov transition matrix for the simple random walk.
If the total number of layers\footnote{the total number of cycles, referenced as depth in the main text, is $\sim T/M$} is
\begin{align}\label{eqn:tmixbound}
T\ge M\left\lceil\frac{16\cdot 36^M\log n}{\Delta_n}\right\rceil,
\end{align}
then uniformly over all subsystems $\Gamma\subset\{1,\ldots,n\}$,
\begin{align}\label{eqn:lower-bound}
\E S_2(U)&\ge \frac{|\Gamma|(n-|\Gamma|)}{n}\log(\cosh(2s))(1-O(1/n)),
\end{align}
with the $O(1/n)$ term as $n\to\infty$ solely a function of $n$.
\end{thm}
\begin{rmk}
The bound \eqref{eqn:tmixbound} need not be optimal, and in particular it does not work well if $M$ is non-constant. For example, the hypercube graph has $M=\log_2n$ and spectral gap $\Delta_n\propto1/\log n$ \cite[\S12.4]{levin2009book}. But as shown in subsequent work \cite{shou2026log} using different methods, \eqref{eqn:lower-bound} holds at $\Theta(\log n)$ layers.
\end{rmk}

\begin{proof}[Proof of Theorem~\ref{thm:lower-bound}]

We will use the bound \eqref{eqn:s2-lowerbound-walk} in Corollary~\ref{cor:s2-walk}.
To estimate which depths make the right side of \eqref{eqn:s2-lowerbound-walk} large, we want to relate the mixing time for the 1-particle walk through the circuit gates to that of the correlated 2-particle walk $Y_t$. Let $Z_t$ be a 2-particle walk with independent particles, which has uniform stationary vector. We will compare $Y_t$ to $Z_t$. 

For this proof, it will be convenient to consider the versions $\Yt_t$ and $\Zt_t$ of  $Y_t$ and $Z_t$ respectively, where a single time-step of $\Yt_t$ or $\Zt_t$ corresponds to an entire cycle, or $M$ layers, of the circuit; equivalently, $\Yt_t:=Y_{Mt}$ and $\Zt_t:=Z_{Mt}$.
This will be convenient since we will compare to the time-homogeneous Markov chain where a single time-step corresponds to $M$ circuit layers.
Once we have a mixing time bound for $\Yt_t$, we simply multiply the time by $M$ to get the bound for $Y_t$.

We first need to determine the stationary distribution $\phi$ for $\Yt_t$. Considering each layer in a single time step, we see that the stationary distribution $\phi$ for $\Yt_t$ is given by 
\begin{align}\label{eqn:phi}
\phi_{(x,x)}=\frac{2}{n(n+1)},\quad\text{and}\quad \phi_{(x,y)}=\frac1{n(n+1)}\text{ for $x\ne y$};
\end{align}
in each layer, we can check that evolution through any beamsplitter leaves $\phi$ invariant: If $x,y$ are in different beamsplitters, then they evolve independently and each output $(x',y')$ has probability $1/4$ and $x'\ne y'$, so $\phi_{(x,y)}$ is unchanged. If $x,y$ go through the same beamsplitter, then the local transition matrix block is a $4\times 4$ block with $1/3$ in columns for $(x,x)$ and $(y,y)$, and $1/6$ in columns for $(x,y)$ and $(y,x)$. We see again that $\phi_{(x,y)}$, $\phi_{(x,x)}$, $\phi_{(y,y)}$, and $\phi_{(y,x)}$ are unchanged.

The outline for bounding the mixing time for $\Yt_t$ can be expressed heuristically as follows:
\begin{align}
t_{\Yt,\text{mix}}\overset{(1)}\lesssim t_{\Yt, \text{mix, sym}}&\overset{(2)}\lesssim t_{\Zt, \text{mix, sym}}\overset{(3)}\lesssim t_{\text{mix, sym}}^{\text{1-particle}}\overset{(4)}{\lesssim} \tmix^{\text{1-particle SRW}},
\end{align}
where $\lesssim$ is used informally and imprecisely in the diagram, and the labeled steps are:
\begin{enumerate}
\item Compare $\Yt_t$ to its additive symmetrization $\Yt_{t,\sym}$ (\cite{chung2005laplacians}, and Theorem~\ref{thm:nonrev}).
\item Compare $\Yt_{t,\sym}$ to $\Zt_{t,\sym}$, the additive symmetrization of the 2-particle walk with \emph{independent} particles. This uses Dirichlet form comparison using the transition matrix comparison which will be done in \eqref{eqn:indep-compare}.
\item Compare $\Zt_{t,\sym}$ to the 1-particle symmetrized random walk, using independence of the particles in $\Zt_t=(\Zt_t^1,\Zt_t^2)$.
\item Compare the 1-particle symmetrized random walk to a simple random walk on the underlying graph (Corollary~\ref{cor:mixsrw}), using the path comparison method \cite[Theorem 13.23]{levin2009book}).
\end{enumerate}

We do these comparisons as follows, in terms of spectral gaps.
\begin{enumerate}[leftmargin=*]
\item As in Section~\ref{sec:lower}, we use Theorem~\ref{thm:nonrev}, which was an adaptation of \cite[Corollary 2.9]{fill1991eigenvalue} and \cite[Theorem 7.3]{chung2005laplacians}.
Let $P$ be the Markov transition matrix for $\Yt_t$, which corresponds to walking through a single circuit cycle consisting of $M$ layers.
Let $P_\sym:=\frac12(P+P^*)$, for $P^*(u,v):=\frac{\phi(v)}{\phi(u)}P(v,u)$, be the transition matrix for the additive symmetrization of $P$. Letting $\Phi:=\op{diag}(\phi)$ be the diagonal matrix whose diagonal elements are given by the stationary vector $\phi$ for $\Yt_t$, then
\begin{align}\label{eqn:psym}
\Phi^{1/2}P_\sym\Phi^{-1/2}=\frac{\Phi^{1/2}P\Phi^{-1/2}+\Phi^{-1/2}P^T\Phi^{1/2}}{2}.
\end{align}
We have holding probabilities $P(x,x)\ge 6^{-M}$, and so $P_\sym(x,x)\ge 6^{-M}$ as well.
Recall a single time-step for $\Yt_t$ corresponds to an entire cycle consisting of $M$ circuit layers.
In terms of the first nonzero eigenvalue $\lambda_Y$ of $I-\Phi^{1/2}P_\sym\Phi^{-1/2}$, 
Theorem~\ref{thm:nonrev} gives 
\begin{align}\label{eqn:p1}
\max_y\|P^t(y,\cdot)-\phi\|_\tv&\le \frac{(1-2\cdot 6^{-M}\lambda_Y)^{t/2}}{2\sqrt{\min_{x,y}\phi_{(x,y)}}}=\frac12\sqrt{n(n+1)}(1-2\cdot 6^{-M}\lambda_Y)^{t/2}.
\end{align}
Note that $\lambda_Y$ is also the first nonzero eigenvalue of $I-P_\sym$, as $I-\Phi^{1/2}P_\sym\Phi^{-1/2}$ and $I-P_\sym$ are similar matrices.

\item We want to compare the additive symmetrizations of $\Yt_t$ and $\Zt_t$.
Recall $P$ is the transition matrix for $\Yt_t$, with a single time-step corresponding to $M$ layers. Let $P_\ind$ be the transition matrix for $\Zt_t$, where a single time-step also consists of $M$ layers. Then by comparing the $1/3,1/6$ probabilities in $P$ to the $1/4$s in $P_\ind$ when the particles go through the same beamsplitter, we see
\begin{align}\label{eqn:pbound}
\left(\frac{2}{3}\right)^M P_\ind(x,y)\le P(x,y)\le \left(\frac{4}{3}\right)^MP_\ind(x,y).
\end{align}
Let $P_{\ind,\sym}=\frac12(P_\ind+P_\ind^T)$ be the transition matrix for the additive symmetrization of $\Zt_t$.
Using that $\frac{1}{2n^2}\le\phi\le\frac{2}{n^2}$ by \eqref{eqn:phi}, and \eqref{eqn:pbound}, we see the edge conductances satisfy
\begin{align}\label{eqn:indep-compare}
\phi(x)P_\sym(x,y)\ge \frac12\left(\frac{2}{3}\right)^M\mu(x) P_{\ind,\sym}(x,y),
\end{align}
where $\mu$ is the uniform measure on $[n]^2$.
As in step (1), let $\lambda_Y$ be the first nonzero eigenvalue of $I-P_\sym$, and let $\lambda_Z$ be the first nonzero eigenvalue of $I-P_{\ind,\sym}$. Note that the stationary distribution for $\Yt_t$ is $\phi$ in \eqref{eqn:phi}, while the stationary distribution for $\Zt_t$ is the uniform distribution $\mu$.
The Dirichlet form expressions, or variational principle, for $\lambda_Y,\lambda_Z$ give \cite[\S13.2]{levin2009book}
\begin{align}
\begin{aligned}
\lambda_Y&=\inf_{f\not\equiv \text{constant}}\frac12\frac{\sum_{u,v}|f(u)-f(v)|^2\phi(u)P_\sym(u,v)}{\operatorname{Var}_\phi(f)},\\
\lambda_Z&=\inf_{f\not\equiv \text{constant}}\frac12\frac{\sum_{u,v}|f(u)-f(v)|^2\mu(u)P_{\ind,\sym}(u,v)}{\operatorname{Var}_\mu(f)},
\end{aligned}
\end{align}
where we take $u,v\in\{1,\ldots,n\}^2$ to be 2-tuples, and $\operatorname{Var}_\nu:=\inf_c \|f-c\|^2_{\ell^2(\nu)}$.
Combined with \eqref{eqn:indep-compare} and $\frac{1}{2n^2}\le\phi\le\frac{2}{n^2}$, 
this shows
\begin{align}
\lambda_Y\ge \frac14\left(\frac23\right)^M\lambda_Z.
\end{align}
From \eqref{eqn:p1}, this gives
\begin{align}\label{eqn:p2}
\max_y\|P^t(y,\cdot)-\phi\|_\tv&\le\frac12 \sqrt{n(n+1)}\left(1-\frac129^{-M}\lambda_Z\right)^{t/2}.
\end{align}

\item Let $P^{\text{1-particle}}_\sym$ be the transition matrix for the additive symmetrization of the 1-particle circuit walk (again with a single time-step corresponding to $M$ circuit layers), and let $\lambda_1$ be the first nonzero eigenvalue of $I-P^{\text{1-particle}}_\sym$. We use a Dirichlet form comparison again. Let $Q$ be the transition for the 1-particle walk, so the transition matrix for $Z$ is $P_{\ind,\sym}=\frac{Q\otimes Q+Q^T\otimes Q^T}{2}$. Then for $F:\{1,\ldots,n\}^2\to\C$, the Dirichlet form for the independent 2-particle walk $\Zt_t$ is
\begin{align*}
D(F):=\frac12\frac{\sum_{x,y,w,z}|F(x,y)-F(w,z)|^2Q(x,w)Q(y,z)}{\sum_{x,y}|F(x,y)|^2}.
\end{align*}
Recall that $Q(x,x)\ge2^{-M}$ for every $x$, and that by the Dirichlet principle, we have $\lambda_1=\inf_{f\not\equiv\text{constant}}\frac12\frac{\sum_{x,w}|f(x)-f(w)|^2\frac1nQ(x,w)}{\Var(f)}$ for $\Var(f)=\frac1n\sum_x|f(x)-\bar f|^2$ with $\bar f=\frac1n\sum_zf(z)$. Then for $F$ mean zero, and considering only the terms $z=y$ and $w=x$, and applying the Dirichlet principle for $\lambda_1$ to $F(\cdot,y)$ and $F(x,\cdot)$, we have
\begin{align*}
D(F)&\ge\frac12\frac{2^{-M}\left[\sum_{x,y,w}|F(x,y)-F(w,y)|^2Q(x,w)+\sum_{x,y,z}|F(x,y)-F(x,z)|^2Q(y,z)\right]}{\sum_{x,y}|F(x,y)|^2}\\
&\ge 2^{-M}\lambda_1n\frac{\sum_y \Var(F(\cdot,y))+\sum_x\Var(F(x,\cdot))}{\sum_{x,y}|F(x,y)|^2}\\
&=2^{-M}\lambda_1\frac{\E_y\Var(F(\cdot,y))+\E_x\Var(F(x,\cdot))}{\Var(F)},\numberthis\label{eqn:dfvar} 
\end{align*}
using that $\sum_{x,y}F(x,y)=0$ to obtain the variance in the denominator.
Tensorization of variance (see e.g. \cite[\S3.1]{boucheron2013concentration}, or alternatively one can directly estimate for this case), gives
\begin{align}
\Var F(X,Y)&\le \E\Var_X(F|Y)+ \E\Var_Y(F|X),
\end{align}
for $X,Y$ independent. In this setting they are both uniform on $\{1,\ldots,n\}$.
Using the above in \eqref{eqn:dfvar} gives
\begin{align}
\lambda_Z\ge 2^{-M}\lambda_1.
\end{align}

\item Let $\Delta_n$ be the spectral gap of the simple random walk on the underlying graph, i.e. $\Delta_n$ is the first nonzero eigenvalue of $1-P_\srw$, where $P_\srw$ is the Markov transition matrix for the simple random walk. We apply the path-comparison method as in Corollary~\ref{cor:mixsrw}, which gives
\begin{align}
\lambda_1\ge 2^{-M}\Delta_n.
\end{align}
Thus
\begin{align}\label{eqn:markov-tvd}
\max_y\|P^t(y,\cdot)-\phi\|_\tv&\le \frac12\sqrt{n(n+1)}\left(1-\frac12 36^{-M}\Delta_n\right)^{t/2},
\end{align}
which is $\le\varepsilon$ if 
\begin{align}\label{eqn:t-bound}
t\ge \frac{4\cdot 36^M\log\frac{\sqrt{n(n+1)}}{2\varepsilon}}{\Delta_n}=:S(\varepsilon).
\end{align}
Note we used $\log(1-x)\le - x$ for $0\le x<1$.
\end{enumerate}

Finally, returning to \eqref{eqn:s2-lowerbound-walk} and the walk $Y_t$, where a single time-step corresponds to a single layer in the circuit, consider depths $T\ge M\cdot \lceil S(o(1/n^2))\rceil$ for some fixed $o(1/n^2)$ function. For convenience and concreteness we can take $\varepsilon=\varepsilon_n=\frac{1}{n^2\sqrt{n(n+1)}}=O(\frac{1}{n^3})$.
Since applying Markov transition matrices with the same stationary distribution never increases the distance to stationarity, we can apply the bound \eqref{eqn:markov-tvd} even when $T$ is not a multiple of $M$.
Recalling $\phi_{(x,y)}=\frac{1}{n(n+1)}$ for $x\ne y$, we have $\P_{(z,z)}[\Yt_t=(x,y)]=\phi_{(x,y)}+O(\varepsilon_n)$, and so using \eqref{eqn:t-bound} we obtain
\begin{align*}
\frac12\P_{(z,z)}[Y_T\text{ split across $\Gamma|\Gamma^c$}]&=\P_{(z,z)}[Y_T^1\in\Gamma,\,Y_T^2\not\in\Gamma]\\
&=\frac{|\Gamma|(n-|\Gamma|)}{n(n+1)}\left(1+O(n^2\varepsilon_n)\right)=\frac{|\Gamma|(n-|\Gamma|)}{n^2}(1+O(1/n)).\numberthis
\end{align*}
Using this in \eqref{eqn:s2-lowerbound-walk} gives the theorem.
\end{proof}

\section{Brickwork lower bound}\label{sec:brick-lower}

In this section, we prove the lower bound $\E S_2(U)\ge A_\Gamma\Omega_D(\sqrt{d})\log\cosh(2s)$ for the random brickwork circuits.
Recall that a depth-$d$ brickwork circuit in $D$ dimensions consists of $d$ complete cycles, with each cycle consisting of a fixed ordering of the $2D$ individual layers. For convenience, we assume in each dimension that the $U^{L,j}$ layer is applied at some point before the $U^{R,j}$ layer. The indexing for walks $Z_t$ in this section is $t\in\frac{1}{2D}\N$.
\begin{thm}[brickwork lower bound]\label{thm:brickwall-lowerbound}
Let $U$ be a depth-$d$ random one-dimensional brickwork linear optical unitary, and consider the subsystem $\{1,\ldots,k\}$ consisting of the first $k\le n/2$ modes. Then there is a constant $c_1$ so that for depths $d\le c_1k^2$,
\begin{align}\label{eqn:lower-1d}
\E S_2(U)&\ge \Omega\big(\sqrt{d}\big)\log\cosh(2s).
\end{align}
More generally, consider the depth-$d$ random $D$-dimensional brickwork unitary on $n=m^D$ modes, and the subsystem $\Gamma=\{(x_1,\ldots,x_D):1\le x_j\le k_j\}$. 
Let $A_\Gamma:=\sum_{j:k_j<m}\prod_{\ell:\ell\ne j}k_\ell$ be the size of the boundary region of $\Gamma$.
There is a constant $c_2=c_2(D)$ so that for depths $d\le c_2\min_{j:k_j<m}\{k_j^2,(m-k_j)^2\}$,
\begin{align}\label{eqn:lower-dim}
\E S_2(U)&\ge A_\Gamma\Omega_D\big(\sqrt{d}\big)\log\cosh(2s).
\end{align}
\end{thm}
\begin{rmk}
The entropy $S_2(U)$ for a size $k$ subsystem is always bounded above by $k\log\cosh(2s)$, so the condition $d=O(k^2)$ for the one-dimensional bound \eqref{eqn:lower-1d} extends up to the order of time for maximal-order saturation.
\end{rmk}

\begin{proof}[Proof of Theorem~\ref{thm:brickwall-lowerbound}]
We will apply \eqref{eqn:s2-lowerbound-walk} from Theorem~\ref{thm:uut}, which means we need to estimate $\sum_{z=1}^n\P_{(z,z)}[Y_d\text{ split across }\Gamma|\Gamma^c]$. 
We start with the 1-dimensional walk bound, and
first consider the 1-dimensional brickwall walk on $\Z$ (rather than on $\{1,\ldots,n\}$).
Let $Z_t$ be a one-particle one-dimensional brickwall walk on $\Z$ with time-steps $t\in\frac12\N$, and let $\chi(x):=\begin{cases}1/2,&x\text{ odd}\\-1/2,&x\text{ even}\end{cases}$.
For $x\in\Z$, let $F(x):=x+\chi(x)$.
The point is that $F(Z_t)$ has simpler increments than $Z_t$, in particular they do not depend on the parity of $F(Z_t)$: considering the cases for the parities of $Z_t$ and $Z_{t+1}$, we see $F(Z_{t+1})-F(Z_t)\in\{-2,0,0,2\}$, with each outcome having probability $1/4$.
Restricting to $t\in\N$, so each increment $t+1$ corresponds to a complete $2$-layer cycle,  
\begin{align}
\E[F(Z_{t+1})|Z_t]&=F(Z_t),
\end{align}
so $F(Z_t)$, $t\in\N$, is a martingale.
Now consider the 2-particle walk $Y_t=(Y_t^1,Y_t^2)$ described by \eqref{eqn:walk2} on $\Z$.  Checking with \eqref{eqn:walk2} shows that that the marginals of $Y_t^1$ and $Y_t^2$ are both the one-particle brickwall walk on $\Z$.
Then $M_t:=F(Y_t^1)-F(Y_t^2)$, $t\in\N$, is also a martingale. We will be able to lower bound $\E|M_t|$ using martingale differences and the Paley--Zygmund inequality.

To apply Paley--Zygmund to $M_d^2$, we want a lower bound on $\E M_d^2$ and an upper bound on $\E M_d^4$. Let $\mathcal F_t$ be the $\sigma$-algebra generated by $Y_s$ up to time $t$.
Since $M_0=0$, from the martingale definition $\E[M_t|\mathcal F_{t-1}]=M_{t-1}$ we have
\begin{align}\label{eqn:m-diff}
\E M_d^2&=\sum_{t=0}^{d-1}\E[(M_{t+1}-M_t)^2].
\end{align}
We can calculate $\E[(M_{t+1}-M_t)^2|\mathcal F_t]$ exactly as a function of the relationship between $Y_t^1$ and $Y_t^2$, but since we do not care about optimizing constants, it is enough to show that $\E[(M_{t+1}-M_t)^2|\mathcal F_t]\ge c$ for an absolute constant $c$.
Recalling the transition through a beamsplitter is given by \eqref{eqn:walk2}, we can simply pick one outcome with $(M_{t+1}-M_t)^2=\left[F(Y_{t+1}^1)-F(Y_t^1)-(F(Y_{t+1}^2)-F(Y_t^2))\right]^2\ge4$ and estimate
\begin{align}
\E[(M_{t+1}-M_t)^2|\mathcal F_t]&\ge\begin{cases}
\frac49,&\{Y_t^1,Y_t^2\}\in\pi_L\\ 
\frac12 ,&\{Y_t^1,Y_t^2\}\not\in\pi_L 
\end{cases}.
\end{align}
For example, in the first case we picked the path where both walks move together as far away from the starting point, and in the second case we picked the path where one walk moves away from the other starting point in the first layer of the step.
Thus
\begin{align}
\E M_d^2&\ge cd.
\end{align}
Also, letting $\Delta_t:=M_{t+1}-M_t$ and noting $|\Delta_t|=|M_{t+1}-M_t|\le4$, we have $\E M_t^2=O(t)$ from \eqref{eqn:m-diff}, and
\begin{align*}
\E[M_{t+1}^4]=\E[(M_t+\Delta_t)^4]&=\E[\E[M_t^4+4M_t^3\Delta_t+6M_t^2\Delta_t^2+4M_t\Delta_t^3+\Delta_t^4|M_t]]\\
&\le \E[M_t^4]+6\cdot 4^2\E[M_t^2]+4^4\E[|M_t|]+4^4\\
&\le \E[M_t^4]+Ct+C\le C\sum_{t'=0}^t(t'+1)\le C(t+1)^2.\numberthis
\end{align*}
Applying the Paley--Zygmund inequality to $M_d^2$, which has $\E M_d^2=\Theta(d)$, gives for some small $\tilde c$,
\begin{align}
\P[|M_d|>\tilde c\sqrt{d}]&\ge c'\frac{(\E M_d^2)^2}{\E M_d^4}\ge c''.
\end{align}
Using the definitions of $F$ and $M_d$, we see $|Y_d^1-Y_d^2|\ge |M_d|-1$, and so from Markov's inequality we also obtain for any starting point $(z,z)$,
\begin{align}\label{eqn:ydiff}
\E_{(z,z)}|Y_d^1-Y_d^2|&\ge c\sqrt{d}.
\end{align}

Fix a cut $k+1/2$ on $\Z$. Consider the displacements of the 1-particle walks $Y_d^1$ and $Y_d^2$ from their starting point $z$, $\delta_d^1:=Y_d^1-z$ and $\delta_d^2:=Y_d^2-z$. We can consider the displacements as a random walk depending only on the parity of $z$, not on the actual value of $z$.
Given displacements $\delta_d^1$ and $\delta_d^2$, them being far apart means there are many possible $z\in\Z$ such that $Y_d$ is split across the cut $k+1/2$. More precisely, for either parity,
\begin{align*}
\sum_{z\in\Z}\oneb_{z+\delta_d^1,z+\delta_d^2\text{ split across $k+1/2$}}[z]
&=\#\{z\in\Z:z+\delta_d^1,z+\delta_d^2\text{ split across $k+1/2$}\}\\
&=|\delta_d^1-\delta_d^2|.\numberthis
\end{align*}
Since we will consider the distribution of $(\delta_d^1,\delta_d^2)$, we separate the above by parity of $z$, giving for example for $z\in2\Z$,
\begin{align}
\sum_{z\in2\Z}\oneb_{z+\delta_d^1,z+\delta_d^2\text{ split across $k+1/2$}}[z]
&\ge\frac{|\delta_d^1-\delta_d^2|-1}{2}.\numberthis
\end{align}
The same lower bound holds for the sum over odd $z$.
Thus using \eqref{eqn:ydiff}, which holds for either parity of $z$, and that $|\delta_d^1-\delta_d^2|=|Y_d^1-Y_d^2|$, we see
\begin{align}\label{eqn:z-lower}
\sum_{z\in\Z}\P_{(z,z)}[Y_d\text{ split across cut $k+1/2$}]&\ge \min_z\E_{(z,z)}|Y_d^1-Y_d^2|-1\ge c\sqrt{d}.
\end{align}

It remains to show a similar bound holds for the walk on the finite interval $\{1,\ldots,n\}$.
The intuitive observation is that since the walk $Y_d^i$ on $\Z$ is doing a 1D random walk, only the starting modes $z$ within around $O(\sqrt{d})$ of the cut $k+1/2$ should be able to contribute to \eqref{eqn:z-lower}. Using the large deviation bound for this walk on $\Z$ in \eqref{eqn:split-indep} and \eqref{eqn:rwprob2}, which implies $\P_z[|Y_d^i-z|\ge r]\le Ce^{-cr^2/d}$, we get
\begin{align*}
\sum_{z:|z-k|>C_1\sqrt{d}}\P_{(z,z)}[Y_d\text{ split across cut $k+1/2$}]&\le \sum_{z:|z-k|>C_1\sqrt{d}}\P_z[\max_i|Y_d^i-z|\ge|z-k|]\\
&\le C\sum_{j=C_1\sqrt{d}}^\infty e^{-cj^2/d}\le \eta(C_1)\sqrt{d},\numberthis
\end{align*}
where $\eta(C_1)\to0$ as $C_1\to\infty$.
Thus choosing $C_1$ sufficiently large so that $\eta(C_1)<c$ from \eqref{eqn:z-lower}, we get
\begin{align*}
\sum_{z\in\Z:|z-k|\le C_1\sqrt{d}}\P_{(z,z)}[Y_d\text{ split across cut $k+1/2$}]&\ge \min_z\E_{(z,z)}|Y_d^1-Y_d^2|-1-\eta(C_1)\sqrt{d}\\
&\ge c\sqrt{d}-\eta(C_1)\sqrt{d}\ge c'\sqrt{d}.\numberthis\label{eqn:z-lower2}
\end{align*}

To restrict to the walk on the finite interval $\{1,\ldots,n\}$, note that the condition $d\le c_1k^2$ means that for sufficiently small $c_1$, that $|z-k|\le C_1\sqrt{d}\le C_1\sqrt{c_1}k$ implies such $z$ are at least $\Omega(k)$ from the boundaries. Couple the walk on $\{1,\ldots,n\}$ with the walk on $\Z$ until they reach a boundary 1 or $n$. Using Doob's martingale inequality (see e.g. \cite[\S4.4]{durrett-book}) for $\tilde F_t:=F(Z_t)-F(Z_0)$, and using that $\tilde F_t$ has bounded increments so $\E \tilde F_d^2\le Cd$ by \eqref{eqn:m-diff}, we get for the walk on $\Z$ (and also for the walk on $\{1,\ldots,n\}$ by coupling),
\begin{align}\label{eqn:doob}
\P_z[\max_{0\le t\le d}|Y_t^i-z|>\Omega(k)]&\le \frac{(\E|\tilde F_d|^2)^{1/2}}{\Omega(k)}\le \frac{C'\sqrt{d}}{k}\le C'\sqrt{c_1},
\end{align}
where we use the condition $d\le c_1k^2$ for the last inequality.
For sufficiently small $c_1$, subtracting this from every probability in \eqref{eqn:z-lower2} gives the bound for the walk $Y_t$ on the finite interval $\{1,\ldots,n\}$,
\begin{align}
\sum_{z=1}^n\P_{(z,z)}[Y_d\text{ split across cut $k$}]&\ge c''\sqrt{d},
\end{align}
which proves \eqref{eqn:lower-1d}.

\vspace{2mm}
To prove the $D$-dimensional bound \eqref{eqn:lower-dim}, we pick a coordinate $j$ with $k_j<m$, and estimate the probability that $Y_d^1,Y_d^2\in\{1,\ldots,m\}^D$ are split across $k_j$ in the $j$th coordinate, while all other coordinates satisfy $(Y_d^i)_\ell\le k_\ell$. By the same proof as in the 1D case, using the condition $d\le c_2\min_{j:k_j<m}\{k_j^2,(m-k_j)^2\}$, we see
\begin{align}\label{eqn:dim-1}
\sum_{z_j:|z_j-k_j|\le C_1\sqrt{d}}\P_{(z,z)}[(Y_d)_j\text{ is split across $k_j$}]&\ge c\sqrt{d}.
\end{align}
Note that the projected walk $(Y_t)_j\in\{1,\ldots,m\}^2$ is not exactly the same as the 1D 2-particle walk, since $(Y_t^1)_j$ and $(Y_t^2)_j$ only interfere in the same beamsplitter if all their other coordinates are also the same; otherwise, they evolve independently, but the same proof still applies.

Let $G_j$ be the event that $(Y_d^1)_\ell,(Y_d^2)_\ell\le k_\ell$ for all $\ell\ne j$. Since we take lower bounds, we are free to restrict to starting coordinates $z_\ell$ for $\ell\ne j$ with $k_\ell/4\le z_\ell\le k_\ell/2$. 
Similarly as in the 1D bounds, by taking $c_2=c_2(D)$ sufficiently small in $d\le c_2\min k_\ell^2$, we see
\begin{align}
\P[G_j^c]&\le 2(D-1)\P[\max_{0\le t\le d}|(Y_t^i)_\ell-z_\ell|>k_\ell/4]\le \frac{C\sqrt{Dd}}{k_\ell}\le C\sqrt{Dc_2}\le c'.
\end{align}
Subtracting this from every probability in \eqref{eqn:dim-1} shows that if $k_\ell/4\le z_\ell\le k_\ell/2$ for $\ell\ne j$, then
\begin{align}
\sum_{z_j:|z_j-k_j|\le C_1\sqrt{d}}\P_{(z,z)}[Y_d\text{ split across $\Gamma|\Gamma^c$, $(Y_d)_j$ across $k_j$}]\ge c''\sqrt{d}.
\end{align}
Summing over the rest of the coordinates $z_\ell$ of $z$ with the restriction $k_\ell/4\le z_\ell\le k_\ell/2$ gives
\begin{align*}
\sum_{z\in\{1,\ldots,m\}^D}\P_{(z,z)}&[Y_d\text{ split across $\Gamma|\Gamma^c$}]\\
&\ge\sum_{z\in\{1,\ldots,m\}^D}\frac{1}{D}\sum_{j:k_j<m}\P_{(z,z)}[Y_d\text{ split across $\Gamma|\Gamma^c$, $(Y_d)_j$ split across $k_j$}]\\
&\ge c_D\sum_{j:k_j<m}\sum_{\substack{z\in\{1,\ldots,m\}^D:\\|z_j-k_j|\le C_1\sqrt{d}\\ k_\ell/4\le z_\ell\le k_\ell/2}}\P_{(z,z)}[Y_d\text{ split across $\Gamma|\Gamma^c$, $(Y_d)_j$ split across $k_j$}]\\
&\ge c_D \sum_{j:k_j<m}c\sqrt{d}\prod_{\ell\ne j}k_\ell=c_D' A_\Gamma\sqrt{d},\numberthis
\end{align*}
finishing the proof of \eqref{eqn:lower-dim}.
\end{proof}

\section{Circuit depth for robust circuit complexity}\label{sec:rcc}

In this section, we continue from Section~\ref{sec:complexity}, and show how to adapt the construction of the approximate linear optical unitary $\tilde U$ so that it has shorter depth.

For depth-$d$ random 1D brickwall unitary $U$, we want to construct an approximate unitary $\tilde U$ which has depth $O(\sqrt{d\log d})$. Previously, we used the effective band width $O(\sqrt{d\log d})$ of $U$ to construct $\tilde U$ with lower gate count $O(n\sqrt{d\log d})$ via \cite{hurwitz1897uber,reck1994experimental}, but the gates were partially sequential, leading to $\Theta(n)$ depth of $\tilde U$. To reduce the depth, we will first approximately block-diagonalize $U$ into blocks of size $O(\sqrt{d\log d})$, which can be done in parallel in depth $O(\sqrt{d\log d})$. Then we can apply the method as before, using \cite{hurwitz1897uber,reck1994experimental}, in parallel to each block, which only takes additional depth $O(\sqrt{d\log d})$.

\begin{thm}\label{thm:cc-depth}
Let $U$ be an $n\times n$ depth-$d$ random one-dimensional brickwall linear optical unitary; this utilizes $\Theta(nd)$ nearest-neighbor gates. Choose a $\kappa\ge1$. With probability $1-O(n^2d^{-\kappa})$ over $U$, there is a depth-$O_\kappa(\sqrt{d\log d})$ unitary $\tilde U$ built from $O_\kappa(n\sqrt{d\log d})$ nearest-neighbor gates such that
\begin{align}
\|U-\tilde U\|_\hs=O(n^3d^{-\kappa}).
\end{align}
The fidelity bound follows similarly as in Theorem~\ref{thm:complexity} with a slightly higher exponent on $n$.
\end{thm}

To approximately block-diagonalize $U$, we use a cut-straightening lemma, which takes an approximate band matrix and splits it into two approximate blocks.
To this end, let $\C^n=A\oplus B_-\oplus B_+\oplus C$ be an orthogonal decomposition, and define $L:=A\oplus B_-$, $R:=B_+\oplus C$, and $B:=B_-\oplus B_+$. For splitting a band of width $w$ between nodes $b$ and $b+1$ into two blocks (Figure~\ref{fig:cut}), we take $A=\op{span}\{1,\ldots,b-w\}$, $B_-=\op{span}\{b-w+1,\ldots,b\}$, $B_+=\op{span}\{b+1,\ldots,b+w\}$, and $C=\op{span}\{b+w+1,\ldots,n\}$.

\begin{lem}[cut-straightening]\label{lem:cut}
Let $\C^n=A\oplus B_-\oplus B_+\oplus C$ be an orthogonal decomposition as above with $\dim B_-=\dim B_+=w$, and set $L:=A\oplus B_-$ and $R:=B_+\oplus C$. Let $U\in\U(n)$, and let $P_X$ denote the orthogonal projection onto a subspace $X\subseteq\C^n$. Define
\begin{align}\label{eqn:delta}
\delta^2:=\|P_RUP_A\|_\hs^2+\|P_LUP_C\|_\hs^2.
\end{align}
Then there is a unitary $G\in\U(n)$ of the form $G=I_A\oplus G_B\oplus I_C$ such that
\begin{align}\label{eqn:delta2}
\|P_RUGP_L\|_\hs^2+\|P_LUGP_R\|_\hs^2&\le 6\delta^2.
\end{align}
\end{lem}

\begin{figure}[htb]
\begin{tikzpicture}
\def\psize{1.7cm}
\def\msize{1.6cm}
%structure inside
\def\w{.4cm} % band width * 2
\draw[fill=gray!20] (-\msize,\msize-\w)--(-\msize,\msize)--(-\msize+\w,\msize)--(\msize,-\msize+\w)--(\msize,-\msize)--(\msize-\w,-\msize)--cycle;
\draw[<->,xshift=.6cm,yshift=-.6cm] (-\w,0)--(\w,0);
\node[above] at (.5,-.65) {$2w+1$};
% block division
\draw (-1.5,.5)--(1.5,.5);
\draw (-.5,-1.35)--(-.5,1.5);
% square
\fill[color=blue!30,opacity=.75] (-.5cm-\w,.5cm-\w) rectangle ++(2*\w,2*\w);
% left and right parenthesis
\node at (-\psize,0) {$\left(\vphantom{\rule{\psize}{\psize}}\right.$};
\node at (\psize,0) {$\left.\vphantom{\rule{\psize}{\psize}}\right)$};
% labels
\node[left] at (-2,1.2) {$A$};
\node[left] at (-2,.65) {$B_-$};
\node[left] at (-2,.3) {$B_+$};
\node[left] at (-2,-.6) {$C$};
\draw (-2,-\msize)--(-2,\msize);
\foreach \t in {.1,.5,.9}{
\draw (-2.1,\t)--++(.2,0);
}
% L and R
\draw[decoration={brace,raise=3pt,aspect=.5,amplitude=4pt},decorate,xshift=-1mm] (-2.5,-1.6)--(-2.5,.48);
\node[left] at (-2.8,-.5) {$R$};
\draw[decoration={brace,raise=3pt,aspect=.5,amplitude=4pt},decorate,xshift=-1mm] (-2.5,.52)--(-2.5,1.6);
\node[left] at (-2.8,1.1) {$L$};
% label areas
\fill[color=red!40,opacity=.7] (-.5cm+\w,.5cm) rectangle (\msize,\msize);
\node[above] at (.7,.9) {$P_LUP_C$};
\fill[color=red!40,opacity=.7] (-.5cm-\w,.5cm) rectangle (-\msize,-\msize);
\node[above] at (-1.15,-1) {$P_RUP_A$};

\node[above] at (2.5,.3) {$G$};
\node[above] at (2.5,0) {$\longrightarrow$};
% right picture
\begin{scope}[xshift=5cm]
\draw[fill=gray!20] (-\msize,\msize-\w)--(-\msize,\msize)--(-\msize+\w,\msize)--(\msize,-\msize+\w)--(\msize,-\msize)--(\msize-\w,-\msize)--cycle;
\draw[<->,xshift=.6cm,yshift=-.6cm] (-\w,0)--(\w,0);
\node[above] at (.5,-.65) {$2w+1$};
% block division
\draw (-1.5,.5)--(1.5,.5);
\draw (-.5,-1.35)--(-.5,1.5);
% square
\fill[color=blue!30,opacity=.75] (-.5cm-\w,.5cm-\w) rectangle ++(2*\w,2*\w);
% left and right parenthesis
\node at (-\psize,0) {$\left(\vphantom{\rule{\psize}{\psize}}\right.$};
\node at (\psize,0) {$\left.\vphantom{\rule{\psize}{\psize}}\right)$};
% label areas
\fill[color=red!40,opacity=.7] (-.5cm,.5cm) rectangle (\msize,\msize);
\node[above] at (.6,.8) {$P_LUGP_R$};
\fill[color=red!40,opacity=.7] (-.5cm,.5cm) rectangle (-\msize,-\msize);
\node[above] at (-1.1,-1) {$P_RUGP_L$};
% band areas
\draw[xshift=\w,yshift=-\w] (-.5cm,.5cm)--++(0,-\w)--++(-\w,0);
\draw[xshift=-\w,yshift=\w] (-.5cm,.5cm)--++(0,\w)--++(\w,0);
% striped triangles
\fill[pattern=north east lines,yshift=\w] (-.5cm,.5cm)--++(0,\w)--++(-\w,0)--cycle;
\fill[pattern=north east lines,yshift=-\w] (-.5cm,.5cm)--++(0,-\w)--++(\w,0)--cycle;
\end{scope}
\end{tikzpicture}
\caption{Applying the cut-straightening lemma for a matrix with effective band width $w$. The red shaded regions show the regions \eqref{eqn:delta} (left) and \eqref{eqn:delta2} (right), which have small norm. Since the matrix $G$ acts only on the columns corresponding to $B$, it can only move $\ell^2$ mass from the effective band into the two triangular hatched regions. The resulting matrix $UG$ thus has an approximate block structure, while also keeping the effective band structure combined with the additional hatched triangular regions.}\label{fig:cut}
\end{figure}

We first explain the cut-straightening lemma in the case of an exact band of width $w$, in which case we take sets $A,B_-,B_+,C$ as in Figure~\ref{fig:cut}, and have $\delta^2=0$. To compare to the effective band case, we phrase this in terms which will also be useful there. 
We want the new unitary $UG$ to send $L$ to itself, or equivalently for $G$ to send $L$ to $U^{-1}L$. Let $P:=U^{-1}P_LU$ which is the orthogonal projection onto $U^{-1}L$. 
Since $U$ is a band matrix with exact band width $\le w$, we know $U$ sends $A$ to $UA\subset L$, and $C$ into $R$. $U$ can send parts of both $B_-$ and $B_+$ into $L$, so we need to ``untangle'' where $B_-$ and $B_+$ go. 
On vectors in $A$, $P$ simply acts as the identity (or as $P_A$) since $A\subset U^{-1}L$.
Letting $K:=P_BPP_B$, then we can write $P=P_A\oplus K\oplus0_C$.
$K$ is a projection with rank $\rank P-\rank P_A=\dim L-\dim A=w$, so it has $w$ eigenvalues which are 1 and $w$ which are 0. Let $X=\op{Range}(K)$, so $K=P_X$. Note that $X=B\cap U^{-1}L$, which is the orthogonal complement of $A$ in $U^{-1}L$. 
Since $\dim B_-=\dim X=w$, we simply construct $G_B\in\U(B)$ as a $2w\times 2w$ unitary which takes an orthonormal basis of $B_-$ to one of $X$. $G_B$ then takes $B_+$ to $Y$ the orthogonal complement of $X$ in $B$. The matrix $G$ is formed from $G_B$ by extending it as the identity outside of $B$. We see that $UG$ takes $B_-$ to $UX\subset L$, as desired.

\begin{proof}[Proof of Lemma~\ref{lem:cut}]
As in the above explanation, let $P=U^{-1}P_LU$ be the orthogonal projection onto $U^{-1}L$, and let $K:=P_BPP_B$. While $K$ may not be an exact projection, we will show it is still close to a rank $w$ projection. Let the eigenvalues of $K$ be $1\ge\lambda_1\ge\cdots\ge\lambda_{2w}\ge0$, with corresponding eigenvectors $(\varphi_i)_i$. Let $X:=\operatorname{span}(\varphi_1,\ldots,\varphi_w)$, and let $Y:=\operatorname{span}(\varphi_{w+1},\ldots,\varphi_{2w})$, the orthogonal complement of $X$ in $B$.
We take $G_B\in\U(B)$ so that $G_BB_-=X$ and $G_BB_+=Y$, and then $G:=I_A\oplus G_B\oplus I_C$. 

Next, we need to prove the estimate \eqref{eqn:delta2}, which we will do by showing $K$ is close enough to a rank $w$ projection, in particular that $\Tr(K-K^2)$ is small and $\Tr K$ is close to $w$.  Recalling $P=U^\dagger P_LU$, let
\begin{align*}
\eta_A:=\|P_RUP_A\|_\hs^2=\Tr(P_A(1-P)),\quad\text{and}\quad \eta_C:=\|P_LUP_C\|_\hs^2=\Tr(P_CP).
\end{align*}
Since one can check $GP_{B_-}=P_XG$, we have $GP_L=(P_A+P_X)G$, and similarly $GP_R=(P_Y+P_C)G$. Then
\begin{align*}
\|P_RUGP_L\|_\hs^2&=\Tr(P_LG^\dagger U^\dagger P_RUG)=\Tr((P_A+P_X)U^\dagger (I-P_L)U)=\Tr((P_A+P_X)(1-P)),\\
\|P_LUGP_R\|_\hs^2&=\Tr(P_RG^\dagger U^\dagger P_LUG)=\Tr((P_Y+P_C)P)=\Tr(P_YP)+\eta_C.
\end{align*}
We thus get
\begin{align}\label{eqn:sum-bound}
\|P_RUGP_L\|_\hs^2+\|P_LUGP_R\|_\hs^2&=\delta^2+\Tr(P_X(I-P))+\Tr(P_YP).
\end{align}
We claim
\begin{align}\label{eqn:trace-bound}
\Tr(P_X(I-K))+\Tr(P_YK)\le 4\Tr(K-K^2)+|\Tr K-w|.
\end{align}
To show this, we will use that $\min(\lambda,1-\lambda)\le2\lambda(1-\lambda)$, so that $\sum_{i=1}^{2w}\min(\lambda_i,1-\lambda_i)\le 2\Tr(K-K^2)$.
Let $m:=\#\{i:\lambda_i\ge1/2\}$. Then recalling the eigenvalues are in decreasing order, we have
\begin{align*}
\sum_{i=1}^w(1-\lambda_i)+\sum_{i=w+1}^{2w}\lambda_i &\le \sum_{i=1}^{2w}\min(\lambda_i,1-\lambda_i)+|m-w|\\
&\le 2\Tr(K-K^2)+|m-\Tr K|+|\Tr K-w|.\numberthis
\end{align*}
We can estimate
\begin{align*}
|m-\Tr K|=\left|\sum_{i=1}^{2w}\oneb_{\lambda_i\ge1/2}-\lambda_i\right|&\le\sum_{i=1}^{2w}|\oneb_{\lambda_i\ge1/2}-\lambda_i|
=\sum_{i=1}^{2w}\min(\lambda_i,1-\lambda_i)\le 2\Tr(K-K^2),\numberthis
\end{align*}
which gives \eqref{eqn:trace-bound}.

It remains to show that $\Tr(K-K^2)$ and $|\Tr K-w|$ are small, specifically each $\le\delta^2$.
First,
\begin{align*}
\Tr K=\Tr(P_BP)&=\rank P-\Tr(P_AP)-\Tr(P_CP)\\
&=\rank P+\eta_A-\dim A-\eta_C=w+\eta_A-\eta_C,
\end{align*}
which shows $|\Tr K-w|\le \eta_A+\eta_C\le\delta^2$.

Note that for any $P\ge0$ and matrix $A$, we have $|\Tr(AP)|\le (\Tr P)\|A\|$. Then
\begin{align*}
\Tr(K-K^2)=\Tr(P_BP(I-P_B)PP_B)&=\Tr(P(P_A+P_C)PP_B)\\
&=\Tr(PP_APP_B)+\Tr(PP_CPP_B)\\
&\le \Tr(PP_AP(1-P_A))+\Tr(P_CP)\\
&\le \Tr(P_A-P_APP_A)+\eta_C=\eta_A+\eta_C\le\delta^2.\numberthis
\end{align*}
Inserting these bounds with \eqref{eqn:trace-bound} into \eqref{eqn:sum-bound} gives \eqref{eqn:delta2}.
\end{proof}

\begin{proof}[Proof of Theorem~\ref{thm:cc-depth}]

For any $\kappa>0$, there is $C_\kappa$ so that with probability at least $1-Cn^2d^{-\kappa}$, the Hilbert-Schmidt norm of $U$ outside the effective $w=C_\kappa\sqrt{d\log d}$-band is at most $nd^{-\kappa}=:\varepsilon$. 
To construct an approximate $\tilde U$, we first use the cut-straightening Lemma~\ref{lem:cut} to transform $U$ to an approximate block matrix with blocks of size $O(w)$, as shown in Figure~\ref{fig:band-block}. 

\begin{figure}[htb]
\begin{tikzpicture}
\def\psize{1.7cm}
\def\msize{1.6cm}
%structure inside
\def\w{.8cm} % block size
\def\band{.4cm} % band size - half of block size
\begin{scope}[xshift=-5cm]
\draw[fill=gray!30] (-\msize,\msize-\band)--(-\msize,\msize)--(-\msize+\band,\msize)--(\msize,-\msize+\band)--(\msize,-\msize)--(\msize-\band,-\msize)--cycle;
\draw[<->] (-\band,0)--(\band,0);
\node[above] at (0,0) {$2w+1$};
% left and right parenthesis
\node at (-\psize,0) {$\left(\vphantom{\rule{\psize}{\psize}}\right.$};
\node at (\psize,0) {$\left.\vphantom{\rule{\psize}{\psize}}\right)$};
\end{scope}
\node[above] at (-2.5,0) {$\longrightarrow$};
\begin{scope}
\draw[fill=gray!30] (-\msize,\msize-\band)--(-\msize,\msize)--(-\msize+\band,\msize)--(\msize,-\msize+\band)--(\msize,-\msize)--(\msize-\band,-\msize)--cycle;
\foreach \b in {0.5,1.5,2.5}{
    % draw block outlines
    \draw (-\msize+\b*\w-.5*\w,\msize-\b*\w-.5*\w)--++(2*\w,0);
    \draw (-\msize+\b*\w+.5*\w,\msize-\b*\w-1.5*\w)--++(0,2*\w);
    % blue blocks
	\fill[fill=blue!30,opacity=.75,xshift=-\msize+\b*\w,yshift=\msize-\b*\w] (0,0)--(\w,0)--(\w,-\w)--(0,-\w)--cycle;
    
}
% left and right parenthesis
\node at (-\psize,0) {$\left(\vphantom{\rule{\psize}{\psize}}\right.$};
\node at (\psize,0) {$\left.\vphantom{\rule{\psize}{\psize}}\right)$};
\end{scope}
\node[above] at (2.5,0) {$\longrightarrow$};
\begin{scope}[xshift=5cm]
\foreach \b in {0,1,2,3}{
	\draw[fill=gray!30,xshift=-\msize+\b*\w,yshift=\msize-\b*\w] (0,0)--(\w,0)--(\w,-\w)--(0,-\w)--cycle;
}
% left and right parenthesis
\node at (-\psize,0) {$\left(\vphantom{\rule{\psize}{\psize}}\right.$};
\node at (\psize,0) {$\left.\vphantom{\rule{\psize}{\psize}}\right)$};
\end{scope}
\end{tikzpicture}
\caption{Effective band of width $w$ pictured. We cut the band at intervals of length $2w$, applying the cut-straightening lemma over each non-overlapping shaded blue square, producing an approximate block matrix with blocks size $O(w)$. Once we have an approximate block matrix with blocks size $O(\sqrt{d\log d})$, we can apply \cite{hurwitz1897uber,reck1994experimental} to each block in parallel, giving an approximating circuit with depth $O(\sqrt{d\log d})$.
}\label{fig:band-block}
\end{figure}

Apply Lemma~\ref{lem:cut} every $2w$ coordinates with $\dim B_-=\dim B_+=w$ to obtain unitaries $G_1,\ldots,G_m$. 
The cuts are spaced out so that the nontrivial part of each $G_i$ acts on separate columns (Figures~\ref{fig:cut} and \ref{fig:band-block}), so we can implement all $G_1,\ldots,G_m$ in a parallel circuit. 
Moreover, since the columns for each application are disjoint, 
applications of other $G_i$ only shuffle around the $O(\delta_i)$ 2-norm mass in the regions described by \eqref{eqn:delta} or \eqref{eqn:delta2}.

Each $G_i$ has only a $2w\times2w$ nontrivial block, so can be implemented in depth $O(w)$ by \cite{hurwitz1897uber,reck1994experimental,clements2016optimal}.
Moreover, by Lemma~\ref{lem:cut}, the end result $\tilde V:=UG_1\cdots G_m$ will be an approximate block matrix with blocks of size $O(w)$. 
By the above observation about disjoint columns for the $G_i$'s, we see the squared Hilbert--Schmidt norm outside the blocks in $\tilde V$ is at most
\begin{align}\label{eqn:hs-block}
\sum_{i=1}^m6\delta_i^2\le 6n\varepsilon^2=6n^3d^{-2\kappa}.
\end{align}
We then apply the procedure in the proof in Section~\ref{sec:complexity} (using \cite{hurwitz1897uber,reck1994experimental}) to each block of $\tilde V$ to get
\begin{align}
UG_1\cdots G_mT_{n,n-w}\cdots T_{2,1}\Phi=\tilde D,
\end{align}
with $\Phi$ a diagonal phaseshifter matrix chosen so the resulting unitary matrix $\tilde D$ has nonnegative diagonal.
The $T_{ij}$ matrices in separate blocks can be applied in parallel, so in total depth $O(w)$. 
The matrix $\tilde U$ can thus be implemented in depth $O(w)=O_\kappa(\sqrt{d\log d})$.

Finally, let $\tilde U:=\Phi^{-1}T_{2,1}^{-1}\cdots T_{n,n-w}^{-1}(G_1\cdots G_m)^{-1}$. By the argument of Section~\ref{sec:complexity} and using \eqref{eqn:hs-block}, we get
\begin{align}
\|\tilde U-U\|_\hs&=\|I-\tilde D\|_\hs=O(n^3d^{-\kappa}).
\end{align}
\end{proof}

\vspace{2mm}
\noindent
\textbf{Acknowledgments.}
L.S. thanks Sarah Miller and Zhi-Yuan Wei for discussions on boson sampling.
V.G. was sponsored by the Army Research Office under Grant Number W911NF-23-1-0241. J.T.I.~and A.V.G.~acknowledge support from the U.S.~Department of Energy, Office of Science, Accelerated Research in Quantum Computing, Fundamental Algorithmic Research toward Quantum Utility (FAR-Qu). J.T.I.~and A.V.G.~were also supported in part by the ARL (W911NF-24-2-0107), ONR MURI, DoE ASCR Quantum Testbed Pathfinder program (award No.~DE-SC0024220), NSF QLCI (award No.~OSI-2553574), NSF STAQ program, AFOSR MURI,  DARPA SAVaNT ADVENT, and NQVL:QSTD:Pilot:FTL. J.T.I.~and A.V.G.~also acknowledge support from the U.S.~Department of Energy, Office of Science, National Quantum Information Science Research Centers, Quantum Systems Accelerator (award No.~DE-SCL0000121). Y.-X.W.~acknowledges support from a QuICS Hartree Postdoctoral Fellowship. 

In arXiv v2, we added the results sequestered in Appendices~\ref{sec:lower2}--\ref{sec:rcc}, which involved use of GPT-5.5. 
Also, GPT-5.6 pointed out the simplification to remove a $\sqrt{\log d}$ factor in Proposition~\ref{prop:expbound}(ii)/Eq.~\eqref{eqn:exp-bound} of Theorem~\ref{thm:up}.

\bibliographystyle{amsalpha_edit}
\bibliography{entanglement.bib}

\end{document}